\documentclass[letterpaper,11pt,onecolumn,oneside]{article}

\usepackage[top=1in,bottom=1in,left=1in,right=1in]{geometry}
\usepackage{amsmath}
\usepackage{amsthm}
\usepackage{tikz}
\usetikzlibrary{calc,shapes,arrows.meta,backgrounds}
\usepackage{todonotes}
\setuptodonotes{inline}
\usepackage{paralist}
\usepackage{thm-restate}
\usepackage{framed}
\usepackage{xifthen}
\usepackage{tabularx}
\usepackage{nicefrac}
\usepackage{thm-restate}
\usepackage{subcaption}
\usepackage[colorlinks=true, allcolors=blue]{hyperref}
\usepackage[capitalize,noabbrev]{cleveref}

\newtheorem{theorem}{Theorem}
\newtheorem{proposition}{Proposition}
\newtheorem{lemma}{Lemma}
\newtheorem{observation}{Observation}
\newtheorem{corollary}{Corollary}

\tikzstyle{ver} = [circle, inner sep=3pt, draw]
\tikzstyle{te} = [circle, inner sep=3pt, draw, label=$t$, fill=brown!60!white]
\tikzstyle{por} = [circle, inner sep=3pt, draw, fill=green!30!white]
\tikzstyle{ter} = [circle, inner sep=3pt, draw,fill=blue!20!white]

\newcommand{\prd}{\textsc{Perfect Resilience Decision}}
\newcommand{\prs}{\textsc{Perfect Resilience Synthesis}}

\newcommand{\nice}{nice}

\newcommand{\rou}[3]{\ensuremath{\pi_{#1}(#3,\{#2\})}}
\newcommand{\rouns}[3]{\ensuremath{\pi_{#1}(#3,#2)}}
\newcommand{\roue}[2]{\ensuremath{\pi_{#1}(#2,\emptyset)}}
\newcommand{\ski}[2]{\ensuremath{\pi_{#1}(#2)}}

\newcommand{\Kt}{\ensuremath{K_{3,3} \setminus e}}
\newcommand{\Kf}{\ensuremath{K_5 \setminus e}}
\newcommand{\sK}{\text{subdivided}~\ensuremath{K_{2,4}}}
\newcommand{\Ktf}{\ensuremath{K_{3,4} \setminus 2e}}

\newlength{\RoundedBoxWidth}
\newsavebox{\GrayRoundedBox}
\newenvironment{GrayBox}[1]%
   {\setlength{\RoundedBoxWidth}{.93\textwidth}
    \def\boxheading{#1}
    \begin{lrbox}{\GrayRoundedBox}
       \begin{minipage}{\RoundedBoxWidth}}%
   {   \end{minipage}
    \end{lrbox}
    \begin{center}
    \begin{tikzpicture}%
       \node(Text)[draw=black!20,fill=white,rounded corners,%
             inner sep=2ex,text width=\RoundedBoxWidth]%
             {\usebox{\GrayRoundedBox}};
        \coordinate(x) at (current bounding box.north west);
        \node [draw=white,rectangle,inner sep=3pt,anchor=north west,fill=white] 
        at ($(x)+(6pt,.75em)$) {\boxheading};
    \end{tikzpicture}
    \end{center}}
    
\newenvironment{defproblemx}[3]{
    \noindent\ignorespaces%
    \FrameSep=6pt%
    \parindent=0pt%
    \vspace*{-1.5em}
    \begin{GrayBox}{\textsc{#1}}%
    \begin{tabular*}{\textwidth}{@{\hspace{.1em}} >{\itshape} p{#2} p{#3} @{}}}{
        \end{tabular*}%
        \end{GrayBox}%
        \ignorespacesafterend
    } 

\newcommand{\defproblem}[3]{
  \begin{defproblemx}{#1}{1.5cm}{.865\textwidth}
    Input:  & #2 \\
    Question: & #3
  \end{defproblemx}
}%

\newcommand{\taskproblem}[3]{
  \begin{defproblemx}{#1}{.9cm}{.9\textwidth}
        Input:  & #2 \\
        Task: & #3
  \end{defproblemx}
}%

\title{Perfect Network Resilience in Polynomial Time}

\author{Matthias Bentert \and Stefan Schmid}

\date{TU Berlin}

\begin{document}

\maketitle

\begin{abstract}
Modern communication networks support local fast rerouting mechanisms to quickly react to link failures: nodes store a set of conditional rerouting rules which define how to forward an incoming packet in case of incident link failures.    
Ideally, such rerouting mechanisms provide \emph{perfect resilience}: any packet is routed from its source~$s$ to its target~$t$ as long as~$s$ and~$t$ are still connected in the underlying graph after the link failures. However, ensuring perfect resilience
is algorithmically challenging as the rerouting decisions at any node~$v$ must rely solely on the \emph{local} information available at~$v$: the link from which a packet arrived at~$v$ (known as the \emph{in-port}), the target of the packet, and the \emph{incident} link failures at~$v$.
Already in their seminal paper at ACM PODC~'12, Feigenbaum, Godfrey, Panda, Schapira, Shenker, and Singla showed that there are instances in which perfect resilience cannot be achieved.
While the design of local rerouting algorithms has received much attention since then, we still lack a detailed understanding of when perfect resilience is achievable.

This paper closes this gap and presents a complete characterization of when perfect resilience can be achieved.
This characterization also allows us to design an~$O(n)$-time algorithm to decide whether a given instance is perfectly resilient and an~$O(nm)$-time algorithm to compute perfectly resilient rerouting rules whenever it is.
Our algorithm is also attractive for the simple structure of the rerouting rules it uses, known as \emph{skipping} in the literature: alternative links are chosen according to an ordered priority list (per in-port), where failed links are simply skipped.
This is also naturally supported in hardware.
The size of such an encoding is in~$\Theta(nm)$ and therefore the running time of our algorithm is optimal when considering skipping rerouting rules.
Intriguingly, our result also implies that in the context of perfect resilience, skipping rerouting rules are as powerful as more general rerouting rules that define the out-port for each set of incident failed links explicitly. This partially answers a long-standing open question by Chiesa, Nikolaevskiy, Mitrovic, Gurtov, Madry, Schapira, and Shenker~[IEEE/ACM Transactions on Networking, 2017] in the affirmative.

While our algorithm is simple, its analysis is intricate. 
A key concept in the analysis are links whose two endpoints also form a node separator.
We prove that removing those links does not change whether a given instance is perfectly resilient or not.
We also show that once all such links are removed, any instance either contains one of four specific rooted minors or belongs to one of three classes. If one of the four rooted minors is contained, then we are dealing with a no-instance (this was previously known for only two of them). Lastly, we show that any instance in any of the three remaining classes is a yes-instance, completing the characterization of perfectly resilient graphs.
We do this by showing that simply following a particular face of a planar embedding of the reduced instance using the right-hand rule until a link directly to the target is found is sufficient.
\end{abstract}

\newpage
\tableofcontents
\newpage

\section{Introduction}

To provide a high dependability, modern communication networks rely on fully decentralized packet rerouting mechanisms
which allow nodes (routers) to respond to link failures \emph{locally} and hence orders of magnitudes faster than traditional approaches. This reduces packet losses and improves throughput. 
Rather than requiring the communication of 
failure information (via the network's control plane), 
these local rerouting mechanisms allow to predefine
conditional failover rules at each node~$v$. 
These rules are static and can depend only on local information 
at a node $v$: they can be conditioned on the incoming port from which a packet
arrives at $v$ (the so-called \emph{in-port}), the target of the packet, the status of links incident to $v$, but \emph{not} on failures in other parts of the network. 

While such local rerouting mechanisms enable a fast reaction, they raise the question of how conditional rerouting rules can be defined to maintain a high resilience under multiple
link failures. 
Ideally, the local rerouting mechanisms always provide for a given target node~$t$
a valid routing from any source $s$ to~$t$ as long as~$s$ and $t$ are still connected in the underlying graph after the failures: 
a property which is known as \emph{perfect resilience} in the literature.
Intuitively, given the locally limited failure information at nodes, this may be difficult to achieve in the presence of multiple link failures: as nodes only have local information, about incident link failures, but not about the status of links in other parts of the network, packets may be forwarded in a loop.
Indeed, it was shown already in the seminal work at ACM PODC~'12~\cite{FGPSSS12} that there are examples where perfect resilience is impossible.
In other words, there is a \emph{price of locality}: local fast rerouting comes at a cost of reduced resilience.
They posed the open problem of characterizing instances that allow for perfectly resilient local rerouting rules and this question was repeated for example in APOCS~'21~\cite{FHPST21} and in a recent keynote talk at DISC~'24~\cite{disc-keynote}.

The design of local fast rerouting algorithms has received
much attention over the last years.
Different variations have been considered and we next give an overview of some of these variants and the current state of the art.
\begin{itemize}
    \item Researchers studied restricted failure scenarios in which the number $f$ of link failures is bounded. In this context, a local fast rerouting scheme which tolerates $f$ link failures is called \emph{$f$-resilient}. It is known that $2$-resilient routing is always possible if the graph is $3$-link-connected~\cite{frr-infocom16} or when rerouting rules can also depend on the source of a packet in addition to the target~\cite{dai2023tight}.

    \item Significant research focused on the special case of highly connected graphs, namely $k$-link-connected graphs.
    In this context, a local fast rerouting scheme which is $(k-1)$-resilient is called \emph{ideally resilient}.
    Note that since $k-1$ link failures never disconnect a $k$-link-connected graph, ideal resilience describes a weaker notion of resilience than perfect resilience.
    It is known that perfect resilience is always possible when~$k \leq 5$~\cite{frr-ton} as well as in scenarios where the rerouting rules can additionally also depend on the source~\cite{frr-ton}.
    However, these results do not provide insights into the more general setting of perfect resilience.

    \item In 2021, Foerster, Hirvonen, Pignolet, Schmid, and Trédan~\cite{FHPST21} showed that instances with outerplanar graphs are always perfectly resilient and allow for simple and efficient rerouting algorithms based on \emph{skipping}: each node stores an ordered priority list of alternative links to try per in-port; failed links are then simply skipped.
    While skipping leads to compact routing tables and is naturally supported by router hardware, it is an open question whether rerouting algorithms which are restricted to skipping (as well as to circular arborescence routing) come at a price of reduced resilience~\cite{frr-ton,FHPST21}.
    We will show in this paper that this is not the case in the context of perfect resilience.

    \item Bentert, Ceylan-Kettler, Hübner, Schmid, and Srba~\cite{bentert2026fast} recently showed that
    \emph{verifying} whether a given forwarding pattern for a network provides perfect resilience is coNP-complete.
    Amusingly, this implies that it is significantly easier to compute a solution (in case it exists) than to verify it.
\end{itemize}

In order to gain additional insights into the structure of perfectly resilient scenarios, tools based on formal methods have been developed to generate routing tables and counterexamples in an automated manner.
However, these algorithms have super-polynomial running times~\cite{frr-syper24,frr-syrep24}. 
The problem was further studied from a randomized perspective~\cite{frr-icalp16,disc21} as well as in more general settings with dynamic packet headers~\cite{stephens2013plinko,yang2014keep,elhourani2016ip,frr-ton}, dynamic node states~\cite{gafni2003distributed,liu2011data}, or the aforementioned scenario where the routing can depend on the source~\cite{FHPST21,dai2023tight}; these models are, however, less practical~\cite{FGPSSS12,FHPST21}.
While the focus of our paper is on the fundamental theoretical structure of perfect resilience, the topic is also studied intensively in the networking community from a more practical perspective~\cite{frr-ankit,holterbach2017swift,foerster2018ti,conext19failover,jensen2020aalwines}, with several protocols being subject to IETF standardization\footnote{E.g.\,\url{https://datatracker.ietf.org/doc/draft-ietf-rtgwg-segment-routing-ti-lfa/}}. We refer the reader to a  recent survey~\cite{frr-survey} for a detailed overview.

Before we present our contributions in detail, let us introduce our model more formally.
The input consists of a graph and a target node~$t$. For each node~$v$ in the graph, we need to compute a rerouting table containing conditional failover rules, which are described by a \textbf{\textit{forwarding function}}~${\pi_v \colon (E_v \cup \{\bot\}) \times 2^{E_v} \rightarrow (E_v \cup \{\bot\})}$, where~$E_v$ denotes the set of all links incident to~$v$.
The forwarding function takes as input a link~$e$ incident to~$v$ (called the in-port; $\bot$ models that the routing starts in~$v$) and a subset of incident links (called the failed links) and outputs an incident link~$e'$ (called the out-port) or~$\bot$ if all incident links fail.
A \textbf{\textit{forwarding pattern}} is a collection~$\pi = (\pi_v)_{v \in V \setminus \{t\}}$ of forwarding functions for each node except for the target node~$t$.
A forwarding pattern, a starting node~$s$, and a set~$F \subseteq E$ of failed links determine a \textbf{\textit{routing}}~$\rho=(\rho_1,\rho_2,\ldots)$ (a sequence of nodes) as follows.
The sequence starts with~$\rho_1 = s$ and the second entry is~$\rho_2 = (\pi_s(\bot,E_s \cap F))$.
For~$i \geq 3$, the~$i$\textsuperscript{th} entry of~$\rho$ is~$\rho_i = \pi_{\rho_{i-1}}(\{\rho_{i-2},\rho_{i-1}\}, E_{\rho_{i-1}} \cap F)$.
The sequence ends if~$\rho_i = t$ at some point.
Otherwise, it never ends.
We can now define the two computational problems we study in this work.

\smallskip

\defproblem{Perfect Resilience Decision}
{A graph~$G=(V,E)$ and a target node~$t \in V$.}
{Is there a forwarding pattern such that for each starting node~$s \in V$ and each set~$F \subseteq E$ of link failures, it holds that the resulting routing contains~$t$ as long as~$s$ and~$t$ are in the same connected component in the graph~$(V, E \setminus F)$?}

\taskproblem{Perfect Resilience Synthesis}
{A graph~$G=(V,E)$ and a target node~$t \in V$.}
{Compute a forwarding pattern such that for each starting node~$s$ and each set~${F \subseteq E}$ of link failures, it holds that the resulting routing contains~$t$ as long as~$s$ and~$t$ are in the same connected component in the graph~$(V,E \setminus F)$ or correctly determine that such a forwarding pattern does not exist.}

An even stronger requirement would be to ask whether a given graph is perfectly resilient for any target node~$t$.
The benefit of considering a specific target node~$t$ (and presumably the reason this variant is introduced in the literature) is that it gives more fine-grained results.
Even if the input does not satisfy this stronger notion, it still allows to construct perfectly resilient forwarding patterns for a subset of all possible target nodes.
Clearly, any algorithm for \prd{} can be used to solve the variant with the stronger requirement at the expense of an additional factor of~$n$ in the running time.
We note that our results generalize to this setting directly and thus this additional factor can be avoided.

\subsection{Our Contribution}

This paper presents a complete characterization of when perfect resilience can be achieved. 
This answers the aforementioned long-standing open question~\cite{FGPSSS12}.
Formally, we show the following.

\begin{restatable}{maintheorem}{structure}
    \label{main:characterization}
    A rooted graph~$(G,t)$ is perfectly resilient if and only if it does not contain the~$\Kf$, the~$\Kt$, the~$\Ktf$, or the~$\sK$ (see \cref{fig:obstructions}) as a rooted minor.
\end{restatable}

\begin{figure}[t]
    \centering
    \hfill
    \begin{subfigure}{0.27\textwidth}
        \begin{tikzpicture}
            \node[ver] at(-3,5) (s) {};
            \node[ver,label=$t$] at(-3,7) (t) {};
            \node[ver] at(-5,4) (u) {} edge(s) edge(t);
            \node[ver] at(-1,4) (v) {} edge(s) edge(t) edge(u);
            \node[ver] at(-3,6) (w) {} edge(s) edge(t) edge(u) edge(v);
        \end{tikzpicture}
        \caption{\Kf}
        \label{fig:k5}
    \end{subfigure}
    \hfill
    \begin{subfigure}{0.27\textwidth}
        \begin{tikzpicture}
            \node[ver,label=$t$] at(3,7) (t) {};
            \node[ver] at(1,5) (u) {} edge(t);
            \node[ver] at(5,5) (v) {} edge(t);
            \node[ver] at(3,6) (p) {} edge(u) edge(v);
            \node[ver] at(3,4) (q) {} edge(u) edge(v);
            \node[ver] at(3,5) {} edge(p) edge(q);            
        \end{tikzpicture}
        \caption{\Kt}
        \label{fig:k3}
    \end{subfigure}
    \hfill\hfill
    
    \medskip

    \hfill
    \begin{subfigure}{0.27\textwidth}
        \begin{tikzpicture}
            \node[ver,label=$t$] at(3,2) (t) {};
            \node[ver] at(3,1) (w) {} edge(t);
            \node[ver] at(1,0) (u) {} edge(t);
            \node[ver] at(5,0) (v) {} edge(t);
            \node[ver] at(3,0) {} edge(u) edge(v) edge(w);
            \node[ver] at(3,-1) {} edge(u) edge(v);
            \node[ver] at(3,-2) {} edge(u) edge(v);
        \end{tikzpicture}
        \caption{\Ktf}
        \label{fig:k34}
    \end{subfigure}
    \hfill
    \begin{subfigure}{0.27\textwidth}
        \begin{tikzpicture}
            \node[ver] at(-4,2) (w) {};
            \node[ver,label=$t$] at(-2,2) (t) {} edge(w);
            \node[ver] at(-5,0) (u) {} edge(w);
            \node[ver] at(-1,0) (v) {} edge(t);
            \node[ver] at(-3,1) {} edge(u) edge(v);
            \node[ver] at(-3,0) {} edge(u) edge(v);
            \node[ver] at(-3,-1) {} edge(u) edge(v);            
        \end{tikzpicture}
        \caption{\sK}
        \label{fig:sk}
    \end{subfigure}
    \hfill\hfill
    \caption{The four structures inhibiting perfect resilience.}
    \label{fig:obstructions}
\end{figure}
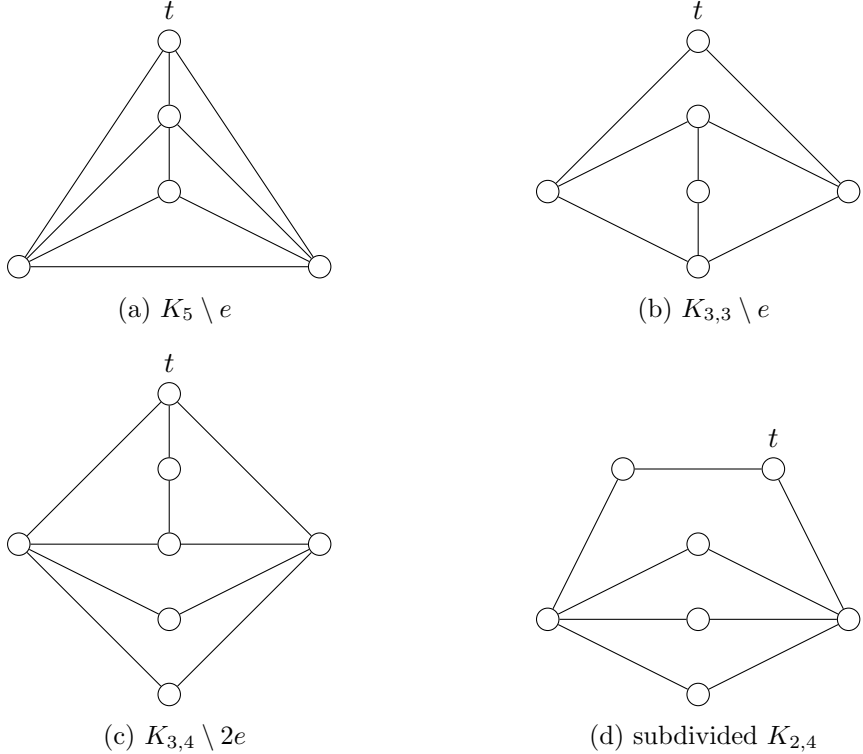

Our characterization is of particular interest in that it implies a polynomial-time algorithm to compute perfectly resilient rerouting rules whenever these exist.
\begin{restatable}{maintheorem}{algo}
    \label{main:algorithm}
    \prd{} can be solved in~$O(n)$ time.
    \prs{} can be solved in~$O(nm)$ time such that the output for perfectly resilient inputs is a collection of skipping forwarding rules.
\end{restatable}
Therein, $n$ and~$m$ denote the number of nodes and links in the input graph, respectively.
Our algorithm also implies that the compact encoding of skipping forwarding patterns is as powerful as general (exponential-size) encodings in the context of perfect resilience.
This partially answers another major open research question~\cite{frr-ton,FHPST21,disc-keynote}.

We next discuss that both of our algorithms are asymptotically optimal.
Simple examples show that the output for an instance of \prs{} with skipping priority lists can have size~$\Theta(nm)$.
We present such an example in \cref{fig:planar} in \cref{sec:prevresults}.
This shows that the running time of our algorithm is optimal in the context of skipping priority lists.
Amusingly, if the input graph is connected, then the running time can be improved to~$O(n^2)$.
For \prd, the running time is clearly optimal.
It is even a little surprising that the running time is in fact faster than~$O(n+m)$ time and that not all links of the input need to be read.
This is a consequence of using an algorithm with the same property to test whether a connected component in a given graph is planar or not.

\subsection{Our Methods}

We next give a high-level overview of the proofs of our two main results.

\paragraph{Structure.}
We first show (in \cref{sec:notpr}) that the \Ktf{} and the \sK{} are not perfectly resilient.
Combined with the known results that the \Kf{} and the \Kt{} are not perfectly resilient and the result that if a rooted graph~$(G,t)$ contains a rooted minor that is not perfectly resilient, then also~$(G,t)$ is not~\cite{FHPST21}, this shows that any rooted graph which contains one of the four mentioned structures as a rooted minor is not perfectly resilient.
Showing the reverse---any rooted graph which is not perfectly resilient contains one of the four structures as a rooted minor---is trickier.
We first show (in \cref{sec:pre}) that we can assume without loss of generality that the input graph has some nice properties.
These include that it is planar (already known), biconnected (new), and does not contain \emph{separating links}, that is, links such that removing both endpoints of the link disconnects the graph (new).
Ensuring the last property is the key novel idea compared to previous approaches.
To prove that such separating links can be ignored, we consider the interplay between routings corresponding to different starting nodes and different sets of failed links.
With the above properties, we then show (in \cref{sec:traps}) via an intricate analysis that if a rooted graph is not perfectly resilient, then it contains one of the four mentioned structures as a rooted minor.

\paragraph{Decision Problem.}
The linear-time algorithm for \prd{} follows a similar strategy to the above. 
We first check whether the connected component containing~$t$ is planar and if so, then we compute the connected component and a planar embedding for it in~$O(n)$~time.
If the connected component is not planar, then we can immediately conclude that the input is not perfectly resilient.
We then show (in \cref{sec:alg}) that if the connected component contains a $4 \times 4$ grid as a minor (ignoring~$t$), then it cannot be perfectly resilient as it contains \Kf{} as a rooted minor.
Using known results due to Gu and Tamaki~\cite{GT12}, this bounds the branchwidth (and therefore treewidth) of the connected component containing~$t$ in all yes-instances to a constant.
This allows us to use known linear-time algorithms for finding rooted minors in graphs of constant branchwidth/treewidth.
While the well-known result due to Courcelle suffices for a linear-time algorithm, the hidden constants would be quite large and fortunately, we can rely on a significantly faster algorithm due to Adler, Dorn, Fomin, Sau, and Thilikos~\cite{ADFST11}.

\paragraph{Synthesis Problem.}
We next present our main result, the~$O(nm)$-time algorithm for \prs.
First, we present the three main steps on a high level and then present them in a little more detail.
The first step (discussed in full detail in \cref{sec:pre}) consists of performing different pre- and postprocessing procedures to simplify the structure.
We mention that in order to show that skipping priority lists can always be used, we require the skipping priority lists for the simplified graph have a particularly simple structure related to a planar embedding of the graph.
The second step is then to show (in \cref{sec:traps}) that if the simplified graph is perfectly resilient, then it has one of three particular structures.
The third step is to show (in \cref{sec:newpr}) how to construct perfectly resilient skipping priority lists with the required simple structure for each rooted graph of each of these three structures in~$O(nm)$ time. 

Let us give a few more details for each of the three main steps of our algorithm for \prs, starting with the third.
We first introduce two classes of rooted graphs and show (in \cref{sec:newpr}) that each rooted graph in these classes is perfectly resilient.
Moreover, we show that \emph{right-hand forwarding patterns}---the aforementioned restriction of skipping priority lists---is sufficient in these cases.
It was previously known that these also suffices for instances in which~$G-t$ (the graph where~$t$ and all incident links are removed) is outerplanar~\cite{FHPST21}.
This forms the third class of graphs which can remain after performing our preprocessing.

Let us now discuss the first main step of our algorithm.
\cref{tab:pre} gives an overview of the different preprocessing rules which we discuss in the following.
\begin{table}[t]
    \centering
    \begin{tabular}{l r}
        connected $\&$ planar & skipping $\rightarrow$ skipping \\
        biconnected & skipping $\rightarrow$ skipping\\
        no separating links & right-hand $\rightarrow$ skipping \\ 
    \end{tabular}
    \caption{An overview of the different pre- and postprocessing steps we employ. The left column describes what the reduction rule achieves and $A \rightarrow B$ in the right column stands for the following: If the reduced instance produces a forwarding pattern of type~$A$, then the postprocessing step produces forwarding patterns of type~$B$.}
    \label{tab:pre}
\end{table}
To this end, let~$(G,t)$ be the input.
We show in \cref{sec:pre} how to compute a connected planar graph~$G_1$ (this will be the connected component of the input graph containing~$t$) such that~$(G_1,t)$ is perfectly resilient if and only if~$(G,t)$ is.
Moreover, given a perfectly resilient skipping forwarding pattern for~$(G_1,t)$, we show how to compute a perfectly resilient skipping forwarding pattern for~$(G,t)$ in~$O(nm)$ time.
We then show how to handle cut nodes, that is, nodes whose removal disconnects the graph.
We construct in~$O(n^2)$ time a set of instances~$(G'_i,t_i)$ of total size~$O(n)$ such that~$(G_1,t)$ is perfectly resilient if and only if all instances~$(G'_i,t_i)$ are.
Moreover, given a perfectly resilient skipping forwarding pattern for all instances, we show how to compute a perfectly resilient skipping forwarding pattern for~$(G_1,t)$ in~$O(n^2)$ time.
We then solve each of the constructed instances individually in~$O(|G'_i|^2)$ time.
This results in an overall running time of~$O(n^2)$.
Let~$(G_2,t')$ be one of these instances.
The final preprocessing rule deals with separating links.
We show how to compute the set~$S$ of separating links in~$G_2$ in~$O(n^2)$ time.
Let~$G_3$ be the result of removing all links in~$S$ from~$G_2$.
We show that~$(G_3,t')$ is perfectly resilient if and only if~$(G_2,t')$ is.
Moreover, given a perfectly resilient right-hand forwarding pattern, we construct a perfectly resilient skipping forwarding pattern for~$(G_2,t')$ in~$O(n^2)$ time.
We then show (in \cref{sec:traps}) that if~$(G_3,t')$ is perfectly resilient, then~$G_3 - t'$ is outerplanar or~$(G_3,t')$ belongs to one of the two classes of rooted graphs we analyzed in the beginning.
We note that this statement is not true for general instances~$(G,t)$ that have not been preprocessed.
We can check in~$O(n^2)$ time whether~$(G_3,t')$ fulfills any of these requirements (and also which case applies).
This allows us to compute in~$O(n^2)$ time a perfectly resilient right-hand forwarding pattern for~$(G_3,t')$ or correctly conclude that~$(G_3,t')$ is not perfectly resilient and thus the input~$(G,t)$ is a no-instance.
If~$(G_3,t')$ is perfectly resilient, then we use the above results to compute a perfectly resilient skipping forwarding pattern for~$(G_2,t')$.
Given perfectly resilient skipping forwarding patterns for all instances~$(G'_i,t)$, we compute a perfectly resilient skipping forwarding pattern for~$(G_1,t)$ and from that a perfectly resilient skipping forwarding pattern for~$(G,t)$.
The overall running time is in~$O(nm)$.
Interestingly, the slowest part is the construction of arbitrary skipping priority lists for nodes in different connected components than~$t$.
Assuming the input graph is connected, our algorithm runs in~$O(n^2)$ time.

\subsection{Future Work}
While perfect resilience is the most basic and fundamental problem variant, we believe that our approach allows to characterize additional related models studied in the literature. In particular, the networking community sometimes also considers more powerful router models, in which nodes can match the source or dynamically rewrite header bits. While requiring additional hardware capabilities, such additional information can help realize resilient routings even in scenarios where perfect resilience can otherwise not be achieved.
In order to see how our framework can be applied in such alternative contexts, let us revisit the three steps of our algorithm:
First, one should analyze whether yes-instances (of the decision problem) are closed under taking rooted minors.
This is already known for source matching~\cite{FHPST21} and we believe that this also holds for any constant number of header bits.
It then follows from the graph minor theorem of Robertson and Seymour that there is a polynomial-time algorithm for recognizing these.
However, the theorem is non-constructive and characterizing the set of minimal obstructions is the second step.
Here, we believe that our preprocessing, especially removing separating links, will again be useful for the final analysis.
We mention in passing that our preprocessing for the decision problem directly applies to the setting with source matching as well.
The third and final step is to construct perfectly resilient forwarding patterns whenever this is possible.
Our work shows that skipping priority lists are as powerful as arbitrary forwarding functions in the context of perfect resilience in the setting we consider.
We conjecture that this also holds in these other two settings.

\section{Discussion and Overview of Previous Results}
\label{sec:prevresults}
In this section, we state and discuss results from the existing literature that will later be relevant for proving our main results.
We start with results regarding perfect resilience and graph minors and then continue with graph-theoretic results regarding planar graphs.
We next give informal definitions for some concepts required in the following discussion.
For formal definitions, we refer to \cref{sec:prelim}.
A \emph{rooted graph} is a graph with a set of specified special nodes.
In this paper, we only consider rooted graphs with a single root (the target node~$t$).
We call a rooted graph~$(G,t)$ that is not perfectly resilient a \emph{trap}.
A trap is \emph{minimal} if it does not contain any other trap as a rooted minor.

\subsection{Perfect Resilience and Graph Minors}
Most of the relevant results for perfect resilience are due to Foerster, Hirvonen, Pignolet, Schmid, and Trédan~\cite{FHPST21}.
They showed among other things the following.
\begin{proposition}[\cite{FHPST21}]
    \label{prop:subgraph}
    If a rooted graph~$(G,t)$ is perfectly resilient, then so is~$(G',t)$ for any subgraph~$G'$ of~$G$ that contains~$t$.
\end{proposition}
\begin{proposition}[\cite{FHPST21}]
    \label{prop:contraction}
    If a rooted graph~$(G,t)$ is perfectly resilient, then so is~$(G',t)$ for any graph~$G'$ that is the result of contracting a link in~$G$.
\end{proposition}
\begin{proposition}[\cite{FHPST21}]
    \label{prop:k5k33notpr}
    The complete graph~$K_5$ on five nodes and the complete bipartite graph~$K_{3,3}$ with three nodes on each side are not perfectly resilient independent of where the target node~$t$ is.
\end{proposition}

The last result was generalized in follow-up work, where the authors showed the following.
\begin{theorem}[\cite{FHPST22}]
    \label{prop:notres}
    The rooted graphs~$(K_5 \setminus e,t)$ and~$(K_{3,3} \setminus e,t)$ are not perfectly resilient.
\end{theorem}

Unfortunately, there are two minor flaws in their argument.
First, in the proof of \cref{prop:contraction} the authors did not consider the case where a link incident to~$t$ is contracted.
However, using the same arguments as in the original proof, it is easy to show that simply making the newly created node the new target~$t$ is enough to generalize the proposition to contracting links incident to~$t$.
Second, from the above three propositions they concluded that perfect resilience is closed under taking minors and any non-planar graph is not perfectly resilient.
Here, they again did not consider the role of the target~$t$.
On the one hand, speaking about perfectly resilient graphs does not make sense as the same graph can be perfectly resilient for some target and not perfectly resilient for another.
An example of this is shown in \cref{fig:planar}.
\begin{figure}[t]
    \centering
    \begin{tikzpicture}
        \foreach \i in {1,2,3,4,5}{
            \node[ver] (v\i) at(72*\i+18:1) {};
        }
        \foreach \i in {1,2,3,4,5}{
            \foreach \j in {\i,...,5}{
                \ifnum \i<\j
                    \draw (v\i) to (v\j);
                \fi
            }
        }
        \node[label=$u$] at(0,1) {};
        \node[ver,label=$t$] at(-3,1) (t) {};
        \node[ver] at(-3,-.85) {} edge(t);
    \end{tikzpicture}
    \caption{A perfectly resilient instance with a non-planar graph. If the target node was~$u$ instead, then the instance would not be perfectly resilient.}
    \label{fig:planar}
\end{figure}
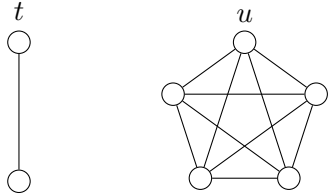%
On the other hand, there are graphs that are not planar but perfectly resilient for some target.
\cref{fig:planar} shows an example of such a case.
Admittedly, this only happens for disconnected graphs and we will show later that we can assume without loss of generality that the input graph is connected.
The reason we included this discussion in such detail is that due to the graph minor theorem by Robertson and Seymour~\cite{RS95,RS04}, any graph class that is closed under taking minors can be recognized in polynomial time.
Fortunately, \cref{prop:subgraph,prop:contraction} show that perfectly resilient rooted graphs are closed under taking \emph{rooted minors}.
\begin{corollary}
    \label{cor:rootedminor}
    If a rooted graph~$(G,t)$ is perfectly resilient, so is each rooted minor~$(H,t)$ of~$(G,t)$.
\end{corollary}
Moreover, a similar result to the graph minor theorem by Robertson and Seymour is also known for rooted minors.
In particular, it states that for any graph class that is closed under taking rooted minors, the set of all minimal obstructions (rooted graphs not contained in the graph class but where each rooted minor belongs to the graph class) is finite~\cite{RS90}.
Since the~\Kf{} is planar (see \cref{fig:k5}) and not perfectly resilient by \cref{prop:notres}, this yields the following.
\begin{corollary}
    \label{cor:mintrap}
    The set of minimal traps is finite and a rooted graph~$(G,t)$ is perfectly resilient if and only if it does not contain a minimal trap as a rooted minor.
\end{corollary}
Using a recent improvement for finding rooted minors in almost linear time~\cite{KPS24}, this immediately implies the following.
\begin{corollary}
    \label{cor:almostlinear}
    There exists an algorithm solving \prd{} in~$n^{1+o(1)}$ time.
\end{corollary}
Unfortunately, the result by Robertson and Seymour is non-constructive, so even though we know such an algorithm exists, the theory does not help us designing it.
Coming back to specific rooted graphs that are or are not perfectly resilient, we next discuss a useful lemma regarding perfectly resilient forwarding patterns~\cite{FHPST21}.
Therein, an \emph{active neighbor} of a node~$v$ is a node~$u$ such that there is a link~$\{u,v\}$ and this link does not fail.
A \emph{relevant neighbor} of~$v$ is an active neighbor~$u$ such that there exists a path in~$G$ from~$u$ to~$t$ that does not contain any other active neighbors of~$v$.
Given a forwarding function~$\pi_v$, a set~$F_v$ of incident failed links, and an active neighbor~$u$ of~$v$, the authors of the lemma considered the following sequence~$\sigma=(w_1,w_2,\ldots)$ of active neighbors~$w_i$ of~$v$, which they call the \emph{orbit} of~$u$.
The sequence starts with~$w_1 = u$.
Then, the next entry is iteratively constructed by~$w_{i+1} = \rouns{v}{F_v}{w_i}$.
That is, they simulate what happens if each active neighbor immediately routes back to~$v$.
The authors then observed the following.
\begin{lemma}[\cite{FHPST21}]
    \label{lem:orbit}
    Let~$(G,t)$ be a rooted graph and let~$v$ be a node in~$G$.
    Let~$\pi_v$ be a forwarding function for~$v$.
    If~$\pi_v$ is part of a perfectly resilient forwarding pattern for~$(G,t)$, then it holds for each possible set~$F_v$ of incident failed links and each relevant neighbor~$u$ of~$v$ with respect to~$F_v$ that all relevant neighbors of~$v$ with respect to~$F_v$ appear in the orbit of~$u$.
\end{lemma}
Note that the above lemma immediately implies that if there are at least two relevant neighbors of~$v$, then~$\rouns{v}{F_v}{u} \neq u$ for any relevant neighbor~$u$.
Moreover, if there are exactly three active neighbors~$x,y$, and~$z$ of~$v$ with respect to~$F_v$ that are all relevant and~$\rouns{v}{F_v}{x}=y$, then~${\rouns{v}{F_v}{y}=z}$ as otherwise~$z$ is not part of the orbit of~$x$.
Similarly~$\rouns{v}{F_v}{z}=x$ as otherwise~$x$ would not be part of the orbit of~$y$.
In the same paper, the authors also showed that instances with outerplanar graphs are perfectly resilient independent of where the root~$t$ is. They also generalized this to the following theorem which will be useful for us later.
Therein, the updated right-hand rule is a special case of a skipping forwarding function.
\begin{theorem}[\cite{FHPST21}]
    \label{thm:gt}
    Let~$(G,t)$ be a rooted graph. If~$G-t$ is outerplanar, then~$(G,t)$ is perfectly resilient.
    Moreover, the updated right-hand rule~$\lambda^{e_v}_v$ for each node~$v$ and a chosen link~$e_v$ results in a perfectly resilient routing.
\end{theorem}
The chosen link~$e_v$ in the above theorem for a node~$v$ is the link that enters~$v$ if traversing the outer face in clockwise order.
While not specifically analyzed in the paper, it is not difficult to verify that we can compute this link for all nodes in overall~$O(n)$~time and then compute the corresponding updated right-hand rule in~$O(n^2)$ time.
Note that this is optimal as a node with~$O(n)$ incident links has~$O(n)$ priority lists each of size~$O(n)$, so the output can have size~$\Theta(n^2)$.

\subsection{Planar Graphs and Branchwidth}

We continue with known graph-theoretic and algorithmic results that will be important for us later.
We start with relevant results for planar graphs.
First, we can check whether a given graph is planar in~$O(n)$ time using the classic algorithm by Hopcroft and Tarjan~\cite{HT74}.
Given a node~$v$, the algorithm can in the same time also decide whether the connected component containing~$v$ is planar and compute an embedding only for this component.
We will use the data structure known as doubly connected half-edge lists to represent planar graphs.
The exact definition is not relevant for the following discussion and thus deferred to \cref{sec:prelim}.
\begin{proposition}[\cite{HT74}]
    \label{prop:embedding}
    Given a graph~$G$ with~$n$ nodes and a node~$v$, we can decide in~$O(n)$ time whether~$G$ and/or the connected component of~$G$ containing~$v$ are planar and compute a planar embedding of~$G$ and/or the connected component if it is (and initialize a doubly connected half-edge list data structure for this embedding).
\end{proposition}
It will be convenient for the presentation to assume that we have an embedding where~$t$ is incident to the outer face.
It is folklore knowledge that for any planar embedding of a planar graph~$G$ and each face in this embedding, there exists an embedding in which this chosen face is the outer face and where the doubly connected half-edge list data structures for both embeddings are identical.
Hence, we can make the above assumption without loss of generality.
As a consequence of \cref{prop:embedding}, it is folklore knowledge that outerplanar graphs can also be recognized in~$O(n)$~time.
Simply add a new node and make it adjacent to all other nodes.
The resulting graph is planar if and only if the original graph was outerplanar and a planar embedding of the constructed graph with the new node on the outer face results in an outerplanar embedding for the original graph after removing the new node with all incident links.
\begin{corollary}
    \label{cor:outerplanar}
    Given a graph~$G$, we can decide in~$O(n)$ time whether~$G$ is outerplanar and compute an outerplanar embedding in case it exists.
\end{corollary}

Next, by the well-known result due to Wagner~\cite{Wag37}, a graph is planar if and only if it does not contain the graphs~$K_5$ or~$K_{3,3}$ as a minor.

\begin{theorem}[\cite{Wag37}]
    \label{thm:wagner}
    A graph~$G$ is planar if and only if it does not contain~$K_5$ or~$K_{3,3}$ as a minor.
\end{theorem}

Next, it is folklore knowledge that in a 2-node-connected plane graph, each face is bounded by a cycle~\cite{Diestel}.
\begin{proposition}
    \label{prop:cycle}
    In a 2-node-connected plane graph, every face is bounded by a simple cycle.
\end{proposition}

The last important result about planar graphs is that each planar graph either has constant branchwidth or contains a large grid as a minor.
The currently best known bound is due to Gu and Tamaki~\cite{GT12}.
\begin{theorem}[\cite{GT12}]
    Let~$G$ be a planar graph and let~$g$ be the largest integer such that~$G$ contains a~$g \times g$ grid as a minor. Then, the branchwidth of~$G$ is at most~$3g$.
\end{theorem}
We will be interested in planar graphs not containing a $4 \times 4$ grid as a minor.
Moreover, it is known that the treewidth of a graph is at most~$\nicefrac{3}{2}$ of its branchwidth (minus one)~\cite{RS91}.
This yields the following corollary.
\begin{corollary}
    \label{cor:branchwidth}
    Let~$G$ be a planar graph that does not contain a $4 \times 4$ grid as a minor.
    Then, the branchwidth of~$G$ is at most nine and the treewidth of~$G$ is at most twelve.
\end{corollary}

Using the fact that we will deal with graphs of constant branchwidth and treewidth, we can compute a branch decomposition of width at most nine or a tree decomposition of width at most~25 in~$O(n)$ time~\cite{BT97,Kor21}.
Using either of these decompositions, we can then use standard algorithms to find a given rooted minor in linear time~\cite{RS86,ADFST11,KPS24}.
\begin{corollary}
    \label{cor:algforbw}
    Let~$(G,t)$ be a rooted graph where~$G$ is a planar graph with~$n$ nodes and let~$(H,t)$ be a rooted graph.
    Let a branch decomposition or tree decomposition of~$G$ of constant width be given.
    Then, we can test whether~$(H,t)$ is a rooted minor of~$(G,t)$ in~$O(n)$ time.
\end{corollary}

\section{Preliminaries}
\label{sec:prelim}
In this section, we introduce notation and concepts used throughout the paper.
We assume the reader to be familiar with the Bachmann-Landau notation (also known as big-O notation).
For an introduction, we refer the reader to the textbook by Cormen, Leiserson, Rivest, and Stein~\cite{introduction}.
For a positive integer~$k$, we denote the set~$\{1,2,\ldots,k\}$ by~$[k]$.

\subsection{Graphs and (Rooted) Minors}
We use standard graph notation.
In particular, a graph~$G=(V,E)$ is a tuple where~$V$ is the set of nodes and~$E \subseteq \binom{V}{2}$ is a set of (undirected) links.
We will occasionally also use directed graphs.
A directed graph~$D=(V,A)$ is a tuple where~$V$ is the set of nodes and~$A \subseteq V \times V$ is the set of directed arcs.
We will denote the size of~$V$ by~$n$ and the size of~$E$ or~$A$ by~$m$.
Two nodes~$u$ and~$v$ are called adjacent if there is a link~$\{u,v\}$ between them.
In this case, we also say that~$u$ and~$v$ are incident to the link~$\{u,v\}$ and that~$u$ and~$v$ are endpoints of~$\{u,v\}$.
For a node~$v$, we denote the set of all adjacent nodes by~$N(v)$ and the set of all incident links by~$E_v$.
A \textbf{\textit{path}} in a graph~$G$ is a sequence~$(v_0,v_1,\dots,v_\ell)$ of distinct nodes such that each consecutive pair~$\{v_{i-1},v_i\}$ is connected by a link in~$G$.
The first and last node $v_0$ and $v_\ell$ are called the end points of $P$ and all other nodes are called internal nodes of~$P$.
We also say that~$P$ is a path from $v_0$ to $v_\ell$ or a~$v_0$-$v_\ell$-path.
Two paths are \textbf{\textit{internally disjoint}} if the internal nodes of either path do not appear in the other path.
Equivalently, any node that appears in both paths is an endpoint of each.
A \textbf{\textit{walk}} is similar to a path but nodes may repeat in a walk.
A \textbf{\textit{cycle}} (or sometimes called simple cycle) in a graph is a path~$(v_0,v_1,\ldots,v_\ell)$ with an additional link~$\{v_0,v_\ell\}$.

A \textbf{\textit{subgraph}} of~$G$ is a graph~$H = (U,F)$ with~$U \subseteq V$ and~$F \subseteq E$.
If~$H$ is a subgraph of~$G$, then we also say that~$G$ is a supergraph of~$H$.
For a subset $V' \subseteq V$ of nodes in a graph~$G$, we denote by~$G[V']$ the \textbf{\textit{subgraph of $G$ induced by~$V'$}}, that is, $G[V']=(V',\{e \in E \mid e \subseteq V'\})$.
For a node~$v$, a set~$V'$ of nodes, a link~$e$, and a set~$E'$ of links, we also use~$G-v, G-V', G \setminus e$ and~$G \setminus E'$ to denote~$G[V \setminus \{v\}], G[V \setminus V'], (V,E\setminus \{e\})$, and~$(V,E \setminus E')$, respectively.
A graph is \textbf{\textit{connected}} if there is a path between each pair of nodes.
A \textbf{\textit{connected component}} in a graph~$G$ is a maximal subset~$V' \subseteq V$ such that for every pair of nodes~$u, v \in V'$, there exists a $u$-$v$-path in~$G[V']$.
That is, the subgraph~$G[V']$ is connected, and no proper superset of~$V'$ induces a connected subgraph.
Two graphs~$G=(V,E)$ and~$H=(U,F)$ are \textbf{\textit{isomorphic}} if there exists a link-preserving bijection between the nodes in~$V$ and the nodes in~$U$.
That is, $G$ and~$H$ are isomorphic if and only if there exists a bijection~$f \colon V \rightarrow U$ such that for any pair~$u,v \in V$ of nodes in~$G$, there is a link~$\{u,v\} \in E$ if and only if there is a link~$\{f(u),f(v)\} \in F$.

A \textbf{\textit{cut node}} in a connected graph~$G$ is a node~$v$ such that~$G - v$ is disconnected.
A graph~$G$ is \textbf{\textit{biconnected}} if it is connected and does not contain any cut nodes.
A graph~$G$ is \textbf\textit{{2-node-connected}} if for every pair~$u,v$ of nodes, there are at least two internally disjoint~$u$-$v$-paths.
A connected graph is 2-node-connected if and only if it is biconnected and it is not the graph with exactly 2 nodes and one link.
A graph is~$k$-link-connected if for each pair~$u,v$ of nodes, there are~$k$ paths between~$u$ and~$v$ that pairwise do not share any links.
We will abbreviate 2-node-connectivity to simply 2-connectivity and mention explicitly when we talk about link-connectivity.
A \textbf{\textit{biconnected component}} is a maximal biconnected subgraph.
Any connected graph decomposes into a tree of biconnected components which are attached to each other at shared cut nodes.
A \textbf{\textit{separating link}} is a link~$\{u,v\}$ such that removing both~$u$ and~$v$ (and all incident links) disconnects the graph.
We say that such a link \emph{separates} a node~$w \notin \{u,v\}$ from a node~$x \notin \{u,v,w\}$ if~$w$ and~$x$ are in different connected components in~$G - \{u,v\}$.
We say that~$\{u,v\}$ \emph{separates} another link~$e \neq \{u,v\}$ from~$x$ if~$\{u,v\}$ separates at least one endpoint of~$e$ from~$x$.
We also say that~$e$ is \emph{separated} from~$w$ by~$\{u,v\}$.

The \textbf{\textit{treewidth}} and the \textbf{\textit{branchwidth}} of a graph are measures for how ``tree-like'' it is.
As the precise definitions are not essential for our purposes, we refer the reader to the work of Robertson and Seymour~\cite{RS91} for details.

A graph is \textbf{\textit{planar}} if it can be drawn in the plane without any link crossings.
A graph with such an embedding is called \emph{plane} and the maximal regions (bounded by links) that do not contain any nodes or links are called \textbf{\textit{faces}}.
The infinite outer region bounded by links is called the outer face and a face is incident to a link if it is bounded by the link.
A face is incident to a node~$v$ if it is incident to a link~$\{u,v\}$ for some other node~$u$.
A graph is outerplanar if there is a planar embedding where all nodes are incident to the outer face.
We will represent planar graphs using \textbf{\textit{doubly connected half-edge lists}}~\cite{MP78,BCKO08}.
Given a planar graph together with a planar embedding, this data structure can be initialized in~$O(n)$ time.
It allows to remove a link, add a link (which can be added to the planar embedding without link crossings), and find a node or link incident to a given face in constant time.
Using this data structure, one can also traverse a given face in clockwise or counterclockwise order in time linear in the size of a face, that is, the number of incident links.
With this approach, we can enumerate all nodes and links incident to a given face in the same time.
Finally, given a node~$v$ and an incident link~$e$, the data structure can return the next link incident to~$v$ in clockwise or counterclockwise order in constant time.
It is known that a planar graph with~$n$ nodes contains at most~$3n$ links.

A \textbf{\textit{link contraction}} in~$G$ is the operation of replacing a link~$\{u,v\} \in E$ by a single node~$w$, making~$w$ adjacent to all former neighbors of~$u$ and~$v$ (except possibly removing duplicate links or self-loops, that is, a link~$\{w,w\}$).
A graph~$H$ is a \textbf{\textit{minor}} of a graph~$G$ if~$H$ can be obtained from a subgraph of~$G$ by a sequence of link contractions.
Equivalently, $H$ is a minor of~$G$ if there is a connected subgraph~$G_u$ of~$G$ for each~$u \in U$ such that~$G_u$ and~$G_v$ do not share any nodes for each~$u \neq v \in U$ and for each link~$\{u,v\} \in F$, there is a link in~$E$ between a node in~$G_u$ and a node in~$G_v$.
A \textbf{\textit{rooted graph}} is a pair~$(G,\sigma)$, where~$G$ is a graph and~$\sigma$ is a sequence of nodes in~$G$ called the roots.
A rooted graph~$(H,\sigma_1=(u_1,u_2,\ldots,u_\ell))$ is a \textbf{\textit{minor}} of a rooted graph~$(G,\sigma_2=(v_1,v_2,\ldots,v_k))$ if
\begin{compactitem}
    \item $\sigma_1$ and~$\sigma_2$ have the same length, that is, $\ell = k$,
    \item there is a connected subgraph~$G_u$ of~$G$ for each~$u \in U$ such that~$G_u$ and~$G_v$ are node-disjoint for each~$u \neq v \in U$,
    \item for each link~$\{u,v\} \in F$, there is a link in~$E$ between a node in~$G_u$ and a node in~$G_v$, and
    \item $v_i$ is a node in~$G_{u_i}$ for each~$i \in [k]$.
\end{compactitem}
If the roots are clear from the context, we also simply say that~$H$ is a \textbf{\textit{rooted minor}} of~$G$.
In this work, we only work with rooted graphs with a single root.
For the sake of notational convenience, we will denote this root always by~$t$ in both the target graph~$H$ and the host graph~$G$.
We next define two special types of nodes in a rooted graph~$(G,t)$.
First, we say that all nodes in~$N(t)$ are \textbf{\textit{access nodes}}.
Second, all cut nodes in~$G-t$ are called \textbf{\textit{fracture nodes}}.
In all figures, we will color access nodes green and fracture nodes blue for easier recognition.

To conclude this subsection, we list a couple of specific graphs and introduce the two classes of rooted graphs mentioned in the introduction.
These will be important for our work later.
Therein, we sometimes make use of a single root~$t$.
All other names are interchangeable.
We denote the complete graph with~$\ell$ nodes by~$\mathbf{K_\ell}$ and the complete bipartite graph with~$a$ nodes on one side and~$b$ nodes on the other side by~$\mathbf{K_{a,b}}$.
That is, in~$K_\ell$ each pair of nodes is adjacent and in~$K_{a,b}$ there are two sets~$A$ and~$B$ of nodes (of size~$a$ and~$b$, respectively) such that there are no links between nodes in the same set and each pair~$u \in A$ and~$v \in B$ is adjacent.
The $\mathbf{a \times b}$ \textbf{\textit{grid}} is a graph with a node for each pair~$(x,y)$ with~$x \in [a]$ and~$y \in [b]$.
Two nodes with pairs~$(x_1,y_1)$ and~$(x_2,y_2)$, respectively, are adjacent if and only if (i) $x_1 = x_2$ and~$|y_1 - y_2| = 1$ or~(ii) $|x_1 - x_2| = 1$ and~$y_1 = y_2$, that is, either~$x_1 = x_2$ and~$y_1$ and~$y_2$ differ by exactly one or the other way around.

We next recall the names of the four rooted graphs depicted in~\cref{fig:obstructions} that we will show to be the minimal obstructions for perfect resilient rooted graphs.
For convenience, they are repeated in \cref{fig:obstructionsrepeat}.
\begin{figure}[t]
    \centering
    \hfill
    \begin{subfigure}{0.333\textwidth}
        \begin{tikzpicture}
            \node[ver,label=below:$p$] at(-3,5) (s) {};
            \node[te] at(-3,7) (t) {};
            \node[por,label=left:$u$] at(-5,4) (u) {} edge(s) edge(t);
            \node[por,label=right:$v$] at(-1,4) (v) {} edge(s) edge(t) edge(u);
            \node[por,label=right:$w$] at(-3,6) (w) {} edge(s) edge(t) edge(u) edge(v);
        \end{tikzpicture}
        \caption{\Kf}
    \end{subfigure}
    \hfill
    \begin{subfigure}{0.333\textwidth}
        \begin{tikzpicture}
            \node[te] at(3,7) (t) {};
            \node[por,label=left:$u$] at(1,5) (u) {} edge(t);
            \node[por,label=right:$v$] at(5,5) (v) {} edge(t);
            \node[ver,label=$p$] at(3,6) (p) {} edge(u) edge(v);
            \node[ver,label=below:$r$] at(3,4) (q) {} edge(u) edge(v);
            \node[ver,label=right:$q$] at(3,5) {} edge(p) edge(q);            
        \end{tikzpicture}
        \caption{\Kt}
    \end{subfigure}
    \hfill\hfill
    
    \medskip

    \hfill
    \begin{subfigure}{0.333\textwidth}
        \begin{tikzpicture}
            \node[te] at(3,2) (t) {};
            \node[por,label=right:$w$] at(3,1) (w) {} edge(t);
            \node[por,label=left:$u$] at(1,0) (u) {} edge(t);
            \node[por,label=right:$v$] at(5,0) (v) {} edge(t);
            \node[ter,label=below:$x$] at(3,0) {} edge(u) edge(v) edge(w);
            \node[ver,label=below:$p$] at(3,-1) {} edge(u) edge(v);
            \node[ver,label=right:$q$] at(3,-2) {} edge(u) edge(v);
        \end{tikzpicture}
        \caption{\Ktf}
    \end{subfigure}
    \hfill
    \begin{subfigure}{0.333\textwidth}
        \begin{tikzpicture}
            \node[por,label=$u$] at(-4,2) (w) {};
            \node[te] at(-2,2) (t) {} edge(w);
            \node[ter,label=left:$x$] at(-5,0) (u) {} edge(w);
            \node[por,label=right:$v$] at(-1,0) (v) {} edge(t);
            \node[ver,label=$p$] at(-3,1) {} edge(u) edge(v);
            \node[ver,label=$q$] at(-3,0) {} edge(u) edge(v);
            \node[ver,label=$r$] at(-3,-1) {} edge(u) edge(v);            
        \end{tikzpicture}
        \caption{\sK}
    \end{subfigure}
    \hfill\hfill
    \caption{The four structures inhibiting perfect resilience.}
    \label{fig:obstructionsrepeat}
\end{figure}
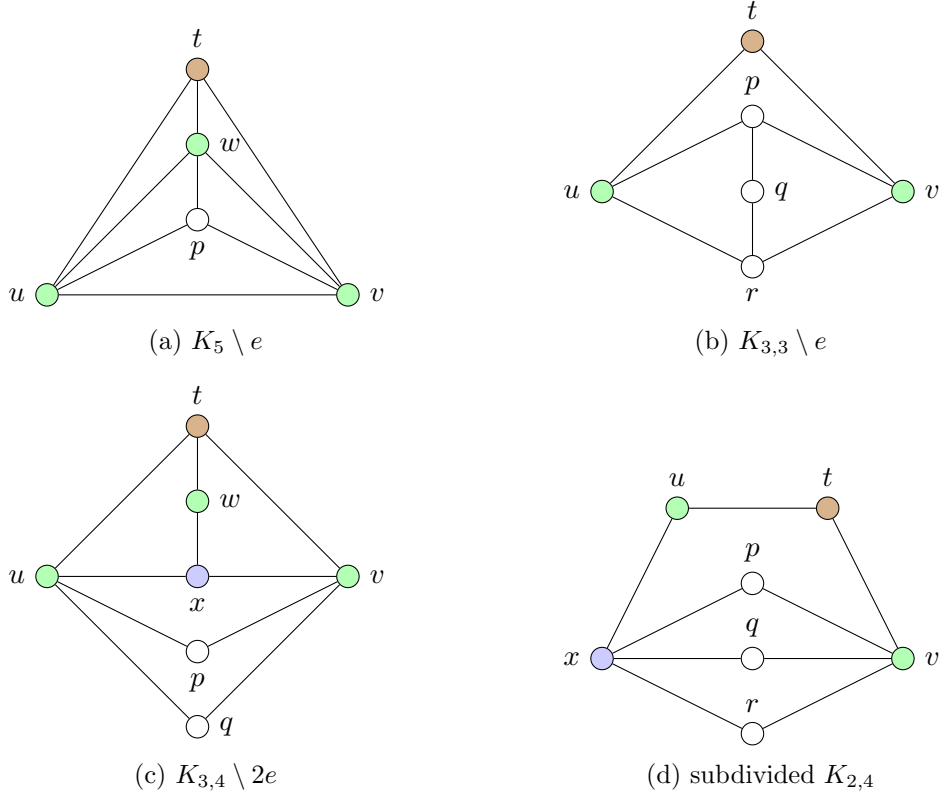%
We use $\mathbf{\Kf}$ and~$\mathbf{\Kt}$ to denote the complete graph on five nodes and the complete bipartite graph with three nodes on each side, respectively, where each time the root~$t$ is an arbitrary node in it and where one arbitrary link incident to~$t$ is removed.
Note that~$K_5$ and~$K_{3,3}$ are completely symmetric, so the choice of which node~$t$ is, is irrelevant.
We denote by~$\mathbf{\Ktf}$ the complete bipartite graph with three nodes~$\{u,v,w\}$ on side, four nodes~$\{p,q,t,x\}$ on the other side, and the two links~$\{w,p\}$ and~$\{w,q\}$ removed. 
Finally, we call the complete bipartite graph with two nodes~$v$ and~$x$ on one side, four nodes~$u,p,q$, and~$r$ on the other side, and the link~$\{u,v\}$ subdivided with the root~$t$ the \textbf{\textit{$\mathbf{\sK}$}}.

A \textbf{\textit{dipole outerplanar graph}} is a rooted graph~$(G,t)$ where~$t$ has exactly two neighbors~$u$ and~$v$.
Moreover, the graph~$G[V_i \cup \{u,v\}] \setminus \{u,v\}$ is outerplanar for each~$i\in [k]$, where~$V_1,V_2,\ldots,V_k$ are the sets of nodes of the connected components of~$G-\{u,v,t\}$.
A \textbf{\textit{ring of outerplanar graphs}} is a rooted graph~$(G,t)$ constructed as follows.
Start with a set of~$k \geq 3$ connected outerplanar graphs~$G_1,G_2,\ldots,G_k$.
In each graph~$G_i$, select two nodes~$u_i$ and~$v_i$.
Then, identify~$v_i$ with~$u_{i+1}$ for each~$i \in [k-1]$ and identify~$v_k$ with~$u_1$.
Now add~$t$ and make it adjacent to all nodes~$u_i$ with~$i \in [k]$.
An example of a ring of outerplanar graphs is shown in \cref{fig:ring}.
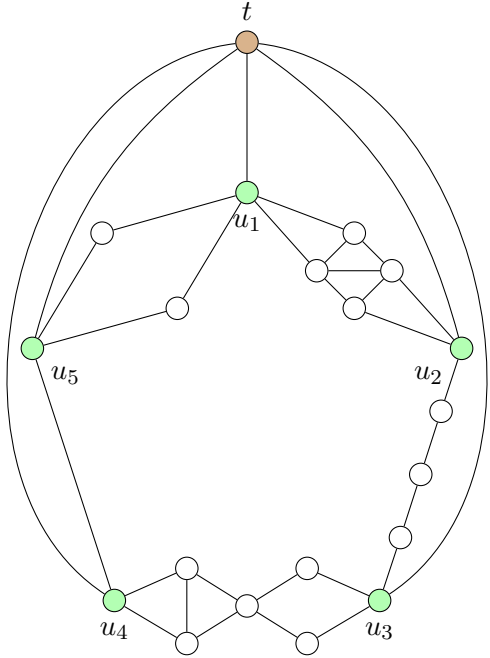
\begin{figure}
    \centering
    \begin{tikzpicture}
        \foreach \i/\j in {1/below,2/below left,3/below,4/below,5/below right}{
            \node[por, label=\j:$u_\i$] at (-72*\i-198:3cm) (u\i) {};
        }
        \node[te] at(0,5) {} edge(u1) edge[bend left=20](u2) edge[bend left=70](u3) edge[bend right=70](u4) edge[bend right=20](u5);
        
        \node[ver] at(1.9265,1.96) (p) {} edge(u2);
        \node[ver] at(1.4265,1.46) (q) {} edge(u2) edge(p);
        \node[ver] at(0.9265,1.96) (r) {} edge(u1) edge(p) edge(q);
        \node[ver] at(1.4265,2.46) (s) {} edge(u1) edge(p) edge(r);

        \node[ver] at(2.58,.09) (x) {} edge(u2);
        \node[ver] at(2.31,-.75) (y) {} edge(x);
        \node[ver] at(2.04,-1.59) (z) {} edge(u3) edge(y);
        
        \node[ver] at(-.8,-2) (a) {} edge(u4);
        \node[ver] at(-.8,-3) (b) {} edge(u4) edge(a);
        \node[ver] at(0,-2.5) (c) {} edge(a) edge(b);
        \node[ver] at(.8,-2) (d) {} edge(c) edge(u3);
        \node[ver] at(.8,-3) (e) {} edge(c) edge(u3);
        
        \draw (u4) to (u5);
        
        \node[ver] at(-.9265,1.46) (q) {} edge(u1) edge(u5);
        \node[ver] at(-1.9265,2.46) (q) {} edge(u1) edge(u5);
    \end{tikzpicture}
    \caption{An example of a ring of outerplanar graphs.}
    \label{fig:ring}
\end{figure}

\subsection{Sequences, Skipping Forwarding Functions, and Traps}
Let~$A=\{a_1,a_2,\ldots,a_\ell\}$ be a set.
We say a sequence~$\sigma=(a_1,a_2,\ldots,a_\ell)$ is a permutation of~$A$ if each element in~$A$ appears exactly once in~$\sigma$.
We also say that~$A$ is the set underlying~$\sigma$.
We denote the \textbf{\textit{concatenation}} of two sequences~$\sigma = (a_1,a_2,\ldots,a_\ell)$ and~$\sigma'=(b_1,b_2,\ldots,b_k)$ by~${\sigma \circ \sigma' = (a_1,a_2,\ldots,a_\ell,b_1,b_2,\ldots,b_k)}$.
If~$a_\ell = b_1$, then we say that~$(a_1,a_2,\ldots,a_\ell,b_2,b_3,\ldots,b_k)$ is the \textbf{\textit{gluing}} of~$\sigma$ and~$\sigma'$.
We will assume that sequences are implemented as doubly linked lists so that inserting a new element at some position takes constant time.

A rooted graph for which a perfectly resilient forwarding pattern exists is called \textbf{\textit{perfectly resilient}}.
Otherwise, it is called a \textbf{\textit{trap}}.
A trap is \textbf{\textit{minimal}} if it does not contain another trap as a rooted minor.

A \textbf{\textit{skipping priority list}} for a node~$v$ is a function~$\pi_v$ that takes as input a link in~$E_v$ (or~$\bot$ modeling that the package starts in~$v$) and outputs a permutation of~$E_v$.
The routing then chooses for the given in-port the out-port as the first link in the permutation that does not failed (or~$\bot$ if all incident links fail).
A \textbf{\textit{skipping forwarding pattern}} is a collection of skipping priority lists for each node~$v$.
For the sake of notational convenience, for a given node~$v$, we do not distinguish between its incident links and the corresponding adjacent nodes.
So instead of writing~$\ski{v}{\{v,u\}}=(\{v,t\},\{v,w\},\{v,u\})$, we simply write~$\ski{v}{u}=(t,w,u)$.
Moreover, for a given node~$v$ and a set~$F$ of failed links, we use~$F_v = E_v \cap F$ to denote the failed links incident to~$v$.
For a skipping priority list~$\pi_v$, an in-port~$e$, and a set~$F_v$ of incident failed links, we write~$\rouns{v}{F_v}{e}$ for the link chosen next by~$\pi_v$ for in-port~$e$ if the incident links~$F_v$ fail.
We again do not distinguish between incident links and adjacent nodes, so we write e.g.\ $\rou{v}{t,u,w}{x}=y$.
An example of skipping priority lists and the resulting routing is given in \cref{fig:skipping}.
\begin{figure}[t]
    \centering
    \begin{tikzpicture}
        \node[te] at(0,-2) (t) {};
        \node[por,label=$u$] at(-1,-3) (u) {} edge[red](t);
        \node[por,label=$v$] at(1,-3) {} edge(u) edge(t);
    \node at(6,-2.5) {\begin{tabular}{c|c|c}
        node & in-port & priority list \\\hline
        $u$ & $t$ & $(t,v)$\\
        $u$ & $v$ & $(t,v)$\\
        $u$ & $\bot$ & $(t,v)$\\
        $v$ & $t$ & $(t,u)$\\
        $v$ & $u$ & $(t,u)$\\
        $v$ & $\bot$ & $(u,t)$
    \end{tabular}};
    \end{tikzpicture}
    \caption{A rooted graph with two access nodes~$u$ and~$v$ (green). The skipping priority lists are depicted on the right. If~$s=v$ is the starting node and~$F=\{u,t\}$ (red link), then the resulting routing is~$(v,u,v,t)$ as in the first step~$\roue{v}{\bot}=u$, then~$\rou{u}{t}{v}=v$, and finally~$\roue{v}{u}=t$.}
    \label{fig:skipping}
\end{figure}
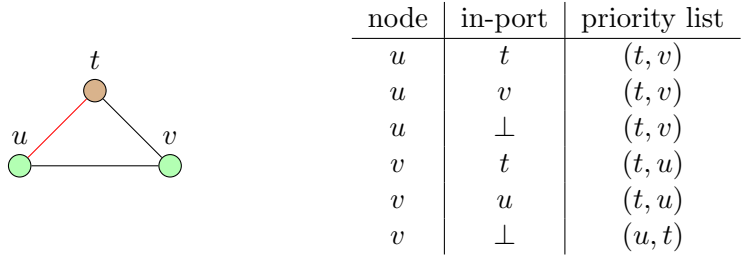
Given a node~$v$ in a rooted graph~$(G=(V,E),t)$ and a set~$F_v$ of incident failed links, we say that a node~$u$ is an \textbf{\textit{active neighbor}} with respect to~$F_v$ is~$\{u,v\} \in E \setminus F_v$.
An active neighbor of~$v$ is called \textbf{\textit{relevant}} with respect to~$F_v$ if there exists a path from~$u$ to~$t$ in~$G$ that does not contain any other active neighbor of~$v$.

We conclude with an important subclass of skipping priority lists called the \textbf{\textit{right-hand rule}}, which we denote by~$\lambda'_v$.
The right-hand rule works as follows.
Given a plane graph and a node~$v$, the skipping priority list lists for an in-port~$e \in E_v$ all links incident to~$v$ in counterclockwise fashion starting from the link next to~$e$ (such that~$e$ is the last link in the list).
An example is given in \cref{fig:righthand}.
A slight adaptation of the right-hand rule (which was also considered before) is to move the target root~$t$ to the front of every priority list for every access node.
We denote this updated right-hand rule for a node~$v$ by~$\lambda_v$.
We note that right-hand rules and updated right-hand rules are not uniquely defined for a given planar embedding as the output for~$\bot$ has to be specified.
To this end, we will pick a link~$e$ incident to~$v$ and write~$\lambda_v^e$ to denote the updated right-hand rule where the priority list for~$\bot$ is equal to the priority list for in-port~$e$.
A forwarding pattern where all forwarding functions are updated right-hand rules is called a \textbf{\textit{right-hand forwarding pattern}}. 
\begin{figure}[t]
    \centering
    \begin{tikzpicture}
        \node[ver] at(0,0) (t) {};
        \draw[blue,thick] (0,-1.4) to node [midway,right] {\color{black}\raisebox{.5pt}{\textcircled{\raisebox{-.9pt} {4}}}} (t);
        \draw (1,1) to node [midway,right] {\raisebox{.5pt}{\textcircled{\raisebox{-.9pt} {1}}}} (t);
        \draw (-1,1) to node [midway,above] {\color{black}\raisebox{.5pt}{\textcircled{\raisebox{-.9pt} {2}}}} (t);
        \draw (-1.4,0) to node [midway,below] {\raisebox{.5pt}{\textcircled{\raisebox{-.9pt} {3}}}} (t);
    \end{tikzpicture}
    \caption{A node and one of the incident links (blue) chosen as in-port. The numbering of the incident links shows the right-hand rule~$\lambda'$ for this node.}
    \label{fig:righthand}
\end{figure}
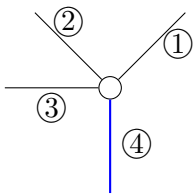

\section{New Fundamental Structures Supporting Perfect Resilience}
\label{sec:newpr}
In this section, we show that dipole outerplanar graphs and rings of outerplanar graphs are perfectly resilient.
Moreover, we will show that the updated right-hand rule~$\lambda_v^{e_v}$ for each node~$v$ and some chosen link~$e_v$ for~$v$ will be perfectly resilient in both cases.
This will later be important to construct perfectly resilient skipping forwarding patterns for general perfectly resilient rooted graphs.
We start with dipole outerplanar graphs.
Recall that a rooted graph~$(G,t)$ is dipole outer\-planar if the following two conditions hold.
First, the root~$t$ has exactly two neighbors~$u$ and~$v$.
Second,~${G[V_i \cup \{u,v\}] \setminus \{u,v\}}$ is outerplanar for each~$i\in [k]$, where~$V_1,V_2,\ldots,V_k$ are the sets of nodes of the connected components of~$G-\{u,v,t\}$.

Intuitively, to show that these rooted graphs are perfectly resilient, we compute an outerplanar embedding for each of the induced subgraphs, then ``stack'' these graphs on top of each other, and then argue that traversing an outer face of any of these graphs yields a solution.
For an intuitive example, consider the graph in \cref{fig:merge}.
If we for example start in the face incident to~$u,v,r,s$, and~$w$, and going ``in clockwise direction'', then starting from~$u$, we reach~$v$ unless the outerplanar graph containing~$u,v,p,q,r$, and~$s$ is disconnected by link failures.
If for example~$\{r,s\},\{p,q\},\{u,v\}$, and~$\{u,t\}$ fail, then we traverse the new bigger face in the order~$(u,r,p,u,w,v,t)$.
The main idea is that the next graph~$G_i$ leads to both~$u$ or~$v$ or is disconnected.
If it is disconnected, then we traverse along the outer face and reach the next component.
Finally, since we always assume that some path from the source to~$t$ remains, not all components can be disconnected.

\begin{proposition}
    \label{prop:dipole}
    Let~$(G,t)$ be a dipole outerplanar graph.
    Then, we can compute in~$O(n)$ time a planar embedding of~$G$ and a link~$e_v \in E_v$ for each node~$v$ such that~$(\lambda_v^{e_v})_{v \in V \setminus \{t\}}$ is a perfectly resilient forwarding pattern for~$(G,t)$.
\end{proposition}

\begin{proof}
    Let~$u$ and~$v$ be the two neighbors of~$t$.
    We first compute all connected components in~${G - \{u,v,t\}}$ in~$O(n)$ time.
    Let~$V_1,V_2,\ldots,V_k$ be the sets of nodes of these connected components.
    Let~$G_i = G[V_i \cup \{u,t\}] \setminus \{u,v,t\}$ for each~$i \in [k]$.
    By assumption, each graph~$G_i$ is outerplanar.
    We can therefore compute an outerplanar embedding for it in~$O(|G_i|)$ time using \cref{cor:outerplanar}.
    We now traverse the outer face in this embedding in clockwise order (also in~$O(|G_i|)$ time) and for each node~$w$ except for~$u$ and~$v$, we set~$e_w$ to be the link entering~$w$ in this traversal.
    Let~$\sigma^u_i$ be the sequence of incident links of~$u$ in the computed embedding of~$G_i$, when going counterclockwise around~$u$ and starting at the outer face.
    Let~$\sigma^v_i$ be similarly defined for~$v$ but going in clockwise direction.
    Note that we can compute these sequences in~$O(|G_i|)$ time.

    In the next step, we combine the different outerplanar embeddings.
    To this end, we start with the embedding of~$G_1$ and then iteratively add the next graph~$G_i$ such that when going counterclockwise around~$u$, the first node in~$\sigma^i_u$ comes directly after the last node in~$\sigma_u^{i-1}$ and when going clockwise around~$v$, the first node of~$\sigma_v^{i}$ comes directly after the last node of~$\sigma_v^{i-1}$.
    See \cref{fig:merge} for an illustration.
    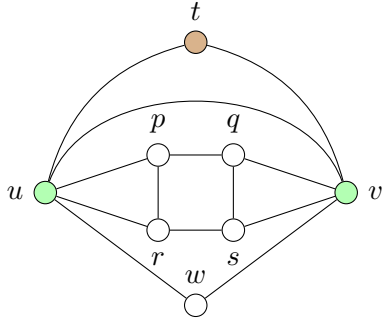
\begin{figure}
        \centering
        \begin{tikzpicture}
            \node[por,label=left:$u$] at(-2,0) (u) {};
            \node[por,label=right:$v$] at(2,0) (v) {} edge[bend right=70] (u);
            \node[ver,label=$w$] at(0,-1.5) (w) {} edge(u) edge(v);
            \node[ver,label=$p$] at(-.5,.5) (p) {} edge(u);
            \node[ver,label=$q$] at(.5,.5) (q) {} edge(v) edge(p);
            \node[ver,label=below:$r$] at(-.5,-.5) (r) {} edge(u) edge(p);
            \node[ver,label=below:$s$] at(.5,-.5) (s) {} edge(q) edge(r) edge(v);
            \node[te] at(0,2) {} edge[bend right=30](u) edge[bend left=30](v);
        \end{tikzpicture}
        \caption{An example of a dipole outerplanar graph. The sets of nodes of connected components of~$G - \{u,v,t\}$ are~$V_1 = \{w\}$ and~$V_2=\{p,q,r,s\}$. In the chosen embedding, $\sigma_u^1=\sigma_v^1=(w)$, $\sigma_u^2 = (r,p)$, and~$\sigma_v^2 = (s,q)$. Note that when going counterclockwise around~$u$ starting at the outer face, we encounter first~$w$, then~$r$, and then~$p$. When going clockwise around~$v$ starting on the outer face, we encounter~$w$ first, then~$s$, and then~$q$.}
        \label{fig:merge}
    \end{figure}
    To conclude the construction, if the link~$\{u,v\}$ exists in~$G$, then we place it ``after''~$G_k$ and~$t$ with its two incident links is placed after that.
    For~$u$, we choose~$e_u = \{u,t\}$ and for~$v$, we choose~$e_v=\{w,v\}$, where~$w$ is the first node in~$\sigma_v^1$.

    Before we show that~$(\lambda_w^{e_w})_{w \in V \setminus \{t\}}$ for the computed embedding and the chosen links~$e_w$ is perfectly resilient for~$(G,t)$, we first analyze the running time.
    Note that the combining step of the computed embeddings take~$O(|N(u)|+|N(v)|)$ time as we simply place all links incident to~$u$ and~$v$ in the order given by~$\sigma_u^i$ and~$\sigma_v^i$ for all~$i$.
    Moreover, whenever we add a new graph~$G_i$, we create one new face between~$G_i$ and~$G_{i-1}$ which takes constant time in the initialization of the doubly linked half-edge data structure.
    Since the running time for each individual graph~$G_i$ is in~$O(|G_i|)$ as analyzed above, this yields a total running time of~$O(\sum_{i=1}^k|G_i|) \subseteq O(n+m) = O(n)$.

    It remains to show that~$(\lambda_w^{e_w})_{w \in V \setminus \{t\}}$ for the computed embedding and the chosen links~$e_w$ is perfectly resilient for~$(G,t)$.
    To this end, consider an arbitrary starting node~$s$ and a set~$F$ of failed links such that an~$s$-$t$-path remains in~$G \setminus F$.
    Due to the choice of~$e_s$, the routing~$\rho$ corresponding to~$\lambda$, $s$, and~$F$ traverses either along the outer face or along one of the faces ``between'' two components~$G_i$ and~$G_{i+1}$ for some~$i \in [k-1]$.
    Since a path towards~$t$ remains, the routing~$\rho$ contains~$u$ or~$v$.
    We assume without loss of generality that it contains~$u$ first (as the other case is completely symmetric).
    If the link~$\{u,t\}$ does not fail, then~$t$ is reached in the next step due to the definition of the updated right-hand rule~$\lambda$.
    So assume that~$\{u,t\}$ fails and note that this means that~$\{v,t\}$ does not fail.
    Moreover, since the outer face and any face between two components~$G_i$ and~$G_{i+1}$ contain both~$u$ and~$v$, after removing the links in~$F$, it is still true that the new (potentially larger) face is incident to both~$u$ and~$v$ and there is still a path connecting~$u$ and~$v$ along this face (as otherwise there would be no path between~$u$ and~$v$ and therefore~$s$ and~$t$ would be disconnected).
    Hence~$\rho$ also contains~$v$ and since~$\{v,t\}$ does not fail, $\rho$ also contains~$t$.
\end{proof}

The second class of rooted graphs that we investigate in this section are rings of outerplanar graphs.
Recall that a rooted graph~$(G,t)$ is a ring of outerplanar if it can be constructed from chaining together a set of~$k \geq 3$ connected outerplanar graphs~$G_1,G_2,\ldots,G_k$ in a ring and making all of the intersections neighbors of~$t$.

\begin{proposition}
    \label{prop:ring}
    Let~$(G,t)$ be a ring of outerplanar graphs.
    Then, we can compute in~$O(n)$ time a planar embedding of~$G$ and a link~$e_v \in E_v$ for each node~$v$ such that~$(\lambda_v^{e_v})_{v \in V \setminus \{t\}}$ is a perfectly resilient forwarding pattern for~$(G,t)$.
\end{proposition}

\begin{proof}
    Let~$U= \{u_1,u_2,\ldots,u_k\}$ be the neighbors of~$t$ in~$G$.
    Let~$G_i$ be the outerplanar graph connecting~$u_i$ with~$u_{i+1}$ (where~$u_{k+1} = u_k$).
    Note that we can compute the ordering within~$U$ and each graph~$G_i$ in overall~$O(n)$ time by first computing the (node sets~$V_1,V_2,\ldots,V_\ell$ of the) connected components~$C$ of~$G-(N(t) \cup \{t\})$, then checking for each node~$u\in U$ to which connected components they connect to, and then computing~$G_i = G[\{u_i,u_{i+1}\} \cup \{v \in V_j \mid j \in K_i\}]$, where~$K_i$ is the set of indices~$j$ such that the connected component corresponding to~$V_j$ connects to~$u_i$ and~$u_{i+1}$.
    Note that~$|K_i|=2$ is possible if~$G_i$ is a cycle.
    Next, since each graph~$G_i$ is by assumption outerplanar, we can compute an outerplanar embedding of it in~$O(|G_i|)$ time by \cref{cor:outerplanar}.
    We now traverse the outer face in this embedding in clockwise order (also in~$O(|G_i|)$ time) and for each node~$v$ except for~$u_i$ and~$u_{i+1}$, we set~$e_v$ to be the link entering~$v$ in this traversal.
    
    We next combine the different outerplanar embeddings.
    To this end, for each node~$u_i$, we place all incident links to neighbors in~$G_{i-1}$ in the computed order for the embedding for~$G_{i-1}$ and then all links to neighbors in~$G_i$ in the order corresponding to the computed embedding for~$G_i$.
    All other nodes keep their incident link structure from the embedding of their respective outerplanar graph.
    We then traverse around the outer face of the current graph~$G-t$ in clockwise order and for each node~$u_i$, we set~$e_{u_i}$ to be the link through which~$u_i$ is entered in this traversal.
    Note that each node~$u_i$ is on the outer face by construction.
    To conclude the construction, we add~$t$ and make it adjacent to all nodes in~$U$.
    Since all nodes~$u_i$ belong to the outer face, the constructed embedding is plane.
    Moreover, when starting at the link~$\{u_1,t\}$ and walking around~$t$ in counterclockwise order, the links incident to~$t$ appear in the order~$(\{u_1,t\},\{u_2,t\},\ldots,\{u_k,t\})$.
    
    The running time of the construction is~$O(n)$ as each link is considered for only one graph~$G_i$ and each node is considered at most twice.
    Hence~$\sum_{i=1}^k O(|G_i|) = O(n)$.
    Adding all links incident to~$t$ also takes at most~$O(n)$ time as adding one link takes constant time.
    Moreover, note that the faces of~$G-t$ can be characterized as follows.
    Each internal face of a graph~$G_i$ appears exactly the same in~$G-t$.
    The outer face of each graph~$G_i$ can be expressed as the gluing of two walks between~$u_i$ and~$u_{i+1}$.
    The outer face of~$G-t$ is the gluing of exactly one of these walks for each~$G_i$.
    Moreover, there is a new internal face~$F^*$ that is the gluing of the remaining walks.

    It remains to show that~$(\lambda_w^{e_w})_{w \in V \setminus \{t\}}$ for the computed embedding and the chosen links~$e_w$ is perfectly resilient for~$(G,t)$.
    To this end, consider an arbitrary starting node~$s$ and a set~$F$ of failed links such that an~$s$-$t$-path remains in~$G \setminus F$.
    Due to the choice of~$e_s$, the routing~$\rho$ corresponding to~$\lambda$, $s$, and~$F$ traverses either along the outer face or along the face~$F^*$.
    If at least one graph~$G_i\setminus F$ is disconnected, then both faces are the same.
    In either case, the updated right-hand rule traverses along the face (in~$G-t$) until an active neighbor of~$t$ is reached at which point, $t$ is the next node in the routing.
    Since the face is incident to all nodes in~$U$, the only possibility that~$t$ is not reached is if all links incident to~$t$ fail.
    However, this contradicts the assumption that an~$s$-$t$-path remains in~$G \setminus F$.
    This concludes the proof.
\end{proof}

\section{New Fundamental Structures Inhibiting Perfect Resilience}
\label{sec:notpr}
In this section, we show that the~\sK{} and the~\Ktf{} are not perfectly resilient.
These will be the basis for our characterization of perfectly resilient rooted graphs later on.
For reference, the two rooted graphs are depicted in~\cref{fig:nonres}, together with the naming of nodes we will use in this section.
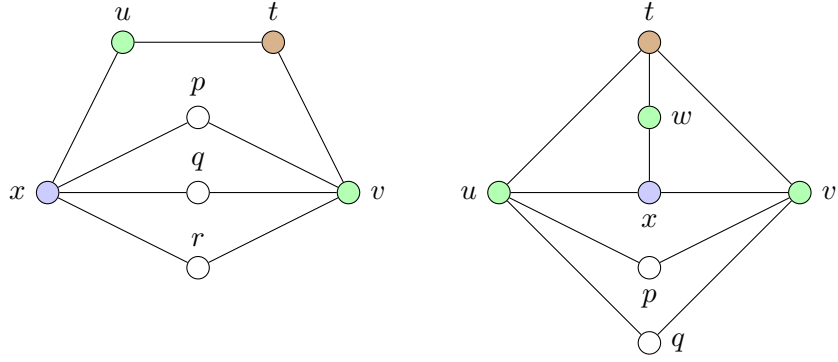
\begin{figure}[t]
    \centering
    \begin{tikzpicture}
        \node[por,label=$u$] at(-4,2) (w) {};
        \node[te] at(-2,2) (t) {} edge(w);
        \node[ter,label=left:$x$] at(-5,0) (u) {} edge(w);
        \node[por,label=right:$v$] at(-1,0) (v) {} edge(t);
        \node[ver,label=$p$] at(-3,1) {} edge(u) edge(v);
        \node[ver,label=$q$] at(-3,0) {} edge(u) edge(v);
        \node[ver,label=$r$] at(-3,-1) {} edge(u) edge(v);

        \node[te] at(3,2) (t) {};
        \node[por,label=right:$w$] at(3,1) (w) {} edge(t);
        \node[por,label=left:$u$] at(1,0) (u) {} edge(t);
        \node[por,label=right:$v$] at(5,0) (v) {} edge(t);
        \node[ter,label=below:$x$] at(3,0) {} edge(u) edge(v) edge(w);
        \node[ver,label=below:$p$] at(3,-1) {} edge(u) edge(v);
        \node[ver,label=right:$q$] at(3,-2) {} edge(u) edge(v);
    \end{tikzpicture}
    \caption{The \sK{} (left) and \Ktf{} (right).}
    \label{fig:nonres}
\end{figure}
We start with the \sK.

\begin{proposition}
    \label{prop:sk}
    The \sK{} is not perfectly resilient.
\end{proposition}

\begin{proof}
    Assume towards a contradiction that~$\pi$ is a perfectly resilient forwarding pattern for the $\sK = (G,t)$.
    Note that for~$F_x = \emptyset$, all neighbors of~$x$ ($u$, $p$, $q$, and~$r$) are relevant for~$x$.
    Hence, $\pi_x(u,\emptyset) \neq u$ by \cref{lem:orbit}.
    Since~$p$, $q$, and~$r$ are all equivalent up to isomorphism, we assume without loss of generality that~$\pi_x(u,\emptyset)=p$.
    Again by \cref{lem:orbit}, $\pi_x(p,\emptyset) \notin \{u,p\}$.
    Again due to the symmetry of~$q$ and~$r$, we assume that~$\pi_x(p,\emptyset) = q$.
    The same argument also shows that we may assume~$\pi_x(q,\emptyset)=r$ and~$\pi_x(r,\emptyset)=u$.
    We will argue that for the starting node~$s=q$ and the set~$F=\{\{v,t\},\{q,v\}\}$ of failed links, the resulting routing~$\rho$ for~$\pi$, $s$, and~$F$ does not contain~$t$, contradicting the assumption that~$\pi$ is perfectly resilient as the path~$(q,x,u,t)$ from~$s=q$ to~$t$ remains in~$G \setminus F$.
    Note that since~$x$ is the only neighbor of~$q$ in~$G \setminus F$, it holds that~$\pi_q(f,\{v\}) = x$ for each in-port~$f$.
    Moreover, by \cref{lem:orbit} and the fact that for the set~$F_v = \{\{v,q\},\{v,t\}\}$ of failed links incident to~$v$ both~$p$ and~$r$ are relevant neighbors for~$v$, it holds that~$\pi_v(r,\{q,t\}) = p$.
    Similarly, $\pi_p(v,\emptyset) = x$ and~$\pi_r(x,\emptyset) = v$.
    Thus, $\rho = (q,x,r,v,p,x,q,x,\ldots)$.
    Note that the routing~$\rho$ repeats the subsequence~$(q,x)$ at this point and hence the routing will continue indefinitely with the same pattern not containing~$t$.
    This concludes the proof.
\end{proof}

We continue with the \Ktf.

\begin{proposition}
    \label{prop:ktf}
    The \Ktf{} is not perfectly resilient.
\end{proposition}

\begin{proof}
    The proof is similar to that of \cref{prop:sk}.
    Assume towards a contradiction that~$\pi$ is a perfectly resilient forwarding pattern for~$\Ktf = (G,t)$.
    Note that for~$F_x = \emptyset$, all neighbors of~$x$ ($u$, $v$, and~$w$) are relevant for~$x$.
    Hence, $\pi_x(w,\emptyset) \neq w$ by \cref{lem:orbit}.
    Since~$u$, and~$v$ are equivalent up to isomorphism, we assume without loss of generality that~$\pi_x(w,\emptyset)=u$.
    The same argument shows that~$\pi_x(u,\emptyset) = v$ and~$\pi_x(v,\emptyset) = w$.
    Similarly, for the set~$F_u=\{u,t\}$ of failed links incident to~$u$, all three remaining active neighbors ($x$, $p$, and~$q$) are relevant for~$u$.
    Hence, $\rou{u}{t}{x} \neq x$ by \cref{lem:orbit} and since~$p$ and~$q$ are equivalent up to isomorphism, we assume without loss of generality that~$\rou{u}{t}{x} = p$.
    By \cref{lem:orbit}, this shows~$\rou{u}{t}{p}=q$ and~$\rou{u}{t}{q}=x$.
    We will argue that for the starting node~$s=q$ and the set~$F=\{\{u,t\},\{v,t\},\{q,v\}\}$ of failed links, the resulting routing~$\rho$ for~$\pi$, $s$, and~$F$ does not contain~$t$, contradicting the assumption that~$\pi$ is perfectly resilient as the path~$(q,u,x,w,t)$ from~$s=q$ to~$t$ remains in~$G \setminus F$.
    Note that since~$u$ is the only neighbor of~$q$ in~$G \setminus F$, it holds that~$\pi_q(f,\{v\}) = u$ for each in-port~$f$.
    Moreover, by \cref{lem:orbit} and the fact that for the set~$F_v = \{\{v,q\},\{v,t\}\}$ of failed links incident to~$v$ both~$p$ and~$x$ are relevant neighbors for~$v$, it holds that~$\rou{v}{q,t}{x} = p$.
    Similarly, $\roue{p}{v} = u$.
    Thus, $\rho = (q,u,x,v,p,u,q,u,\ldots)$.
    Note that the routing~$\rho$ repeats the subsequence~$(q,u)$ at this point and hence the routing will continue indefinitely with the same pattern not containing~$t$.
    This concludes the proof.
\end{proof}

To conclude this section, we mention that the choice of~$t$ matters in all four minimal traps we consider in this paper.
For the~\Kf{} and the~\Kt, note that if~$t$ is not incident to the missing link, then~$G-t$ is outerplanar in all cases (easily seen by considering the node~$u$ in \cref{fig:obstructionsrepeat}).
Hence, the resulting rooted graph is perfectly resilient by \cref{thm:gt}.
For the \sK, note that~$u$ (in \cref{fig:nonres}) is symmetric to~$t$ up to isomorphism.
For~$t \in \{x,v\}$, the graph~$G-t$ is again outerplanar.
For~$t \in \{p,q,r\}$, the resulting graph is a dipole outerplanar graphs and therefore perfectly resilient by \cref{prop:dipole}.
Finally, for the~\Ktf, $x$ (again in \cref{fig:nonres}) is symmetric to~$t$.
For~$t \in \{u,v\}$, the graph~$G-t$ is outerplanar and therefore perfectly resilient by \cref{thm:gt}.
For~$t \in \{p,q,w\}$, the resulting rooted graph is also not perfectly resilient as~\Kt{} is a rooted minor in this case.

\section{Preprocessing: Simplifying Structure}
\label{sec:pre}
In this section, we present a number of preprocessing routines to simplify the structure of the input graph by constructing an equivalent instance in terms of perfect resilience with the following additional properties.
We will first show in \cref{sec:pre1} how to preprocess the graph so that the result is a connected planar subgraph.
Afterwards, we will show in \cref{sec:pre2} how to ensure that the graph is biconnected.
These preprocessing routines will be relatively simple.
We conclude in \cref{sec:pre3} with a much more technically involved preprocessing routine to deal with separating links.
In order to do so, we show a couple of auxiliary results regarding separating links in \cref{sec:prehelp1} and define a new kind of planar embedding in \cref{sec:prehelp2}.
We also prove some lemmas regarding these embeddings and a relation to an important subclass of skipping forwarding patterns that we also define in \cref{sec:prehelp2}.
These preprocessing routines and the additional structural properties they provide will help us later in characterizing perfectly resilient rooted graphs in terms of forbidden rooted minors.
For the sake of readability, we will first present each of the preprocessing routines for \prd{} and afterwards describe how to generalize it for \prs, that is, how to construct perfectly resilient forwarding patterns for the original instance given a solution for the instance after the preprocessing.

\subsection{Ensuring Planarity and Connectivity}
\label{sec:pre1}
We start with the simple observation that we can assume without loss of generality that the input graph is connected.
To this end, consider any connected components that does not contain the root~$t$.
Note that if the routing starts in this component, then it can never reach~$t$.
This immediately yields the following observation.
\begin{observation}
    \label{obs:connected}
    Let~$(G,t)$ be a rooted graph and let~$G'$ be the connected component of~$G$ containing~$t$.
    Then, $(G,t)$ is perfectly resilient if and only if~$(G',t)$ is.
\end{observation}

Before we describe how to use this observation algorithmically for \prd{} and \prs, we first present a second result.
This should be seen as a small fix of the claim that any non-planar graph cannot be perfectly resilient~\cite{FHPST21} (which we showed in \cref{fig:planar} to be incorrect).

\begin{proposition}
    \label{prop:planar}
    Let~$(G,t)$ be a rooted graph where~$G$ is connected.
    If~$G$ is not planar, then~$(G,t)$ is not perfectly resilient.
\end{proposition}

\begin{proof}
    Assume towards a contradiction that~$G$ is a connected non-planar graph and~$(G,t)$ is perfectly resilient.
    Since it is not planar, it contains~$K_5$ or~$K_{3,3}$ as a minor by \cref{thm:wagner}.
    Hence, there are (i) five node-disjoint subgraphs~$G_1,G_2,G_3,G_4,$ and~$G_5$ of~$G$ with at least one link between each pair of these five subgraphs or (ii) six subgraphs~$G_1,G_2,G_3,G_4,G_5,$ and~$G_6$ of~$G$ and for each~$i \in \{1,2,3\}$ and each~$j \in \{4,5,6\}$ there is at least one link between a node in~$G_i$ and a node in~$G_j$.
    We will show next that we can assume without loss of generality that~$t$ is contained in one of the five or six subgraphs.
    If this is not the case, then pick an arbitrary node~$v$ in~$G_1$.
    Since~$G$ is connected, there is a $t$-$v$-path~$P=(t,u_1,u_2,\ldots,u_k,v)$ in~$G$.
    Let~$w$ be the first node in~$P$ that is contained in one of the subgraphs.
    By assumption,~$w \neq t$, but~$w=v$ is possible.
    Let~$P'=(t,u_1,u_2,\ldots,w)$ be the beginning of~$P$ up to node~$w$.
    We will assume without loss of generality that~$w$ is contained in~$G_1$.
    The other cases are completely symmetric.
    By contracting all links in~$P'$, we obtain a new subgraph~$G'_1$ that is still node-disjoint from all other subgraphs and still has the above listed links towards other subgraphs.
    Thus, $(G,t)$ contains~$K_5$ or~$K_{3,3}$ as a minor where~$t$ is one of the nodes of the minor.
    By \cref{prop:k5k33notpr,cor:rootedminor}, the rooted graph~$(G,t)$ is not perfectly resilient, a contradiction.
    This concludes the proof.
\end{proof}

To conclude this subsection, we show how to apply the above results algorithmically.
For \prd{} the application follows directly from \cref{prop:embedding}.
We can test in~$O(n)$ time whether the connected component containing~$t$ is planar.
If it is not, then we can immediately conclude that the answer is no due to \cref{prop:planar}.
Otherwise, we can compute the connected component and a planar embedding of it.

\begin{corollary}
    \label{cor:pre1dec}
    Given a rooted graph~$(G,t)$, we can verify in~$O(n)$ time that~$(G,t)$ is not perfectly resilient or compute a rooted graph~$(G',t)$ such that~$(G,t)$ is perfectly resilient if and only if~$(G',t)$ is and where~$G'$ is a planar connected subgraph of~$G$.
    We can also compute a planar embedding of~$G'$ in the same time.
\end{corollary}

For \prs, we do exactly the same but in order to return a perfectly resilient forwarding pattern, we compute arbitrary skipping priority lists for each node in a different connected component than~$t$.
Note that there are~$O(m)$ possible pairs of a node and one of its in-ports by the handshaking lemma\footnote{The handshaking lemma states that~$\sum_{u\in V} |N(u)| = 2m$ \cite{Eul36}.} and the priority list has size~$O(n)$.
Since we cannot assume the number of links to be linear in the number of nodes in other connected components, the preprocessing for \prs{} takes~$O(nm)$ time.

\begin{corollary}
    \label{cor:pre1syn}
    Given a rooted graph~$(G,t)$, we can compute in~$O(n)$ time a rooted graph~$(G',t)$ where~$G'$ is planar and connected such that given a perfectly resilient (skipping) forwarding pattern for~$(G',t)$, we can compute in~$O(nm)$~time a perfectly resilient (skipping) forwarding pattern for~$(G,t)$.
\end{corollary}

\subsection{Removing Cut Nodes}
\label{sec:pre2}
We continue with a preprocessing routine to remove cut nodes.
For \prd, we simply compute all biconnected components and solve them individually.
\begin{proposition}
    \label{prop:pre2dec}
    Let~$(G,t)$ be a rooted graph where~$G$ is planar and connected.
    Let~$v$ be a cut node in~$G$ and let~$V_1,V_2,\ldots,V_k$ be the set of nodes in the connected components in~$G-v$ where~${t \in V_1}$ if~$v \neq t$.
    Let~$G_i = G[V_i \cup \{v\}]$ for each~$i \in [k]$.
    Then, $(G,t)$ is perfectly resilient if and only if all of~$(G_1,t), (G_2,v), (G_3,v), \ldots, (G_k,v)$ are.
\end{proposition}

\begin{proof}
    We start by showing the forward direction, that is, we show that if~$(G,t)$ is perfectly resilient, then so are~$(G_1,t), (G_2,v), \ldots, (G_k,v)$.
    This follows directly from \cref{cor:rootedminor}.
    For~$(G_1,t)$, we can contract all nodes in~$V_2 \cup V_3 \cup \ldots \cup V_k$ into~$v$.
    This results exactly in the rooted graph~$(G_1,t)$ and hence shows that~$(G_1,t)$ is perfectly resilient.
    For any other rooted graph~$(G_i,v)$ with~${i \in \{2,3,\ldots,k\}}$, we can contract all nodes in~$(V_1 \cup V_2\ \cup \ldots \cup V_k) \setminus V_i$ into~$v$.
    Note that if~$v \neq t$, then~$t$ is contracted into~$v$ in the process and hence~$v$ becomes the new root.
    Moreover, the resulting rooted graph is precisely~$(G_i,v)$ showing that~$(G_i,v)$ is perfectly resilient.

    For the backward direction, assume that all rooted graphs~$(G_1,t),(G_2,v),\ldots,(G_k,v)$ are perfectly resilient and let~$\pi^i$ be a perfectly resilient forwarding pattern for~$(G_i,x_i)$ where~$x_1=t$ and~$x_i = v$ for all~$i > 1$.
    we make a case distinction whether~$v = t$ or~$v \neq t$.
    If~$v = t$, then for each node~$u \neq v$, we do the following.
    Let~$j_u$ be the index such that~$u \in V_{j_u}$.
    We use the forwarding function~$\pi_u^{j_u}$.
    Note that only the neighborhood of~$v=t$ is different in~$G_{j_u}$ and~$G$ and hence~$\pi_u^{j_u}$ is a valid forwarding function for each~$u \neq v$ in~$G$.
    Moreover, for any set~$F$ of failed links in~$G$ and any index~$j$, let~$F_{j}$ be the subset of~$F$ that appear in~$G_{j}$.
    Now, if the routing in~$G$ with failure set~$F$ starts in~$u$, then it is the same as the routing in~$G_{j_u}$ with failure set~$F_{j_u}$ as in each case the current node cannot distinguish between~$F$ and~$F_{j_u}$.
    Thus, any such routing reaches~$v=t$ if there exists an~$s$-$t$-path in~$G \setminus F$ by the assumption that~$\pi^{j_u}$ is a perfectly resilient forwarding pattern.
    Since the node~$u$ was chosen arbitrarily, it holds for each possible starting node~$s$ and any possible failure set~$F$ that a routing starting in~$s$ reaches~$t$ if~$s$ and~$t$ remain connected in~$G \setminus F$, that is~$(G,t)$ is perfectly resilient.

    If~$v \neq t$, then we also use the same forwarding function~$\pi^{j_u}$ for each node~$u$ except for~$v$ (and~$t$ which does not forward anyway).
    For the node~$v$, we do the following.
    For each set~$F_v$ of incident failed links and for each in-port~$e$, if no node~$w \in V_1$ remains an active neighbor of~$v$ with respect to~$F_v$, then there is no path from~$v$ to~$t$ in~$G \setminus F_v$ so the choice of the forwarding function does not matter.
    So assume that at least one node~$w \in V_1$ remains an active neighbor.
    If~$e$ is a link in~$G_1$, then we use the forwarding function~$\pi_v^1$ for the same in-port.
    Otherwise, we use the forwarding function~$\pi_v^1$ for the in-port~$\bot$.

    We will show that the resulting forwarding pattern is perfectly resilient.
    To this end, consider an arbitrary starting node~$s$ and an arbitrary set~$F$ of failed links such that there is an~$s$-$t$-path in~$G \setminus F$.
    Let~$\rho$ be the corresponding routing.
    We will show that~$\rho$ contains~$t$.
    If~$s$ belongs to~$G_1$ (in particular if~$s=v$), then~$\rho$ is identical to the routing for~$(G_1,t)$ with the set~$F_1$ of failed links that contain all links in~$F$ that appear in~$G_1$.
    Hence, $t$ is reached by the routing by the assumption that~$\pi^1$ is perfectly resilient and is therefore also contained in~$\rho$.
    If~$s$ belongs to some other set~$V_{j_s}$, then~$\pi^{j_s}$ guarantees that~$\rho$ contains~$v$.
    There, the routing continues by construction as if the routing started in~$v$ using the forwarding function~$\pi^1$.
    Thus, $\rho$ is the gluing of (i) the routing corresponding to~$\pi^{j_s}$, starting node~$s$, and the set~$F_{j_s}$ of failed links in~$G_{j_s}$ and (ii) the routing corresponding to~$\pi^1$, starting node~$v$, and the set~$F_1$ of failed links in~$G_1$.
    Thus, $\rho$ contains~$t$ and since the starting node~$s$ was chosen arbitrarily, this shows that~$(G,t)$ is perfectly resilient.
\end{proof}

For \prs{} we observe the following.
Let~$G$ be a connected planar graph, let~$v$ be a cut node in~$G$, and let~$G'$ be a connected component in~$G-v$ that does not contain~$t$.
Then, $G'$ does not contain any relevant neighbors of~$v$ for any set~$F_v$ of incident failed links.
This yields the following.

\begin{proposition}
    \label{prop:pre2syn}
    Let~$(G,t)$ be a rooted graph where~$G$ is planar and connected.
    We can compute in~$O(n)$ time a set~$\{(G_1,t_1),(G_2,t_2),\ldots,(G_k,t_k)\}$ of rooted minors of~$(G,t)$ where each graph~$G_i$ is planar and biconnected and~$(G,t)$ is perfectly resilient if and only if all~$(G_i,t_i)$ with~${i \in [k]}$ are.
    Moreover, given a perfectly resilient skipping forwarding pattern for each of the rooted graphs~$(G_i,t_i)$, we can construct a perfectly resilient skipping forwarding pattern for~$(G,t)$ in~$O(n^2)$ time.
\end{proposition}

\begin{proof}
    We start by performing a breadth-first search from~$t$ to compute the distance (length of a shortest path) between~$t$ and each other node in~$G$ in~$O(n+m)=O(n)$ time.
    We can then compute all cut nodes and all biconnected components~$G_1,G_2,\ldots,G_k$ in~$G$ in~$O(n)$ time using an algorithm by Hopcroft and Tarjan~\cite{HT73}.
    In each biconnected component~$G_i$ that does not contain~$t$, we make the node with minimal distance from~$t$ the new root.
    Note that this is the cut node~$v$ separating~$G_i$ from~$t$ in~$G-v$ and hence the construction is well-defined.
    Moreover, we can find the new root in time linear in the size of~$G_i$ and hence the whole procedure takes~$O(n)$ time overall.\footnote{We mention that cut nodes are copied in each biconnected component. However, the number of links does not change and the number of nodes increases by at most~$m$. Since~$m \leq 3n$, the sum of all instance sizes is still linear in the number of nodes in the original graph~$G$.}
    Note that the above procedure is equivalent to repeatedly applying the preprocessing routing described in \cref{prop:pre2dec}.
    Hence, $(G,t)$ is perfectly resilient if and only if all~$(G_i,t_i)$ with~$i \in [k]$ are.
    Moreover, by contracting all other biconnected components, it is easy to see that each rooted graph~$(G_i,t_i)$ is a minor of~$(G,t)$.

    For the second part of the proof, assume that we are given a perfectly resilient skipping forwarding pattern~$\pi^i$ for each rooted graph~$(G_i,t_i)$.
    The proof roughly follows the same structure as the proof of \cref{prop:pre2dec}.
    We first show how to combine the different perfectly resilient skipping forwarding patterns and then prove that the result is a perfectly resilient skipping forwarding pattern for~$(G,t)$.
    For each node~$u \neq t$, if~$u$ is not a cut node in~$G$, then let~$j_u$ be the index such that~$u$ is contained in~$G_{j_u}$.
    The skipping forwarding function for~$u$ is~$\pi^{j_u}_u$.
    For each cut node~$u \neq t$ in~$G$, let~$G_{j_u}$ be the biconnected component containing~$u$ for which~$u$ is not the new root for~$(G_j,t_j)$.
    Note that we can find this biconnected component in~$O(|N(u)|)$ time by finding a node with a smaller distance from~$t$ than~$u$ and that this biconnected component is unique.
    Let~$G_{i_1},G_{i_2},\ldots,G_{i_k}$ be the other biconnected components containing~$u$.
    The skipping priority list~$\ski{u}{e}$ for any in-port~$e$ is (i)~$\pi_u^{j_u}(e) \circ \sigma$ if~$e$ is a link in~$G_{j_u}$ and (ii) $\pi^{j_u}_u(\bot) \circ \sigma$, otherwise.
    Therein, $\sigma$ is an arbitrary permutation of all neighbors of~$u$ in~$G_{i_1},G_{i_2},\ldots,G_{i_k}$.
    Note that we indeed constructed valid skipping priority lists for each node and each possible in-port.
    Moreover, computing an arbitrary permutation of~$x$ elements takes~$O(x)$ time, so the time needed for computing one skipping priority list for a node~$u$ is in~$O(|N(u)|)$.
    Since there are~$|N(u)|+1$ possible in-ports for~$u$, computing a skipping forwarding function for a node~$u$ takes~$O(|N(u)|^2)$ time.
    Summed over all nodes and using the handshaking lemma and the fact that~$m \in O(n)$, the overall running time is in~$O(n^2)$.

    It remains to show that the constructed skipping forwarding pattern is perfectly resilient.
    To this end, consider an arbitrary starting node~$s$ and an arbitrary set~$F$ of failed links such that there is an~$s$-$t$-path in~$G \setminus F$.
    Let~$\rho=(s,u_1,u_2,\ldots)$ be the routing corresponding to the constructed skipping forwarding pattern, starting node~$s$, and set~$F$ of failed links.
    We will first show that~$u_\ell = t_{j_s}$ for some~$\ell \geq 1$.
    Afterwards, we will show how this implies that~$u_{k} = t$ for some~$k \geq 1$, that is, the routing reaches~$t$.
    Since the starting node~$s$ and the set~$F$ of failed links was chosen arbitrarily (up to the condition that an~$s$-$t$-path remains), this shows that the constructed forwarding pattern is perfectly resilient, concluding the proof.
    
    To verify that~$u_\ell = t_{j_s}$ for some~$\ell \geq 1$, note that as long as~$u_i \neq t_{j_s}$, $u_{i+1} = \pi_{u_i}^{j_s}(e_i,F_{j_s})$, where~$u_0 = s$,~$e_0 = \bot$ and~$e_i = \{u_{i-1},u_i\}$ for all~$i \geq 1$.
    This is due to the fact that only the skipping priority lists of cut nodes are different between~$\pi^{j_s}$ and the constructed skipping forwarding functions.
    Moreover, for cut nodes in~$G_j$ other than~$t_{j_s}$, all links to other biconnected components appear after all links in~$G_j$ and each such node that is contained in~$\rho$ has at least one incident link in~$G_j$ that does not fail.
    This holds as all such nodes are connected to~$s$ as witnessed by~$\rho$ and~$s$ is connected to~$t$ in~$G \setminus F$ by assumption, that is, there is path between~$s$ and~$t_{j_s}$.
    Hence, until~$t_{j_s}$ is reached by~$\rho$, $\rho$ is the same as the routing corresponding to~$\pi^{j_s}$ with starting node~$s$ and set~$F_{j_s}$ of failed links.
    By assumption, $\pi^{j_s}$ is perfectly resilient for~$(G_{j_s},t_{j_s})$ and therefore reaches~$t_{j_s}$ eventually.
    This proves that~$t_{j_s}$ is contained in~$\rho$ or equivalently, $t_{j_s} = u_\ell$ for some~$\ell \geq 1$.

    Finally, we show that~$u_{k} = t$ for some~$k \geq 1$.
    If~$t = t_{j_s}$, then we are trivially done.
    So assume otherwise and let~$i \geq 1$ be the smallest index such that~$u_i = t_{j_s}$ and consider the two sequences~$\rho_1 = (s,u_1,u_2,\ldots,u_i)$ and~$\rho_2=(u_i,u_{i+1},\ldots)$.
    Note that the distance between~$s$ and~$t$ is finite (in~$G$) as we assumed~$G$ to be connected.
    Moreover, by construction, when reaching~$t_{j_s}$ via a link in~$G_{j_s}$, the next node~$u_{i+1} = \pi^{j_{t_{j_s}}}_{t_{j_s}}(\bot,F)$ in~$\rho$ is the same as if the routing started in~$t_{j_s}$.
    Since the node~$s$ above was chosen arbitrarily, we can repeat the argument with starting node~$t_{j_s}$ to show that~$\rho$ also contains the node~$t_{j_{t_{j_s}}}$ which has even smaller distance from~$t$.
    Since the distance between~$s$ and~$t$ is finite, we can repeat this argument until the distance becomes zero, that is, this shows that the node~$t$ is eventually reached.
    This concludes the proof.
\end{proof}

\subsection{Structure of Separating Links}
\label{sec:prehelp1}
In order to present our preprocessing for \prs{} to remove separating links, we first show a number of auxiliary results here.
Recall that a link~$e=\{u,v\}$ is separating if~$G'= G-e = G[V \setminus \{u,v\}]$ is disconnected.
We will often distinguish between separating links which have~$t$ as an endpoint and other separating links.
To make this distinction and for notational convenience, we denote for a separating link~$e$ that is not incident to~$t$ by~$S_e$ the set of all links that are separated from~$t$ by~$e$.
For a separating link~$f$ incident to~$t$ or a non-separating link~$f$, we define~$S_f = \emptyset$.
Notice that if a separating link~$e$ and another link~$e'$ do not share any endpoints, then both endpoints of~$e'$ are in the same connected component of~$G'$.
\begin{observation}
    \label{obs:nocross}
    Let~$G$ be a graph and let~$e$ and~$e'$ be two links, where~$e$ is a separating link and~$e$ and~$e'$ do not share any endpoints.
    Then, both endpoints of~$e'$ are in the same connected component of~$G-e$.
\end{observation}

We will assume the input graph to be biconnected.
In order to apply the following arguments multiple times, we first show that removing a separating link from a biconnected graph results in a 2-connected graph.

\begin{lemma}
    \label{lem:stay2con}
    Let~$G$ be a biconnected graph and let~$e$ be a separating link in~$G$. Then~$G \setminus e$ is 2-connected.
\end{lemma}
\begin{proof}
    Note that the graph~$K_2$ does not contain a separating link and hence~$G$ is 2-connected.
    Let~$e=\{u,v\}$ and let~$V_1,V_2,\ldots,V_k$ be the set of nodes in the connected components in~$G-e$.
    Since~$e$ is separating, there are~$k \geq 2$ such components.
    Moreover, each such component contains a neighbor of~$u$ and a neighbor of~$v$ in~$G$ as if one set~$V_i$ contains no neighbor of~$u$, then~$v$ is a cut node in~$G$, contradicting that~$G$ is 2-connected.
    This also shows that~$G\setminus e$ is connected.
    Let~$G_i = G[V_i \cup \{u,v\}] \setminus \{v,t\}$ for each~$i \in [k]$.
    Now, assume towards a contradiction that~$G \setminus e$ is not 2-connected.
    Then, there exists a cut node~$q$ in~$G\setminus e$, that is, there are nodes~$p,q,r$ such that all~$p$-$r$-paths contain~$q$.
    Since~$G$ is 2-connected, there exist two internally disjoint~$p$-$r$-paths~$P_1$ and~$P_2$ in~$G$.
    As~$G$ and~$G \setminus e$ only differ in the link~$e$, one of these paths contains the link~$e$.
    Let without loss of generality be~$P_1$ the path that contains~$e$.
    Then, $P_2$ is completely contained in one graph~$G_i$.
    Hence, we can replace the link~$e$ by an arbitrary path in a component in~$\{G_1,G_2\}\setminus \{G_i\}$.
    As argued above, each such component contains neighbors of both~$u$ and~$v$ and since they are connected, they provide the desired connectivity.
    This contradicts the assumption that~$G \setminus e$ is not 2-connected and concludes the proof.
\end{proof}

We next prove a a lemma that will later allow us to temporarily remove separating links from a graph.
We prove that this does neither create new separating links nor makes any other separating links be no longer separating.
Moreover, the removal of a link~$e$ does not change the set~$S_f$ for any other link~$f$ except that~$e$ is removed from it (if $e \in S_f$ in the first place).

\begin{lemma}
    \label{lem:onego}
    Let~$G$ be a biconnected graph.
    Let~$f$ be a separating link in~$G$ and let~$e \neq f$ be a link in~$G$.
    Then,~$f$ is separating in~$G \setminus e$ and if~$e$ is separating in~$G \setminus f$, then~$e$ is also separating in~$G$ and the set~$S'_e$ of links separated from~$t$ by~$e$ in~$G \setminus f$ satisfies~$S'_e = S_e \setminus \{f\}$.
\end{lemma}
    
\begin{proof}
    Since~$f$ is separating in~$G$, it separates at least two nodes~$u$ and~$v$ from one another, that is, all paths from~$u$ to~$v$ contain at least one endpoint of~$f$.
    Since removing a link~$e$ does not create any new paths, the same holds in~$G \setminus e$.
    
    For the second part of the proof, assume towards a contradiction that~$e$ is separating in~$G \setminus f$ but not in~$G$.    
    Let~$f=\{u,v\}$ and let~$u'$ and~$v'$ be two nodes that are separated from one another by~$e$ in~$G \setminus f$.
    We will show that~$u$ and~$v$ are also separated by~$e$ in~$G$.
    Assume towards a contradiction that this is not the case.
    Then, there exists an~$u'$-$v'$-path~$P$ in~$G-e$.
    Since~$P$ does not exist in~$(G \setminus f) - e$, $P$ contains the link~$f=\{u,v\}$.
    Since~$P$ exists in~$G-e$, it holds that~$e$ and~$f$ do not share any endpoints, that is~$e \cap f = \emptyset$.
    By \cref{obs:nocross} and the fact that~$f$ is separating in~$G$, it holds that both endpoints of~$e$ are in the same connected component of~$G - f$.
    Moreover, there exists a different connected component~$G^*$ in~$G-f$ (that does not contain~$e$) for the same reason.
    Note that there is a neighbor of~$u$ in~$G^*$ as otherwise~$v$ would be a cut node, contradicting that~$G$ is 2-connected (it is biconnected and contains at least two links~$e$ and~$f$).
    The same also holds for~$v$.
    Since~$G^*$ is by definition connected, there is an~$u$-$v$-path~$P'$ in~$G$ where all internal nodes are in~$G^*$.
    Thus, we can replace the link~$\{u,v\}$ in~$P$ by the path~$P'$ to get a walk from~$u'$ to~$v'$ that does not contain~$f$.
    Then, there also exists a path~$P^*$ with the same property.
    This path does not contain~$f$ and it therefore exists in~$(G \setminus f)-e$, contradicting that~$e$ separates~$u'$ and~$v'$ from one another in~$G \setminus f$.
    Note that the above argument also shows that~$S'_e = S_e \setminus \{f\}$ since it holds for each pair of nodes that they are separated by~$e$ in~$G$ if and only if they are separated by~$e$ in~$G \setminus f$.
    This concludes the proof.
\end{proof}

We next observe that the set~$S$ of all separating links can be computed in~$O(n^2)$ time.

\begin{observation}
    \label{obs:computeS}
    Let~$G$ be a planar graph.
    Then, we can compute the set~$S_t$ of all separating links incident to~$t$, the set~$S^*$ of separating links not incident to~$t$, and~$S_e$ for each link~$e$ in~$G$ in overall~$O(n^2)$ time.
\end{observation}

\begin{proof}
Since~$G$ is planar, it contains at most~$O(n)$ links.
For each link~$e$, we can compute~$G-e$ in~$O(n)$ time.
We can then compute all connected component in~$G-e$ in~$O(n)$ time using breadth-first search.
If there are at least two such components, then~$e$ is a separating link and we add it to~$S_t$ if it contains~$t$ and to~$S^*$ otherwise.
If there is only a single connected component in~$G-e$ or if we added~$e$ to~$S_t$, then we set~$S_e = \emptyset$.
If~$t$ is not an endpoint of~$e$ and~$e$ is a separating link, then we add all links with at least one endpoint in a different connected component than~$t$ in~$G-e$ to~$S_e$.
Note that the time required for a single link is~$O(n)$ yielding an overall running time of~$O(n^2)$.
\end{proof}

Lastly, we need a lemma specifically for the case where~$G$ does not contain any separating links incident to~$t$.
Then, we show that if~$G$ contains a separating link, then it also contains a separating link that is not separated from~$t$ by any other separating link.
To this end, we first show an intermediate lemma regarding the set~$S_e$ of links separated from~$t$ by a link~$e$.

\begin{lemma}
    \label{lem:separating}
    Let~$(G,t)$ be a rooted graph where~$G$ is a 2-connected planar graph.
    Let~$e,f,g$ be three links in~$G$.
    It holds that
    \begin{compactenum}
        \item if~$f \in S_e$, then~$e \notin S_f$,
        \item if~$f \in S_e$ and~$g \in S_f$, then~$g \in S_e$, and
        \item if~$g \in S_e$ and~$g \in S_f$, then~$f \in S_e$ or~$e \in S_f$.
    \end{compactenum}
\end{lemma}

\begin{proof}
    For the first point, assume that~$f \in S_e$.
    Let~$w$ be an endpoint of~$f$ that is separated from~$t$ by~$e$.
    Note that~$w$ is not an endpoint of~$e$.
    Since~$e$ is separating (as~$f \in S_e$), it holds by \cref{obs:nocross} that the other endpoint of~$f$ is also an endpoint of~$e$ or it is also separated from~$t$ by~$e$.
    Since~$G$ is 2-connected, there are two internally node-disjoint~$w$-$t$-paths~$P_1$ and~$P_2$ in~$G$.
    Since~$e$ separates~$w$ from~$t$, each of these paths contains exactly one endpoint of~$e$.
    Since~$P_1$ and~$P_2$ only overlap in~$w$ and~$w$ is not an endpoint of~$e$, it holds that at most one of~$P_1$ and~$P_2$ contain the other endpoint of~$f$.
    Let, without loss of generality, $P_1$ be a path that does not contain the other endpoint of~$f$.
    Then, it holds that the subpath of~$P$ from an endpoint of~$e$ to~$t$ does not contain an endpoint of~$f$.
    Thus, $f$ cannot separate this node from~$t$.
    By \cref{obs:nocross}, the other endpoint of~$e$ is contained in~$f$ or is also not separated from~$t$ by~$f$.
    Hence, $e$ is not separated from~$t$ by~$f$, that is, $e \notin S_f$.

    For the second point, assume towards a contradiction that~$f \in S_e$, $g \in S_f$, and~$g \notin S_e$.
    Let~${g=\{u,v\}}$.
    We assume without loss of generality that~$u$ is not an endpoint of~$e$.
    Then, since~$g \notin S_e$, it holds by \cref{obs:nocross} that~$u$ is in the same connected component~$G_1$ as~$t$ in~$G - e$.
    As~$f \in S_e$, at least one endpoint of~$f$ is in a different connected component~$G_2$ and by \cref{obs:nocross}, the other endpoint is shared with~$e$ or is also in~$G_2$.
    Hence, there is a connection between~$u$ and~$t$ within~$G_1$, that is, there is a path that does not use either endpoint of~$e$ or~$f$, yielding~$g \notin S_f$, a contradiction.

    For the third point, assume towards a contradiction that~$g \in S_e$, $g \in S_f$, $f \notin S_e$, and~$e \notin S_f$.
    Let~$g = \{u,v\}$, where $u$ is not an endpoint of~$e$.
    Note that such a node exists as~$e \neq g$. 
    Moreover, since~$e \neq f$, there is at most one node that is an endpoint of both~$e$ and~$f$.
    Since~$G$ is 2-connected, there exists a path~$P$ from~$u$ to~$t$ in~$G$ that does not contain such a common endpoint of~$e$ and~$f$.
    Let~$w$ be the first node in~$P$ that is contained in~$e$ or~$f$ (where potentially~$u=w \in f$) and in any case~$w$ exists as~$u$ is separated from~$t$ by~$e$.
    To avoid case distinctions, we rename~$e$ and~$f$ to~$c$ and~$d$, where~$c$ is the link that contains~$w$ as an endpoint.
    Since~$c \notin S_d$ and~$w \notin d$, there is a path from~$w$ to~$t$ that does not contain any endpoint of~$d$.
    The union of the subpath of~$P$ from~$u$ to~$w$ (which does not contain an endpoint of~$d$ by assumption) and this~$w$-$t$-path is a path from~$u$ to~$t$ that does not contain any endpoint of~$d$, that is,~$u$ is not separated from~$t$ by~$d$.
    This contradicts~$g \in S_e$ if~$d=e$ and~$g \in S_f$ if~$d=f$ as~$u \notin d$ in both cases.
    This concludes the proof.
\end{proof}

As a direct consequence, we get the above claim that there always exists a separating link that is not separated from~$t$ by any other separating link.
Consequently, we can order the links in~$S^*$ such that whenever~$f \in S_e$, then~$e$ appears before~$f$ in the ordering.

\begin{lemma}
    \label{lem:rootlink}
    Let~$(G,t)$ be a rooted graph where~$G$ is biconnected and planar.
    Let~$S$ be the separating links in~$G$.
    If~$S$ does not contain any links incident to~$t$, then there is an ordering~$(e_1,e_2,\ldots,e_{|S|})$ of the links in~$S$ such that whenever~$e_j \in S_{e_i}$, then~$i < j$.
    This ordering can be computed in~$O(n^2)$ time.
\end{lemma}

\begin{proof}
    We first compute the set~$S$ and the set~$S_e$ for each~$e \in S$ in~$O(n^2)$ time using \cref{obs:computeS}.
    We first show that the claimed ordering exists and then how to compute it.
    Set initially~$S'=S$.
    We show that a separating link~$e_1$ exists such that~$e_1 \notin S_f$ for any~$f \in S'$.
    We then recursively construct a solution sequence for~$S' \setminus \{e_1\}$ (in~$G \setminus e_1$) and add the link~$e_1$ to the beginning of that solution sequence.
    Note that the constructed sequence fulfills the requirement of the lemma statement since~$S_f$ does not change for any link~$f \neq e_1$ by \cref{lem:onego}.
    We next show that such a link~$e_1$ exists.
    To this end, assume towards a contradiction no such link exists and consider an arbitrary link~$f_1 \in S$.
    Then, there exists a link~$f_2 \in S_{f_1}$.
    Moreover, there exists a link~$f_3 \in S_{f_2}$ and so on.
    Since the number of links in~$G$ is finite, it holds that~$f_i = f_j$ for some~$i < j$.
    By the second point of \cref{lem:separating}, it holds that~$f_{j-1} \in S_{f_i}$.
    Moreover, since~$f_i = f_j \in S_{f_{j-1}}$, this contradicts the first point of \cref{lem:separating}.
    Hence, a link~$e_1$ with~$e_1 \notin S_f$ for each~$f \in S'$ exists.

    We next describe how to compute the sequence in~$O(n^2)$ time.
    We start by building a directed graph~${D=(S,E')}$ where the nodes represent the set~$S$ of separating links and there is an arc~$(e,f)$ (a directed link) if and only if~$f \in S_e$.
    Note that the graph can be computed in~$O(m^2) = O(n^2)$ time.
    Moreover, the graph is acyclic as any cycle would contradict the exists of a sequence as constructed above.
    We then simply compute a topological ordering of~$D$ in~$O(n+m')=O(n^2)$ time, where~$m'$ is the number of arcs in~$D$.
    Note that the topological ordering is an ordering of~$S$ such that for all~$e_i,e_j \in S$ it holds that if~$(e_i,e_j) \in E'$ (meaning that~$e_j \in S_{e_i}$), then~$e_i$ comes before~$e_j$ in the ordering.
    This concludes the proof.
\end{proof}

\subsection{Hierarchical Embeddings and Forwarding Patterns}
\label{sec:prehelp2}
Before we can state the main result of this section, we require two last ingredients: \textbf{\textit{hierarchical planar embeddings}} and \textbf{\textit{hierarchical right-hand forwarding patterns}}.
For both of them, we will work with directed planar graphs where arcs either exist in both directions or in neither.
Such graphs are called \textbf{\textit{symmetric directed (planar) graphs}} and we say the process of turning a graph into a symmetric directed graph by replacing each link with a pair of directed arcs, one in each direction, is taking the \textbf{\textit{symmetric orientation}} of the graph.
Note that any pair of symmetric arcs between any two nodes~$u$ and~$v$ partition the plane into an interior region and an exterior region.
See \cref{fig:emptyfaces} for an example of interior regions and some of the concepts introduced in the following.
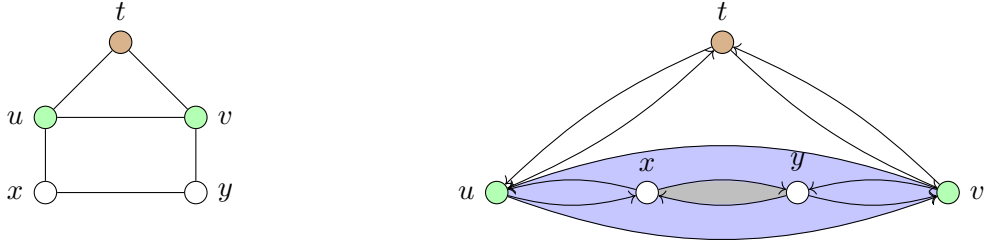
\begin{figure}[t]
    \centering
    \begin{tikzpicture}
        \node[te] at(-8,2) (t2) {};
        \node[por,label=left:$u$] at(-9,1) (u2) {} edge(t2);
        \node[por,label=right:$v$] at(-7,1) (v2) {} edge(t2) edge(u2);
        \node[ver,label=left:$x$] at(-9,0) (x2) {} edge(u2);
        \node[ver,label=right:$y$] at(-7,0) (y2) {} edge(v2) edge(x2);
    
        \node[te] at(0,2) (t) {};
        \node[por,label=left:$u$] at(-3,0) (u) {};
        \node[ver,label=$x$,fill=white] at(-1,0) (x) {};
        \node[ver,label=$y$,fill=white] at(1,0) (y) {};
        \node[por,label=right:$v$] at(3,0) (v) {};

        \draw[->,bend right=15] (u) to (x);
        \draw[->,bend right=15] (x) to (u);
        \draw[->,bend right=20] (u) to (v);
        \draw[->,bend right=20] (v) to (u);
        \draw[->,bend right=10] (u) to (t);
        \draw[->,bend right=10] (t) to (u);
        \draw[->,bend right=10] (v) to (t);
        \draw[->,bend right=10] (t) to (v);
        \draw[->,bend right=15] (v) to (y);
        \draw[->,bend right=15] (y) to (v);
        \draw[->,bend left=15] (x) to (y);
        \draw[->,bend left=15] (y) to (x);

        \begin{scope}[on background layer]
            \path [fill=blue!22] (u.center) to[bend right=21] (v.center) 
    to [bend right=21] (u.center);
            \path [fill=lightgray] (x.center) to[bend right=15] (y.center) 
    to [bend right=15] (x.center);
        \end{scope}
    \end{tikzpicture}
    \caption{A planar rooted graph without separating links incident to~$t$ on the left. The right shows its symmetric orientation. The interior region of~$\{u,v\}$ is shaded blue (including the gray-shaded part) and the interior region of~$\{x,y\}$ is shaded gray. The link~$\{u,v\}$ encloses~$\{u,x\}$, $\{x,y\}$, and~$\{y,v\}$.
    The embedding is not a hierarchical planar embedding as the arcs~$(x,y)$ and~$(y,x)$ enclose the interior region in the wrong orientation. Swapping both arcs results in a hierarchical planar embedding.}
    \label{fig:emptyfaces}
\end{figure}

Since we consider planar embeddings, it holds for each other arc~$f$ that~$f$ is either fully contained in the interior region or fully contained in the exterior region.
Moreover, it is contained in the interior region, if and only if both endpoints are contained in the interior region (where we say that~$u$ and~$v$ are contained in the interior region).
Given a planar embedding of a symmetric directed graph that is the result of taking the symmetric orientation of an undirecetd planar graph~$G=(V,E)$, we say that~$\{u,v\} \in E$ has an \textbf{\textit{empty face}} if the interior region of~$(u,v)$ and~$(v,u)$ does not contain any nodes or links (other than~$u$, $v$, $(u,v)$, and~$(v,u)$).
Otherwise, it has a \textbf{\textit{virtual face}}.
If~$\{u,v\}$ has a virtual face, then note that for each arc~$(x,y)$ contained in the virtual face, it also contains the arc~$(y,x)$ as it contains both~$x$ and~$y$ (and~$\{x,y\} \neq \{u,v\}$).
In this case, we say that~$\{u,v\}$ \textbf{\textit{encloses}}~$\{x,y\}$ in the embedding.

We next define \textbf{\textit{hierarchical planar embeddings}}.
Such an embedding is defined for a rooted graph~$(G,t)$ where~$G$ is a biconnected undirected planar graph without separating links incident to~$t$.
It is an embedding of the symmetric orientation of~$G$ such that for each link~$e=\{u,v\}$ it holds that
\begin{compactitem}
    \item when traversing through~$(u,v)$ and then through~$(v,u)$ one traverses the interior region in counterclockwise order, and
    \item for each other link~$f$, it holds that~$e$ encloses~$f$ if and only if~$f \in S_e$.
\end{compactitem}

We will show a bit later (\cref{lem:hierembedding}) that such a hierarchical planar embedding always exists.
Before doing so, we introduce the last concept in this section: \textbf{\textit{hierarchical right-hand forwarding patterns}}. 
Such forwarding patterns are also defined for rooted graphs~$(G,t)$ where~$G$ is a biconnected undirected planar graph without separating links incident to~$t$.
It also requires a hierarchical planar embedding.
With these prerequisites given, a hierarchical right-hand forwarding pattern consists of a skipping forwarding function~$\Lambda_v^{e_v}$ for each node~$v \in V \setminus \{t\}$ such that for any in-port~$f \in E_v \cup \{\bot\}$ the following holds:
\begin{compactitem}
    \item if~$f = \bot$, then~$\Lambda_v^{e_v}(f) = \Lambda_v^{e_v}(e_v)$ and
    \item if~$f = \{u,v\} \in E_v$, then~$\Lambda_v^{e_v}(f) = \sigma_t \circ \sigma_p \circ \sigma_r$, where
    \begin{compactitem}
        \item $\sigma_t = (t)$ if~$\{v,t\} \in E_v$ and~$\sigma_t = ()$ (the empty sequence) otherwise,
        \item $\sigma_p$ is any permutation of any subset~$D$ of~$\{e \mid f \in S_e\}$, and
        \item $\sigma_r$ is the right-hand rule for~$v$ and in-port~$f$ for the given hierarchical planar embedding in~$G \setminus (D \cup \{v,t\})$, that is, starting from the arc~$(u,v)$, it goes counterclockwise around~$v$ and lists all outgoing arcs in order unless they already appear in~$\sigma_t$ or~$\sigma_p$.
\end{compactitem}
\end{compactitem}
See \cref{fig:embedding} for an example.
We say that the first part~$\sigma_t$ is called the \emph{$t$-prefix}, the second part~$\sigma_p$ is called the \emph{priority part}, and the last part~$\sigma_r$ is called the \emph{regular part}.
\begin{figure}[t]
    \centering
    \begin{tikzpicture}
        \node[te] at(0,2) (t) {};
        \node[por,label=left:$u$] at(-3,0) (u) {};
        \node[ver] at(-1,0) (x) {};
        \node at(-1,.3) {$x$};
        \node[ver,label=$y$] at(1,0) (y) {};
        \node[por,label=right:$v$] at(3,0) (v) {};
        \node[ver,label=left:$p$] at(-3,-1.5) (p) {};
        \node[ver,label=below:$q$] at(-3,-3) (q) {};

        \draw[->,bend right=15] (u) to (x);
        \draw[->,bend right=15] (x) to (u);
        \draw[->,bend right=20] (u) to (y);
        \draw[->,bend right=20] (y) to (u);
        \draw[->,bend right=25] (u) to (v);
        \draw[->,bend right=25] (v) to (u);
        \draw[->,bend right=10] (u) to (t);
        \draw[->,bend right=10] (t) to (u);
        \draw[->,bend right=10] (v) to (t);
        \draw[->,bend right=10] (t) to (v);
        \draw[->,bend right=15] (v) to (y);
        \draw[->,bend right=15] (y) to (v);
        \draw[->,bend right=15] (x) to (y);
        \draw[->,bend right=15] (y) to (x);
        \draw[->,bend right=15] (u) to (p);
        \draw[->,bend right=15] (p) to (u);
        \draw[->,bend right=15] (p) to (q);
        \draw[->,bend right=15] (q) to (p);
        \draw[->,bend right=15] (u) to (q);
        \draw[->,bend right=15] (q) to (u);
        \draw[->,bend right=10] (q) to (v);
        \draw[->] (v) to (q);
    \end{tikzpicture}
    \caption{An example of a hierarchical planar embedding. The priority list~$\Lambda_u^{e_u}(x) = (t,y,q,p,v,x)$ is a valid choice since it can be written as~$(t) \circ (y) \circ(q,p,v,x)$, $\{u,y\}$ encloses~$\{u,x\}$, and starting from the arc~$(x,u)$ and going counterclockwise around~$u$, the remaining outgoing arcs appear in the order~$((u,q),(u,p),(u,v),(u,x))$. The priority list~$\Lambda_u^{e_u}(x) = (t,y,q,v,p)$ would not be valid since only~$\{u,y\}$ and~$\{u,v\}$ enclose~$\{u,x\}$, so everything after~$\{u,q\}$ has to appear in the regular part of the priority list, that is, the links have to appear in counterclockwise order, which it does not ($(u,v)$ appears after~$(u,q)$).}
    \label{fig:embedding}
\end{figure}
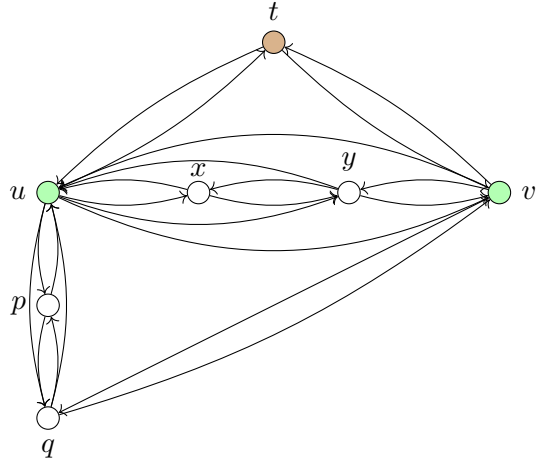%
Note that if~$G$ does not contain any separating links, then~$\Lambda_v^{e_v}$ is completely determined by the hierarchical planar embedding.

We next show that a hierarchical planar embedding always exists for~$(G,t)$ if~$G$ is biconnected, planar, and does not contain any separating links incident to~$t$.
To this end, we first start with the easy case where~$G$ does not contain any separating links.
In this case, we do basically the same we already did when initializing the doubly connected half-edge lists data structure.
For the sake of readability, we say that a graph~$G$ that is planar, biconnected, and does not contain any separating links is \textbf{\textit{\nice}}.

\begin{observation}
    \label{obs:hier}
    Let~$(G,t)$ be a rooted graph where~$G$ is \nice.
    For any planar embedding of~$G$, there is also a hierarchical planar embedding for it such that~$\lambda_v^{e_v} = \Lambda_v^{e_v}$ for each node~$v \in V \setminus \{t\}$.
    This embedding can be computed in linear time.
\end{observation}

\begin{proof}
    Starting with the planar (undirected) embedding for~$G$, we replace each link by two symmetric arcs as follows.
    For each node~$v \in V$, let~$k$ be the number of links incident to~$v$ and let~$(e_1=\{v,u_1\},e_2=\{v,u_2\},\ldots,e_{k}=\{v,u_{k}\})$ be the cyclic ordering of the links incident to~$v$ in the undirected embedding.
    For the hierarchical planar embedding, we set the list of incident arcs to~$((v,u_1),(u_1,v),(v,u_2),(u_2,v),\ldots,(v,u_{k}),(u_{k},v))$.
    Note that each link has an empty face in the embedding and the interior face of any link is traversed in counterclockwise order when traversing two symmetric arcs.
    Thus, the embedding is indeed a hierarchical planar embedding.
    We next show that~$\lambda_v^{e_v} = \Lambda_v^{e_v}$ holds for each node~$v$.
    Consider any in-port~$f$ of~$v$.
    If~$f = \bot$, then replace~$f$ by~$e_v$ in the following argument.
    Let~$f=\{u,v\}$ and let~$\{v,w\}= g \neq h=\{v,x\} $ be two links (which might be equal to~$f$).
    Note that by construction, if we encounter~$g$ before~$h$ when going counterclockwise around~$v$ starting with link next to~$f$, then we also encounter~$(v,w)$ before~$(v,x)$ when going counterclockwise around~$v$ in the constructed hierarchical planar embedding when starting with~$(u,v)$.
    Thus, $\lambda_v^{e_v} = \Lambda_v^{e_v}$.
    The running time is linear since we iterate over all nodes and then over all incident arcs once.
    By the handshaking lemma, the running time is in~$O(m)$ which is~$O(n)$ in planar graphs.
\end{proof}

We next show that a hierarchical planar embedding always exists.
To this end, we only show that the two symmetric arcs of a separating link that are not separated by any existing links can be added to any hierarchical planar embedding.
To show that this implies a hierarchical planar embedding for~$G$, start with a hierarchical planar embedding for~$G \setminus S$, which exists by \cref{obs:hier}.
Initially, set~$S' = S$.
By \cref{lem:rootlink}, we can then iteratively find a link~$e \in S'$ such that~$f \notin S_e$ for all other links~$f \in S'$.
We then remove~$e$ from~$S'$ and add the two symmetric arcs corresponding to~$e$ to the existing hierarchical planar embedding.
We then repeat this process until all links in~$S$ have been embedded.

\begin{lemma}
    \label{lem:hierembedding}
    Let~$(G,t)$ be a rooted graph where~$G$ is planar and biconnected.
    Let~$S$ be the set of separating links in~$G$ and assume that~$S$ contains no links incident to~$t$.
    Let~$e \in S$ such that~$e \notin S_f$ for all~$f \in S$.
    If a hierarchical planar embedding for~$G \setminus e$ is given, then we can add the two symmetric links corresponding to~$e$ to get a hierarchical planar embedding for~$G$.
    The positions where these links are added can be found in linear time.
\end{lemma}

\begin{proof}
    We will first show how to add the two symmetric arcs corresponding to~$e$ to the existing hierarchical planar embedding in~$O(n)$ time.
    We will the show that the resulting embedding satisfies all requirements of a hierarchical planar embedding, in particular, that the embedding is planar.
    Let~$e = \{u,v\}$.
    We first mark all arcs in the hierarchical planar embedding that belong to links in~$S_e$.
    This takes~$O(n)$ time.
    Let~$e_u$ and~$e_v$ be links incident to~$u$ and~$v$, respectively, such that~$e_u,e_v \notin S_e$.
    Note that such links exist as~$G$ is biconnected.
    We go around~$u$ and~$v$ in counterclockwise direction starting from an arc corresponding to~$e_u$ and~$e_v$, respectively, and find the position just before the first marked link and just after the last marked link.
    For node~$u$, we add the outgoing arc~$(u,v)$ just before the first marked link and the incoming arc~$(v,u)$ just after the last marked link.
    We do the same for~$v$.
    See \cref{fig:embedding2} for an example.
    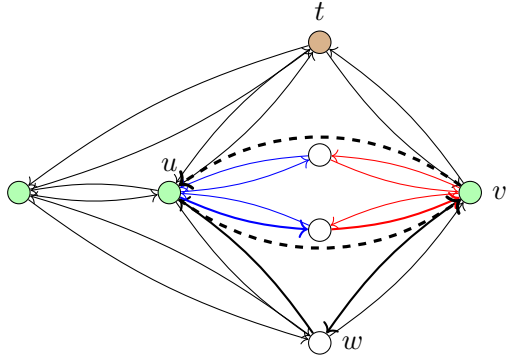
\begin{figure}[t]
        \centering
        \begin{tikzpicture}
            \node[te] at(1,2) (t) {};
            \node[por] at(-3,0) (x) {};
            \node[por,label=$u$] at(-1,0) (u) {};
            \node[ver] at(1,.5) (y) {};
            \node[ver] at(1,-.5) (z) {};
            \node[por,label=right:$v$] at(3,0) (v) {};
            \node[ver,label=right:$w$] at(1,-2) (a) {};
        
            \draw[->,bend right=10] (t) to (x);
            \draw[->,bend right=10] (t) to (u);
            \draw[->,bend right=10] (t) to (v);
            \draw[->,bend right=10] (x) to (t);
            \draw[->,bend right=10] (u) to (t);
            \draw[->,bend right=10] (v) to (t);
            \draw[->,bend right=10] (a) to (x);
            \draw[->,bend right=10,thick] (a) to (u);
            \draw[->,bend right=10] (a) to (v);
            \draw[->,bend right=10] (x) to (a);
            \draw[->,bend right=10] (u) to (a);
            \draw[->,bend right=10,thick] (v) to (a);
            \draw[->,bend right=10] (x) to (u);
            \draw[->,bend right=10] (u) to (x);
            \draw[->,bend right=10,blue] (u) to (y);
            \draw[->,bend right=10,blue,thick] (u) to (z);
            \draw[->,bend right=10,red] (v) to (y);
            \draw[->,bend right=10,red] (v) to (z);
            \draw[->,bend right=10,blue] (y) to (u);
            \draw[->,bend right=10,blue] (z) to (u);
            \draw[->,bend right=10,red] (y) to (v);
            \draw[->,bend right=10,red,thick] (z) to (v);
            \draw[->,very thick,dashed,bend right=36] (u) to (v);
            \draw[->,very thick,dashed,bend right=36] (v) to (u);
        \end{tikzpicture}
        \caption{An example of how we add two symmetric arcs to a hierarchical planar embedding. The marked arcs incident to~$u$ are colored blue and the marked arcs incident to~$v$ are colored red. The bold arcs indicate where we trace the face in which we add~$(u,v)$ (lower dashed arc) in the proof of \cref{lem:hierembedding}. The two colored arcs show the path using links in~$S_{\{u,v\}}$ and the black arcs show the links not in~$S_{\{u,v\}}$ when we go from~$u$ to~$v$.}
        \label{fig:embedding2}
    \end{figure}%
    Note that this takes~$O(n)$ time overall.

    We first show that all marked links form a consecutive interval around~$u$ and~$v$, that is, the positions we computed are unique.
    Assume towards a contradiction that there are arcs~$a_1,a_2,a_3,a_4$ such that~$a_1$ and~$a_3$ are marked, $a_2$ and~$a_4$ are not marked, and the arcs appear (not necessarily consecutively) in order~$(a_1,a_2,a_3,a_4)$ around~$u$.
    Let~$x_1,x_2$ be the two neighbors of~$u$ that are linked with~$u$ through~$a_1$ and~$a_3$.
    Similarly, let~$y_1$ and~$y_2$ be the neighbors of~$u$ corresponding to arcs~$a_2$ and~$a_4$.
    We do not assume that~$x_1 \neq x_2$ or~$y_1 \neq y_2$.
    Since~$G$ is biconnected, there is a path from~$x_1$ to~$u$ that does not pass through~$v$.
    This path only contains links in~$S_e$ and similar paths exist between~$x_1$ and~$v$, $x_2$ and~$u$, and~$x_2$ and~$v$.
    Since~$y_1$ and~$y_2$ are not separated from~$t$ by~$e$, there are also paths from~$y_1$ to~$t$ and from~$y_2$ to~$t$ not passing through either~$u$ or~$v$.
    Note that at least one of these paths has to cross at least one of the four paths between~$u/v$ and~$x_1/x_2$.
    Hence, one of these four paths contains a node~$w$ such that a~$w$-$t$-path exists that does not contain~$u$ or~$v$.
    This contradicts the assumption that~$x_1$ and~$x_2$ are separated from~$t$ by~$e=\{u,v\}$.
    Thus, the marked links indeed form a consecutive interval.
    
    We next show that the constructed embedding is indeed planar.
    To this end, consider the face where we added~$(u,v)$.
    Traversing the face alongside the first marked arc (and traversing links in either direction) must lead to~$v$ as each arc corresponds to a link in~$S_e$ until either~$u$ or~$v$ is reached.
    If~$u$ is reached before~$v$, then removing~$u$ separates all links that where traversed from~$t$, that is, $G \setminus e$ is not biconnected.
    This contradicts \cref{lem:stay2con}, so~$v$ is reached.
    Symmetrically, if we traverse the face in the other direction (along the last non-marked link), then we must also reach~$v$ eventually since~$v$ is incident to the considered face and before we reach~$u$ for a second time.
    Thus, all arcs used correspond to links not in~$S_e$.
    In particular at node~$v$, the face is incident to a marked and a non-marked link.
    See \cref{fig:embedding2} for an example.
    When going counterclockwise around~$v$, the marked link comes before the unmarked link and hence this is precisely the spot where we added the link~$(u,v)$ for~$v$.
    The argument for the link~$(v,u)$ is analogous.

    It remains to show that the other conditions of a hierarchical planar embedding are satisfied, that is, the arcs are added in such a way that when traversing two symmetric arcs, then the interior region is traversed in counterclockwise order and for each pair~$f,g$ of links it holds that~$g$ encloses~$g$ if and only if~$g \in S_f$.
    To this end, first note that the arcs~$(u,v)$ and~$(v,u)$ are placed in such a way that the enclosed links are always on the left side, that is, going along these links traverses the interior region in counterclockwise order.
    Since the embedding of links other than~$e$ did not change, it only remains to show that~$e$ is not enclosed by any other link and that~$e$ encloses precisely the links in~$S_e$.

    Let~$f \neq e$ be a link.
    We show that~$f$ is enclosed by~$e$ if and only if~$f \in S_e$.
    To this end, first assume towards a contradiction that~$f \notin S_e$ but~$e$ encloses~$f$.
    Since~$f \notin S_e$, there is a path from an endpoint of~$f$ to~$t$ that does not pass through an endpoint of~$e$. 
    Since this path can never cross the two links corresponding to~$e$, it holds that~$t$ is enclosed by~$e$.
    By the same argument, every node that is not separated from~$t$ by~$e$ is enclosed by~$e$, a contradiction to the construction that leaves an unmarked arc incident to~$u$ in the exterior region.
    Now assume towards a different contradiction that~$f \in S_e$ and~$f$ is not enclosed by~$e$.
    By construction, all arcs corresponding to links incident to~$u$ in~$S_e$ are enclosed by~$e$.
    Since~$f \in S_e$ and since~$G$ is biconnected, there exists a path from~$u$ to an endpoint of~$f$ in~$G$ that does not pass through~$v$.
    As shown above, the first link (incident to~$u$) in this path is enclosed by~$e$.
    Hence, the path must leave the interior region of~$e$ somewhere.
    Note that this implies that it contains~$u$ or~$v$ or crosses one of the two arcs~$(u,v)$ and~$(v,u)$.
    All of these options lead to a contradiction as we showed above that the embedding is planar, $v$ is not part of the path by assumption, and~$u$ is already contained in the path and thus by definition not contained again.
    
    Finally, assume towards a contradiction that~$e$ is enclosed by some link~$f$.
    Then~$f$ encloses all links that are enclosed by~$e$, that is, each link in~$S_e$ as shown above.
    Let~$g \in S_e$ be such a link.
    It then also holds that~$f$ encloses~$g$ in the original hierarchical planar embedding.
    This means that~$g \in S_f$.
    By \cref{lem:separating} and the fact that~$e \notin S_f$, it holds that~$f \in S_e$.
    As shown above, this implies that~$e$ encloses~$f$, a contradiction to the assumption that~$f$ encloses~$e$ (and~$e \neq f$) since two regions in the plane with distinct borders cannot both contain the other.
    This concludes the proof.
\end{proof}

\subsection{Removing Separating Links}
\label{sec:pre3}
We are finally in a position to present our preprocessing routines for \prd{} and \prs.
We again start with \prd{} and show that a rooted graph~$(G,t)$ with a separating link~$e$ is perfectly resilient if and only if~$(G \setminus e,t)$ is.
That is, we can simply remove all separating links.

\begin{lemma}
    \label{lem:noseplink}
    Let~$(G,t)$ be a rooted graph where~$G$ is biconnected and let~$e$ be a separating link in~$G$.
    Then,~$(G,t)$ is perfectly resilient if and only if~$(G \setminus e,t)$ is.
\end{lemma}

\begin{proof}
    The forward direction follows directly from \cref{prop:subgraph}.
    So we focus on the backward direction and assume that~$(G \setminus e,t)$ is perfectly resilient.
    Let~$\pi$ be a perfectly resilient forwarding pattern for~$(G \setminus e,t)$.
    We make a case distinction whether~$e$ is incident to~$t$ or not.
    If~$e$ is incident to~$t$, then let~$e = \{u,t\}$.
    We construct a perfectly resilient forwarding pattern~$\pi^*$ as follows.
    Note that only the neighborhoods of~$u$ and~$t$ are different in~$G$ and~$G \setminus e$.
    For each node~$v \in V \setminus \{u,t\}$, we set~$\pi^*_v = \pi_v$.
    For node~$u$, any in-port~$f \in (E_u \setminus \{e\}) \cup \{\bot\}$, and any set~$F_u$ of incident failed links, we set~$\pi_u^*(f,F_u) = t$ if~$e \notin F_u$ and~$\pi_u^*(f,F_u) = \pi_u(f,F_u)$ otherwise.
    For in-port~$e$, we always return~$e$ as the out-port.

    We next show that the constructed forwarding pattern is perfectly resilient.
    Consider any starting node~$s$ and any set~$F$ of failed links such that an~$s$-$t$-paths remains in~$G \setminus F$.
    Let~$\rho^*$ be the routing corresponding to the constructed forwarding pattern~$\pi^*$, starting node~$s$, and set~$F$ of failed links and let~$\rho$ be the routing corresponding to~$\pi$, $s$, and~$F'=F \setminus \{e\}$.
    If~$e \in F$, then~$\rho = \rho^*$ by construction.
    Since an~$s$-$t$-path exists in~$G \setminus F = (G \setminus e) \setminus F'$ by assumption, $\rho$ contains~$t$ and therefore~$\rho^*$ contains~$t$.
    If~$e \notin F$, then whenever~$\rho^*$ contains~$u$, it contains~$t$ (in the very next step) by construction.
    So it only remains to show that~$\rho^*$ contains~$u$ or~$t$.
    Since~$e$ is a separating link, there is a node~$x$ that is in a different connected component than~$s$ in~$G - e$.
    Since~$G$ is biconnected, there exist two internally disjoint paths between~$s$ and~$x$.
    Note that each of the two paths contains one of~$u$ and~$t$ each.
    Let~$P_1$ and~$P_2$ be the two subpaths from~$u$ to~$x$ and from~$t$ to~$x$, respectively.
    Note that the union of these two paths forms a path from~$u$ to~$t$.
    Let~$E_P$ be the links in that path.
    Consider the routing~$\rho'$ corresponding to~$\pi$, $s$, and~$F''=F' \setminus E_P$.
    Since an~$s$-$t$-paths exists in~$G \setminus F$, there is also an~$s$-$t$-path in~$(G \setminus e) \setminus F''$ as the path~$P$ ensures that~$u$ and~$t$ are still connected.
    Since~$\pi$ is perfectly resilient, this shows that~$\rho'$ contains~$t$.
    Note that~$\rho'$ and~$\rho^*$ only differ once they reach~$u$.
    Thus, $\rho^*$ contains~$u$ or~$t$ and therefore always contains~$t$.
    Hence, $\pi^*$ is perfectly resilient whenever~$e$ is incident to~$t$.
    
    If~$e$ is not incident to~$t$, then let~$e=\{u,v\}$, let~$w$ be a node that is separated from~$t$ by~$e$, and let~$S_e$ be the set of links separated from~$t$ by~$e$.
    We will first construct a path~$P$ between~$u$ and~$v$ that only consists of links that are separated from~$t$ by~$e$.
    Since~$G$ is biconnected and contains three nodes~$u,v,t$, it is also 2-connected.
    Hence, there exist two internally node-disjoint paths~$P_1,P_2$ between~$t$ and~$w$.
    Note that each of the two paths has to contain~$u$ or~$v$ and since they do not share any internal nodes, we can assume without loss of generality that~$P_1$ contains~$u$ and~$P_2$ contains~$v$.
    Consider the subpath of~$P_1$ from~$u$ to~$w$ and the subpath of~$P_2$ from~$w$ to~$v$.
    The gluing of these two paths is a path from~$u$ to~$v$ as both subpaths only intersect in~$w$.
    Let~$P$ be this path and note that each link in~$P$ is separated from~$t$ by~$e$.
    Let~$w_u$ and~$w_v$ be the nodes adjacent to~$u$ and~$v$ in~$P$, respectively, and let~$p_u = \{u,w_u\}$ and~$p_v = \{v,w_v\}$.
    That is, $p_u$ and~$p_v$ are the first and last link in~$P$.
    
    We next construct the perfectly resilient forwarding pattern~$\pi^*$ for~$(G,t)$.
    We simply copy the forwarding function~$\pi_x$ for each node~${x \in V \setminus \{u,v,t\}}$.
    For the node~$u$, a set~$F_u$ of incident failed links, and an in-port~$f \in (E_u \setminus e) \cup \{\bot\}$, we make the following case distinction.
    If~$e \in F_u$, then we set~$\pi^*_u(f,F_u) = \rouns{u}{F_u \setminus \{e\}}{f}$.
    If~$e \notin F_u$ and~$f \notin S_e\cup \{e\}$ (in particular if~$f = \bot$), then we pretend that all links in~$S_e$ except for those in~$P$ fail and whenever~$\pi$ returns~$p_u$, then~$\pi^*$ returns~$e$ instead.
    That is, we set
    \begin{equation*}
        \pi^*_u(f,F_u) =
        \begin{cases}
            \rouns{u}{((F_u \setminus \{e\}) \cup (S_e \cap E_u)) \setminus \{p_u\}}{f}&\text{if }\rouns{u}{((F_u \setminus \{e\}) \cup (S_e \cap E_u)) \setminus \{p_u\}}{f} \notin S_e\\
             e&\text{otherwise.}\\
        \end{cases}
    \end{equation*}
    If~$e \notin F_u$ and~$f \in S_e$, then we set~${\pi^*_u(f,F_u) = \pi^*(\bot,F_u)}$.
    That is, whenever~$\pi$ would reach~$u$ via a link in~$S_e$ and~$e \notin F_e$, then we ``restart'' the routing at~$u$.
    Finally, for the in-port~$f=e$, we model the in-port~$e$ by a scenario where all incident links in~$S_e$ except for~$p_u$ fail and the in-port is~$p_u$ instead.
    That is, we set
    \begin{equation*}
        \pi^*_u(e,F_u) =
        \begin{cases}
            \rouns{u}{((F_u \setminus \{e\}) \cup (S_e \cap E_u)) \setminus \{p_u\}}{p_u}&\hspace*{-2.2mm}\text{if }\rouns{u}{((F_u \setminus \{e\}) \cup (S_e \cap E_u)) \setminus \{p_u\}}{p_u} \notin S_e\\
             e&\hspace*{-2.2mm}\text{otherwise.}\\
        \end{cases}
    \end{equation*} 
    The construction for the node~$v$ is completely symmetric to the construction for~$u$ and therefore omitted.

    It remains to show that the constructed forwarding pattern~$\pi^*$ is indeed perfectly resilient for~$(G,t)$.
    To this end, consider an arbitrary starting node~$s$ and a set~$F$ of failed links such that an~$s$-$t$-path remains in~$G \setminus F$.    
    Let~$\rho^*$ be the corresponding routing and let~$\rho$ be the routing corresponding to~$\pi$, the same starting node~$s$, and the set~$F'=F\setminus \{e\}$ of failed links.
    If~$e \in F$, then note that~$\rho^* = \rho$.
    Hence, by the assumption that~$\pi$ is perfectly resilient for~$(G \setminus e,t)$ and that an~$s$-$t$-paths exists in~$G \setminus F = (G \setminus e) \setminus F'$, $t$ is contained in both routings.
    So we may assume in the following that~$e \notin F$.
    
    We will next show that we may assume without loss of generality that~$s$ is not separated from~$t$ by~$e$.
    To this end, consider the routing~$\rho'$ corresponding to~$\pi$, $s$, and the set~$F' \cap S_e$ of failed links.
    Note that since an~$s$-$t$-path remains in~$G \setminus F$, the same also holds in~$(G \setminus e) \setminus F''$.
    Thus, $\rho'$ contains~$t$ since~$\pi$ is perfectly resilient.
    Since~$e$ separates~$s$ from~$t$, $\rho'$ contains~$u$ or~$v$ (or both).
    Since~$\rho$ and~$\rho'$ start identically until~$u$ or~$v$ is reached, the same also holds for~$\rho$.
    Thus, by construction it holds that~$\rho^*$ is the gluing of two sequences~$\rho^*_1$ and~$\rho^*_2$, where~$\rho^*_1$ is identical to the beginning of~$\rho$ until~$u$ or~$v$ is reached for the first time and~$\rho^*_2$ is the rest.
    Moreover, the link between the last two nodes in~$\rho^*_1$ is contained in~$S_e$.
    By construction, $\rho^*_2$ then behaves as if the routing started in the last node~$u$ or~$v$, which is not separated from~$t$ by~$e$.
    So we may assume in the following that~$s$ is not separated from~$t$ by~$e$.
    
    Since we assume that~$e \notin F$ and $s$ is not separated from~$t$, note that the forwarding pattern~$\pi^*$ will never return any link in~$S_e$ by construction.
    Now consider the routing~$\rho'$ corresponding to~$\pi$, the starting node~$s$ and the set~$F'=(F \cup S_e) \setminus E_P$, where~$E_p$ contains all links in the path~$P$.
    Note that~$s$ is still connected to~$t$ in~$G \setminus F'$ since~$s$ is not separated from~$t$ by~$e$ and~$E_P$ ensures that~$u$ and~$v$ are still connected.
    Hence,~$\rho'$ contains~$t$.
    Moreover, whenever~$\rho'$ enters~$P$, it must traverse the entire path~$P$ by \cref{lem:orbit}.
    Thus, $\rho'$ contains links in~$S_e$ only in the subsequence~$(u,p_u,\ldots,p_v,v)$ and/or~$(v,p_v,\ldots,p_u,u)$ corresponding to traversing~$P$ in either direction.
    Replacing each of these subsequences by the link~$e$ results in a new sequence~$\rho''$.
    To conclude the proof, we will show that~$\rho'' = \rho^*$.
    Note that by the construction of~$F'$ and~$\pi^*$, the routings~$\rho'$ and~$\rho^*$ never diverge at other nodes than~$u$ and~$v$.
    Whenever~$\rho'$ contains~$u$ and the returned out-port is not~$p_u$, then~$\pi^*$ returns the same out-port by construction.
    If~$\pi$ outputs~$p_u$, then~$\rho'$ contains the entire subsequence~$(p_u,\ldots,p_v,v)$ and then the in-port at~$v$ is~$p_v$.
    By construction~$\rho^*$ then uses link~$e$ and the next out-port at~$v$ is by construction identical to the next node in~$\rho'$.
    Thus~$\rho''=\rho^*$, concluding the proof.
\end{proof}

We mention that we are not aware of a way to remove all separating links in a connected planar graph in~$O(n)$ time.
So in order to keep the running time for \prd{} linear in~$n$, we do not remove separating links in a preprocessing step but instead use \cref{lem:noseplink} indirectly.
For \prs, however, we do perform an explicit preprocessing to remove separating links.
This preprocessing will be the main result of this section.
Before we show how to deal with separating links, we first show a helpful lemma that will be extensively used later.

\begin{lemma}
    \label{lem:interestingchange}
    Let~$(G,t)$ be a rooted graph.
    Let~$e=\{u,v\}$ be a link in~$G$, let~$\pi$ be a perfectly resilient skipping forwarding pattern for~$(G \setminus e,t)$, and let~$\pi^*$ be a skipping forwarding pattern for~$G$ that is the result of inserting~$e$ at any position in all priority lists in~$\pi_u$ and~$\pi_v$ (and not changing any of the other list in~$\pi_w$ with~$w \notin \{u,v\}$).
    For any starting node~$s$ and any set~$F$ of failed links where~$e \in F$ and where an~$s$-$t$-path remains in~$G \setminus F$, the routing~$\rho^*$ corresponding to~$\pi^*$, $s$, and~$F$ contains~$t$.
\end{lemma}

\begin{proof}
    Consider the routing~$\rho$ corresponding to~$\pi$, $s$, and~$F'=F \setminus \{e\}$.
    Since an~$s$-$t$-path remains in~$G \setminus F$, the same path also exists in~$(G \setminus e) \setminus F'$.
    Thus~$\rho$ contains~$t$ as~$\pi$ is perfectly resilient.
    Moreover, $\rho^* = \rho$ as~$e \in F$ and~$\pi$ and~$\pi^*$ only differ in~$e$.
    Thus,~$\rho^*$ contains~$t$.
\end{proof}

We next show how to handle the set~$S_t$ of separating links incident to~$t$.

\begin{proposition}
    \label{prop:pre3syn}
    Let~$(G,t)$ be a rooted graph where~$G$ is planar and biconnected and let~$S_t$ be the set of separating links in~$G$ that are incident to~$t$.
    Given a perfectly resilient skipping forwarding pattern for the rooted graph~${(G \setminus S_t,t)}$, we can construct a perfectly resilient skipping forwarding pattern for~$(G,t)$ in~$O(n)$ time.
\end{proposition}

\begin{proof}
    Let~$\pi$ be the perfectly resilient skipping forwarding pattern for~$(G \setminus S_t,t)$.
    For each separating link~${e = \{v,t\} \in S_t}$ and each in-port~$f \neq e$ for~$v$, we add~$e$ to the beginning of~$\pi_v(f)$.
    We also set~${\pi_v(e) = (t) \circ \sigma}$ where~$\sigma$ is any permutation of all other incident links of~$v$.
    We will show that the resulting skipping forwarding pattern~$\pi'$ is perfectly resilient for~$(G,t)$.
    To simplify the argument, we will assume that~$S_t$ contains a single link~$\{v,t\}$ and will not use the fact that~$G$ does not contain any other separating links incident to~$t$.
    Note that we can then simply iteratively apply the same argument for each link in~$S_t$ to show that the result always remains a perfectly resilient skipping forwarding pattern as each time the graph remains planar and biconnected by \cref{lem:stay2con} and the set~$S_t$ of separating links incident to~$t$ does not change (except for the removal of~$\{v,t\}$) by \cref{lem:onego}.
    
    So consider an arbitrary starting node~$s$ and a set~$F$ of failed links such that an~$s$-$t$-paths remains in~$G \setminus F$.
    Let~$\rho'$ be the routing corresponding to~$\pi'$, $s$, and~$F$.
    We will show that~$\rho'$ contains~$t$.
    If~$\{v,t\} \in F$, then $\rho'$ contains~$t$ by \cref{lem:interestingchange}.
    So assume that~$\{v,t\}\notin F$.
    Let~$x$ be a node in a different connected component than~$s$ in~$G-\{v,t\}$.
    Since~$\{v,t\}$ is separating, such a node must exist.
    Let~$V_s$ and~$V_x$ be the set of nodes in the connected components of~$G-\{v,t\}$ containing~$s$ and~$x$, respectively, and let~$G_s = (V^s,E^s) = G[V_s \cup \{v,t\}] \setminus \{v,t\}$ and~$G_x=G[V_x \cup \{v,t\}]\setminus \{v,t\}$.
    Let~$F^s = F \cap E^s$ be the set of links in~$G_s$ that fail.
    Consider the routing~$\rho$ corresponding to~$\pi$, $s$, and~$F^s$.
    We will next show that an~$s$-$t$-path exists in~$(G \setminus \{v,t\}) \setminus F^s$.
    Since~$\pi$ is perfectly resilient, this shows that~$\rho$ contains~$t$.
    Note that the beginning of~$\rho$ and~$\rho'$ are identical until either~$v$ or~$t$ is reached.
    If~$t$ is reached first, then~$\rho'$ contains~$t$ and we are done.
    Otherwise, $v$ is reached by~$\rho'$ and by construction, the first link in the priority list of~$v$ for any in-port is~$\{v,t\}$.
    Since we assume that this link does not fail, $\rho'$ contains~$t$.
    So it only remains to prove that an~$s$-$t$-path exists in~$(G \setminus \{v,t\}) \setminus F^s$.
    Consider an~$s$-$t$-path~$P$ in~$G \setminus F$, which exists by assumption.
    If~$P$ does not contain the link~$\{v,t\}$, then it also exists in~$(G \setminus \{v,t\}) \setminus F^s$.
    So assume that it contains the link~$\{v,t\}$ and let~$P'$ be the subpath between~$s$ and~$v$.
    Note that~$P'$ exists in~$(G \setminus \{v,t\}) \setminus F^s$.
    Since~$G$ is biconnected, $G_x$ is connected and hence there exists a path~$P''$ from~$v$ to~$t$ in~$G_x$.
    Note that the gluing of~$P'$ and~$P''$ is a path from~$s$ to~$t$ in~$(G \setminus \{v,t\}) \setminus F^s$.
    This concludes the proof of correctness.
    
    It remains to analyze the running time.
    Note that for each link~$\{v,t\} \in S_t$, we create a skipping priority list~$\pi_v(e)$ in~$O(|N(v)|)$ time.
    Moreover, we iterate over~$O(|N(v)|)$ possible in-ports and for each we add~$\{v,t\}$ to the beginning in constant time.
    Thus, the running time is upper bounded by~$O(\sum_{v\in V}|N(v)|) = O(m) = O(n)$ by the handshaking lemma and the fact that~$G$ is planar.
\end{proof}

Finally, we next prove the main result of this section.
It provides a~$O(n^2)$-time procedure for \prs{} to deal with all separating links in the input graph.
Note that removing all such links does not affect whether the rooted graph is perfectly resilient or not by \cref{lem:noseplink}.

\begin{theorem}
    \label{thm:pre4syn}
    Let~$(G,t)$ be a rooted graph where~$G$ is planar and biconnected and let~$S$ be the set of separating links in~$G$.
    Given a perfectly resilient right-hand forwarding pattern with a corresponding planar embedding for the rooted graph~${(G \setminus S,t)}$, we can construct a perfectly resilient skipping forwarding pattern for~$(G,t)$ in~$O(n^2)$ time.
\end{theorem}

\begin{proof}
We first compute~$S$ in~$O(n^2)$ time using \cref{obs:computeS}.
Next, we partition~$S$ into~$S_t$ and~$S^*$ where~$S_t$ is the set of all link in~$S$ incident to~$t$ and~$S^*$ is the remaining set of links.
We compute~${G' = G \setminus S_t}$ in~$O(n)$ time.
We will show in the following how to construct a perfectly resilient skipping forwarding pattern for~$(G',t)$ in~$O(n^2)$ time.
\cref{prop:pre3syn} then yields a perfectly resilient skipping forwarding pattern for~$(G,t)$ in~$O(n)$ time, concluding the proof.
In order to construct a perfectly resilient skipping forwarding pattern for~$(G',t)$, we will iteratively build graphs~$G_0 = G' \setminus S^*,G_1,\ldots,G_{|S^*|}=G'$ such that~$G_i$ is the result of adding a separating link~$e_i$ to~$G_{i-1}$.
We also show that~$G_0$ has a perfectly resilient hierarchical right-hand forwarding pattern and we show how to compute such a pattern for~$(G_i,t)$ based on such a pattern for~$(G_{i-1},t)$.
The above approach unfortunately does not yield a quadratic running time.
So afterwards, we will show how to speed up the computation by computing the same hierarchical right-hand forwarding pattern for~$(G',t)$ but not computing all of the intermediate forwarding patterns.
We conclude the proof by showing that this can be done in~$O(n^2)$ time.

The order in which we add the separating links matters and we next show how we choose this order and how to compute a perfectly resilient hierarchical right-hand forwarding pattern for~$G_0$.
For the former, we use \cref{lem:rootlink} to compute an ordering of~$S^*$.
It will be convenient for us to consider the reverse order.
So let~$(e_1,e_2,\ldots,e_{|S^*|})$ be this reverse ordering.
It holds for each~$e_i,e_j \in S^*$ that if~$e_j \in S_{e_i}$, then~$i > j$.
Now, we build the sequence of graphs and right-hand forwarding patterns as described above starting with~${G_0 = G' \setminus S^*}$.
Note that~$G_0$ is biconnected by \cref{lem:stay2con}.
Moreover, by the premise of the theorem, we are given a planar embedding for~$G_0$ and a corresponding perfectly resilient right-hand forwarding pattern~$\lambda_v^{e_v}$ for~$(G_0,t)$.
By \cref{obs:hier}, we can compute in~$O(n)$ time a hierarchical planar embedding for~$G_0$ and a corresponding perfectly resilient hierarchical right-hand forwarding pattern~$\Lambda_v^{e_v}$.

We next show how to compute a perfectly resilient hierarchical right-hand forwarding pattern for any graph~$G_i$ in the sequence.
So let~$G_i$ be any graph in the sequence and assume we are given a hierarchical planar embedding and a corresponding perfectly resilient hierarchical right-hand forwarding pattern~$\Lambda$ for~$(G_{i-1},t)$.
By \cref{lem:stay2con}, $G_i$ is biconnected.
Moreover,~$e_i \notin S_f$ for any link~$f$ in~$G_{i-1}$ by \cref{lem:onego,lem:rootlink}.
Let~$e_i = \{u,v\}$.
We can use \cref{lem:hierembedding} to find position for~$(u,v)$ and~$(v,u)$ in~$O(n)$ time such that adding these arcs at the respective positions yields a hierarchical embedding of~$G_i$.
We next show how to compute a perfectly resilient hierarchical right-hand forwarding pattern~$\Lambda^*$ for~$(G_i,t)$.
We make a case distinction whether~$N(t)=\{u,v\}$ or not.
If this is the case, then we add~$e_1$ to the beginning of the priority part of~$\Lambda_v^{e_v}(f)$ and~$\Lambda_u^{e_u}(g)$ for any in-ports~$f$ and~$g$, that is, at the second position (after the respective link to~$t$).
Note that this takes~$O(m) = O(n)$ time.
Consider any starting node~$s$ and any set~$F$ of failed links such that a path between~$s$ and~$t$ remains in~$G \setminus F$.
Let~$\rho$ be the routing corresponding to the computed hierarchical right-hand forwarding pattern, $s$, and~$F$.
We will show that~$\rho$ contains~$t$.
If~$e_i \in F$, then~$\rho$ contains~$t$ by \cref{lem:interestingchange}.
If~$e_i \notin F$, then let~$F' = F \setminus \{\{u,t\},\{v,t\}\}$.
Consider the routing~$\rho'$ corresponding to~$\Lambda$, $s$, and~$F'$.
Note that since an~$s$-$t$-path remains in~$G \setminus F$, there is a path from~$s$ to~$u$ or~$v$ in~$(G \setminus e) \setminus F'$.
Hence, there is also a path from~$s$ to~$t$ in the same graph and~$\rho'$ therefore contains~$t$ as~$\pi$ is perfectly resilient.
Note that since~$N(t)=\{u,v\}$, it holds that~$\rho'$ contains~$u$ or~$v$.
Since~$\rho$ and~$\rho'$ can only differ after they reached~$u$ or~$v$, also~$\rho$ contains~$u$ or~$v$.
We assume without loss of generality that~$u$ is the first node in~$\{u,v\}$ that is contained in~$\rho$.
If~$\{u,t\} \notin F$, then~$\rho$ contains~$t$ in the next step.
Otherwise, by construction of~$\Lambda^*$ and since~$e_i \notin F$, it contains~$v$ in the next step.
In the following step it contains~$t$ as the link~$\{v,t\}$ cannot fail without disconnecting~$s$ from~$t$ as~$N(t)=\{u,v\}$ and~$\{u,t\} \in F$.

So assume in the following that~$N(t) \neq \{u,v\}$.
We then compute an arbitrary path~$P$ between~$u$ and~$v$ using only links in~$S_{e_i}$.
Note that such a path exists as~$G_i$ is biconnected.
Let~$w$ be any node in that path other than~$u$ or~$v$.
Note that~$w$ is separated from~$t$ by~$e_i$.
Next, we compute the routing~$\rho'$ corresponding to~$\Lambda$ for the starting node~$w$ and the set~$F=E \cap \{\{u,t\},\{v,t\}\}$ of failed links.
Note that~$\rho'$ contains~$t$ as~$N(t) \neq \{u,v\}$.
Hence, each link is taken at most twice by the routing (at most once in each direction as otherwise~$\rho'$ never ends as it endlessly cycles through the same sequence of links).
Moreover, since~$F=E \cap \{\{u,t\},\{v,t\}\}$, the next node in the routing can be computed in constant time (it is one of the first two element in the respective priority list).
Let~$f$ be the first link that is used by the routing~$\rho'$ that is not contained in~$S_{e_i}$.
Note that~$f$ is incident to~$u$ or~$v$, but not both.
Let~$\mu(e_i) \in \{u,v\}$ be that node.
We next show how to construct~${\Lambda^*}_x^{e_x}$ for each node~$x \in V \setminus \{t\}$ ($e_x$ never changes).
For~$x \notin \{u,v\}$, we set~${\Lambda^*}_x^{e_x} = \Lambda_x^{e_x}$.
For~$x \in \{u,v\}$ and in-port~$e_i$, we use the updated right-hand rule, that is, for the incoming arc~$(u,v)$ of~$v$ or~$(v,u)$ of~$u$, we list all outgoing arcs except for~$(u,t)$ or~$(v,t)$ in counterclockwise order (the links towards~$t$ are placed at the beginning).
For~$x = \mu(e_i)$, and any in-port~$f \neq e_i$, we simply insert~$e_i$ into the regular part~$\sigma_r$ of~$\Lambda_x^{e_x}(f)$.
For~$x \in \{u,v\} \setminus \{\mu(e_i)\}$, we do the same for each in-port~$f \notin (S_{e_i} \cup \{e_i\})$ (where~$f=\bot$ is treated equally to~$f=e_x$).
Finally, for~$f \in S_{e_i}$, we add~$e_i$ to the beginning of the priority part~$\sigma_p$ in~${\Lambda^*}_x^{e_x}$, that is, we add~$e$ to the very beginning if~$\{\mu(e_i),t\} \notin E$ and in the second position otherwise.
This concludes the construction.

Recall that~$P$ is a path between~$u$ and~$v$ where all links are contained in~$S_{e_i}$ and where~$w$ is a node on~$P$.
Let~$E_P$ be the set of links of~$P$.
We next show that the result~$\Lambda^*$ is perfectly resilient for~$(G_i,t)$.
To this end, consider an arbitrary starting node~$s$ and a set~$F$ of failed links such that an~$s$-$t$-path remains in~$G_i \setminus F$.
Let~$\rho^*$ be the corresponding routing.
We will show that~$\rho^*$ contains~$t$.
If~$e_i \in F$, then this follows from \cref{lem:interestingchange}.
So assume that~$e_i \notin F$.
Recall that~$e_s$ is the link incident to~$s$ such that~$\Lambda_s(\bot) = \Lambda_s(e_s)$.
We make a case distinction whether~$e_s \in S_{e_i}$ or not and first consider the case where~$e_s \notin S_{e_i}$.
In this case, $\rho^*$ never uses any link in~$S_{e_i}$ as shown next.
Assume towards a contradiction that some link in~$S_{e_i}$ is used by~$\rho^*$ and let~$h$ be the first such link.
Note that~$h$ is used to leave~$u$ or~$v$.
Since both cases are symmetric, we assume that~$h$ is used to leave~$u$.
Let~$h'$ be the link corresponding to the in-port of~$u$ that led to out-port~$h$ (where~$h'=e_u$ if the in-port is~$\bot$).
By definition, $h'$ is not in~$S_{e_i}$.
Hence, $e_i$ is by construction inserted in the regular part of~$\Lambda_u^{e_u}(h')$.
Note that this contradicts the construction of~$\Lambda^*$ as all links in~$S_{e_i}$ appear after~$e_i$ in~${\Lambda^*}_u^{e_u}(h')$ as~$h' \notin S_{e_i}$.
Hence, $\rho^*$ does not use any links in~$S_{e_i}$ and can therefore not distinguish between the failure set~$F$ and~$F'=(F \cup S_{e_i}) \setminus E_P$.
Consider the routing~$\rho$ corresponding to~$\Lambda$, $s$, and~$F'' = F' \setminus \{e_i\}$ (in the graph~$G_{i-1}$).
It is now easy to verify that~$\rho$ and~$\rho^*$ are identical except for the fact that whenever~$\rho^*$ uses the link~$e_i$, then~$\rho$ uses the links in~$E_P$.
The next link afterwards is then again identical by construction.
Since~$s$ and~$t$ are the same connected component in~$G_i \setminus F$, they are also in the same connected component in~$G_{i-1} \setminus F''$.
Thus, $\rho$ contains~$t$ by the fact that~$\Lambda$ is perfectly resilient and thus the same holds for~$\rho^*$.

Now consider the case where~$e_s \in S_{e_i}$.
In this case, we consider two additional routings~(in~$G_{i-1}$).
The first routing~$\rho$ corresponds to the forwarding pattern~$\Lambda$, the starting node~$s$, and the failed link set~${F' = F \cap S_{e_i}}$.
The second routing~$\rho'$ corresponds to the forwarding pattern~$\Lambda$, the starting node~$w$, and the set~$F'' = (F \cup S_{e_1}) \setminus E_P$ of failed links.
It is easy to verify that~$s$ and~$t$ are in the same connected component in~$G_{i-1}\setminus F'$ and~$w$ and~$t$ are in the same connected component in~$G_{i-1} \setminus F''$.
Thus, $\rho$ and~$\rho'$ both contain~$t$.
Note that~$\rho^*$ and~$\rho$ start identically until~$\rho^*$ uses a link not in~$S_{e_i}$ for the first time (potentially~$e_i$).
This either happens when node~$\mu(e_i)$ is reached or when the other endpoint of~$e_i$ is reached via an in-port that then selects a link outside of~$S_{e_i}$ (potentially~$e_i$ if no other such links remain in~$G_i \setminus F$) via the right-hand-rule selection in the regular part~$\sigma_r$.
In either case, the routing~$\rho^*$ uses the first non-failed link in~${\Lambda^*}_{\mu(e_i)}(e_i)$ in the next/second to next step.
Let~$f$ be that link.
We next show that that~$\rho'$ also uses link~$f$ (in the same direction).
When starting in~$w$ and with failure set~$F''$, the routing uses the links in~$E_P$ by definition until~$\mu(e_i)$ is reached.
There, the next link is~$f$ by construction of~$\Lambda^*$ (and~$\mu(e_i)$).
After~$\rho^*$ used the link~$f$, it never uses another link in~$S_{e_i}$ again by the same argument as in the case where~$e_s \notin S_{e_i}$.
Moreover, after~$\rho^*$ and~$\rho'$ used the link~$f$, they continue identically except for the fact that whenever~$\rho^*$ uses~$e_i$, the routing~$\rho'$ uses all links in~$E_P$ once.
The next node in both sequences is then by construction of~$\Lambda^*$ the same.
Since~$\rho'$ contains~$t$, so does~$\rho^*$.
This shows that~$\Lambda^*$ is perfectly resilient.

It remains to show how to compute~$\Lambda^*$ for the graph~$G'=G_{|S^*|}$ in~$O(n^2)$ time.
Note that we can compute~$\mu(e_i)$ for each~$i$ in~$O(n)$ time, that is, we can compute~$\mu$ for all links in overall~$O(n^2)$ time.
Moreover, the hierarchical planar embedding of~$G'$ is computed in~$O(n^2)$ time by \cref{lem:hierembedding}.
Given the embedding, we compute~$\Lambda^*$ in~$O(n^2)$ time as follows.
For each combination of a node~$u$ and an in-port~$f=\{u,v\}$ for~$u$ (replace~$f$ with~$e_u$ if the in-port is~$\bot$), we find the link~$(v,u)$, go counterclockwise around~$u$, and add each outgoing arc~$(u,w)$ to either the~$t$-prefix (if~$w = t$) or the end of the regular part if~$f \notin S_{\{u,w\}}$ or~$\mu(\{u,w\}) \neq u$.
Note that this takes~$O(|N(u)|) \subseteq O(n)$ time for each combination of a node and one of its in-ports and thus~$O(n^2)$ time overall by the handshaking lemma.
Finally, we iterate over all links in~$S^*$ in the (reverse) order computed in the beginning by \cref{lem:rootlink} and for each link~$e$, we add~$e$ to the beginning of all priority lists for~$\mu(e)$ with in-ports in~$S_e$.
This takes constant time per priority list and therefore~$O(|N(\mu(e))|) \subseteq O(n)$ time for each link~$e$.
Thus, this step also takes~$O(n^2)$ time overall.
Finally, we can concatenate~$\sigma_t$,~$\sigma_p$, and~$\sigma_r$ for each combination of a node and one of its possible in-ports in constant time each and in~$O(n)$ time overall.
Note that this results in the same hierarchical right-hand forwarding pattern~$\Lambda^*$ for~$(G',t)$ since each of the three parts~$\sigma_t$, $\sigma_p$, and~$\sigma_r$ in each priority list is computed in the same way: links are divided into the three parts identically, the regular part is still ordered in counterclockwise fashion around the respective node, and links in~$S^*$ are added to the front of the respective priority part in the same order.
This concludes the proof.
\end{proof}

\section{Structural Analysis}
\label{sec:traps}
We are finally in a position to prove our first main result:

\structure*

In order to prove the second main result later, we show a slightly stronger result.
To this end, we first show a last helpful lemma.
Recall that a graph is \nice{} if it is planar, biconnected, and does not contain any separating links and that a minimal trap is a rooted graph that is not perfectly resilient but where each proper rooted minor of it is perfectly resilient.

\begin{lemma}
    \label{cor:nice}
    Let~$(G,t)$ be a minimal trap.
    Then, $G$ is \nice.
\end{lemma}

\begin{proof}
    Assume towards a contradiction that~$G$ is not \nice.
    Then, it is not planar, not biconnected, or contains a separating link.
    \cref{obs:connected} shows that we may assume that~$G$ is connected.
    If~$G$ is not planar, then it contains~$K_5$ or~$K{3,3}$ as a minor by \cref{thm:wagner}.
    Since~$G$ is connected, if~$t$ is not part of the respective graph, then it can be contracted into one the nodes.
    Hence, $(G,t)$ contains~$(K_5,t)$ or~$(K_{3,3},t)$ as a rooted minor.
    Note that~$(G,t)$ is therefore not a minimal trap as it contains the~$\Kf$ or the~$\Kt$ as a proper rooted minor, a contradiction.
    If~$G$ is not biconnected, then it contains a cut node~$v$ and contracting all but one of the connected components in~$G-v$ into~$v$ results in a trap by \cref{prop:pre2dec}.
    Hence, $(G,t)$ again contains a trap as a proper rooted minor, contradicting that~$(G,t)$ is a minimal trap.
    Finally, if~$G$ contains a separating link~$e$, then~$(G \setminus e,t)$ is a trap by \cref{lem:noseplink}---a final contradiction to~$(G,t)$ being a minimal trap.
\end{proof}

We now show our first main result.

\begin{theorem}
    \label{thm:structure}
    A rooted graph~$(G,t)$ is perfectly resilient if and only if it does not contain the~$\Kf$, the~$\Kt$, the~$\sK$, or the~$\Ktf$ as a rooted minor.
    If it is perfectly resilient and \nice, then~$G-t$ is outerplanar, $(G,t)$ is a dipole outerplanar graph, or~$(G,t)$ is a ring of outerplanar graphs.
\end{theorem}
\begin{proof}
    The main part of the proof will be to show that if~$(G,t)$ is \nice, $G-t$ is not outerplanar, $(G,t)$ is not a dipole outerplanar graph, and~$(G,t)$ is not a ring of outerplanar graphs, then it contains one of the four mentioned rooted graphs as minor.
    This claim will conclude the proof as follows.

    We make a case distinction whether~$(G,t)$ is perfectly resilient or not.
    If it is, then we distinguish between the cases where~$G$ is \nice{} and where~$G$ is not \nice.
    If~$G$ is nice and~$(G,t)$ is perfectly resilient, then \cref{cor:rootedminor,prop:notres,prop:sk,prop:ktf} show that none of the four rooted graphs are a rooted minor of~$(G,t)$.
    The contraposition of the above claim yields that~$G-t$ is outerplanar, $(G,t)$ is a dipole outerplanar graphs, or~$(G,t)$ is a ring of outerplanar graphs.
    Note that this is exactly what the theorem states for this case.
    If~$G$ is not nice and~$(G,t)$ is perfectly resilient, then the theorem does not make any claim so it trivially holds.

    If~$(G,t)$ is not perfectly resilient, then it contains a minimal trap~$(H,t)$ (potentially~$H=G$).
    Moreover, this trap is nice by \cref{cor:nice}.
    Note that~$H-t$ is not outerplanar by \cref{thm:gt}, $(H,t)$ is not a dipole outerplanar graph by \cref{prop:dipole}, and~$(H,t)$ is not a ring of outerplanar graphs by \cref{prop:ring}.
    Hence, the above claim yields that one of the four mentioned rooted graphs is a rooted minor of~$(H,t)$ and therefore also a rooted minor of~$(G,t)$.
    This is again what the theorem states.

    So it remains to show the claim.
    To this end, assume that~$(G,t)$ is a rooted graph, where~$G$ is \nice, $G-t$ is not outerplanar, $(G,t)$ is not a dipole outerplanar graph, and~$(G,t)$ is not a ring of outerplanar graphs.
    Since~$G$ is \nice, it is planar and there exists a planar embedding of~$G$ where~$t$ belongs to the outer face.
    Let~$z$ be a node that does not belong to the outer face in the assumed embedding.
    Note that such a node exists as~$G-t$ is not outerplanar.
    Let~$G'=(V',E')$ be the biconnected component of~$G-t$ containing~$z$ and let~$C$ be the set of nodes of~$G'$ that are incident to the outer face in~$G'$ (using the same embedding).
    Note that~$z \notin C$ and~$|C| \geq 3$ as the nodes of~$C$ form the outer face of a 2-connected graph.
    Hence, they form a simple cycle by \cref{prop:cycle}.    
    Let~$E_C$ be the set of links incident to the outer face in~$G'$ and let~$G'' = G' - C$ be the induced subgraph of~$G'$ excluding the nodes on the outer face.
    For a node~$x$ in~$G''$, let~$V_x$ be the set of all nodes in the connected component of~$G''$ containing~$x$ and let~$G_x$ be the induced subgraph of~$G'$ containing all nodes in~$V_x$ and all nodes that are adjacent to nodes in~$V_x$ in~$G$.
    Note that the latter set of nodes only contains nodes from~$C$.
    Let~$P \subseteq C$ be the set of nodes in~$C$ such that for each~$p \in P$, it holds that~$\{p,t\} \in E$ ($p$ is an access node) or~$p$ is a cut node in~$G-t$ ($p$ is a fracture node).
    Note that since~$G$ is biconnected, all cut nodes in~$G-t$ are incident to the outer face and hence there are no fracture nodes outside~$P$ in~$G'$.
    For an example, see \cref{fig:triangle}.
    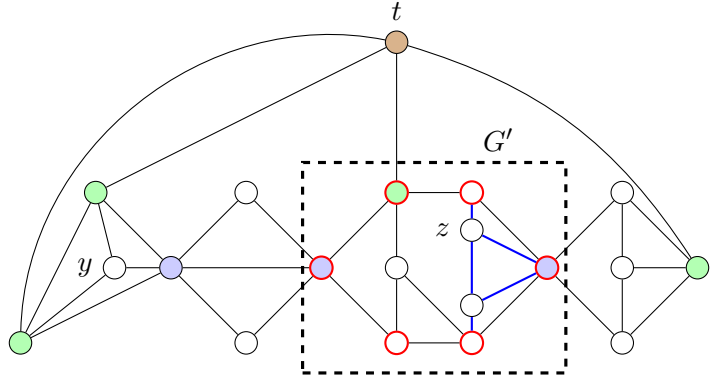
\begin{figure}
        \centering
        \begin{tikzpicture}
            \node[te] at(0,3) (t) {};
            \node[por] at(-5,-1) (a1) {} edge[bend left=50](t);
            \node[por] at(-4,1) (a2) {} edge(a1) edge(t);
            \node[ter] at(-3,0) (a3) {} edge(a1) edge(a2);
            \node[ver] at(-2,1) (a4) {} edge(a3);
            \node[ver] at(-2,-1) (a5) {} edge(a3);
            \node[ver,label=left:$y$] at(-3.75,0) (a6) {} edge(a1) edge(a2) edge(a3);
            \node[ter,draw=red,thick] at(-1,0) (b1) {} edge(a3) edge(a4) edge(a5);
            \node[por,draw=red,thick] at(0,1) (b2) {} edge(b1) edge(t);
            \node[ver,draw=red,thick] at(1,1) (b3) {} edge(b2);
            \node[ter,draw=red,thick] at(2,0) (b4) {} edge(b3);
            \node[ver,draw=red,thick] at(1,-1) (b5) {} edge(b4);
            \node[ver,draw=red,thick] at(0,-1) (b6) {} edge(b5) edge(b1);
            \node[ver] at(3,1) (c1) {} edge(b4);
            \node[ver] at(3,-1) (c2) {} edge(b4);
            \node[por] at(4,0) (c3) {} edge(c2) edge(c1) edge[bend right=20](t);
            \node[ver] at(3,0) (c4) {} edge(c1) edge(c2) edge(c3);
            \node[ver] at(0,0) {} edge(b2) edge(b6) edge(b5);
            \node[ver,label=left:$z$] at(1,.5) (z) {} edge[thick,blue](b3) edge[thick,blue](b4);
            \node[ver] at(1,-.5) {} edge[thick,blue](z) edge[thick,blue](b4) edge[thick,blue](b5);
            \node[very thick,dashed,rectangle,draw,minimum width=3.5cm,minimum height=2.8cm,inner ysep=0pt,label=70:$G'$] at(0.5,0) {};
        \end{tikzpicture}
        \caption{A nice graph~$G$ with a root~$t$. The links in~$G_z$ are drawn in blue and the nodes in~$C$ are marked with a red border. Green nodes are access nodes and blue nodes are fracture nodes. The nodes in~$P$ are the nodes that are green or blue and have a red border.}
        \label{fig:triangle}
    \end{figure}%
    For the sake of notational convenience, let~$C=\{c_1,c_2,\ldots,c_{k}\}$, where~$k = |C|$, $\{c_i,c_{i+1}\} \in E_C$ for each~$i \in [k-1]$, and~$\{c_k,c_1\} \in E_C$.
    To avoid case distinctions, we also set~$c_0=c_k$ and~$c_{k+1}=c_1$.
    Finally, note that~$|P| \geq 2$ as if~$|P|=\emptyset$, then~$t$ is not connected to~$z$ in~$G$, that is, the graph is disconnected, and if~$|P|=1$, then the unique node in~$P$ is a cut node in~$G$, contradicting that~$G$ is \nice.
    
    In the following, we make multiple nested case distinctions and show in each case that one of the four rooted graphs stated in the claim is a rooted minor of~$(G,t)$.
    \cref{fig:overview} shows an overview of these case distinctions for reference.
    Some of the conditions require additional definitions based on the previous cases.
    Those will be introduced throughout the proof.
    \cref{fig:obstructions2} depicts the four mentioned rooted graphs (the same as in \cref{fig:obstructions}).
    Since most of these rooted graphs appear multiple times throughout the proof where the same node has different names, we only included names that are consistent throughout the proof in all cases.
    Instead, all nodes that are the result of contracting links in~$E_C$ are marked with a red border.
    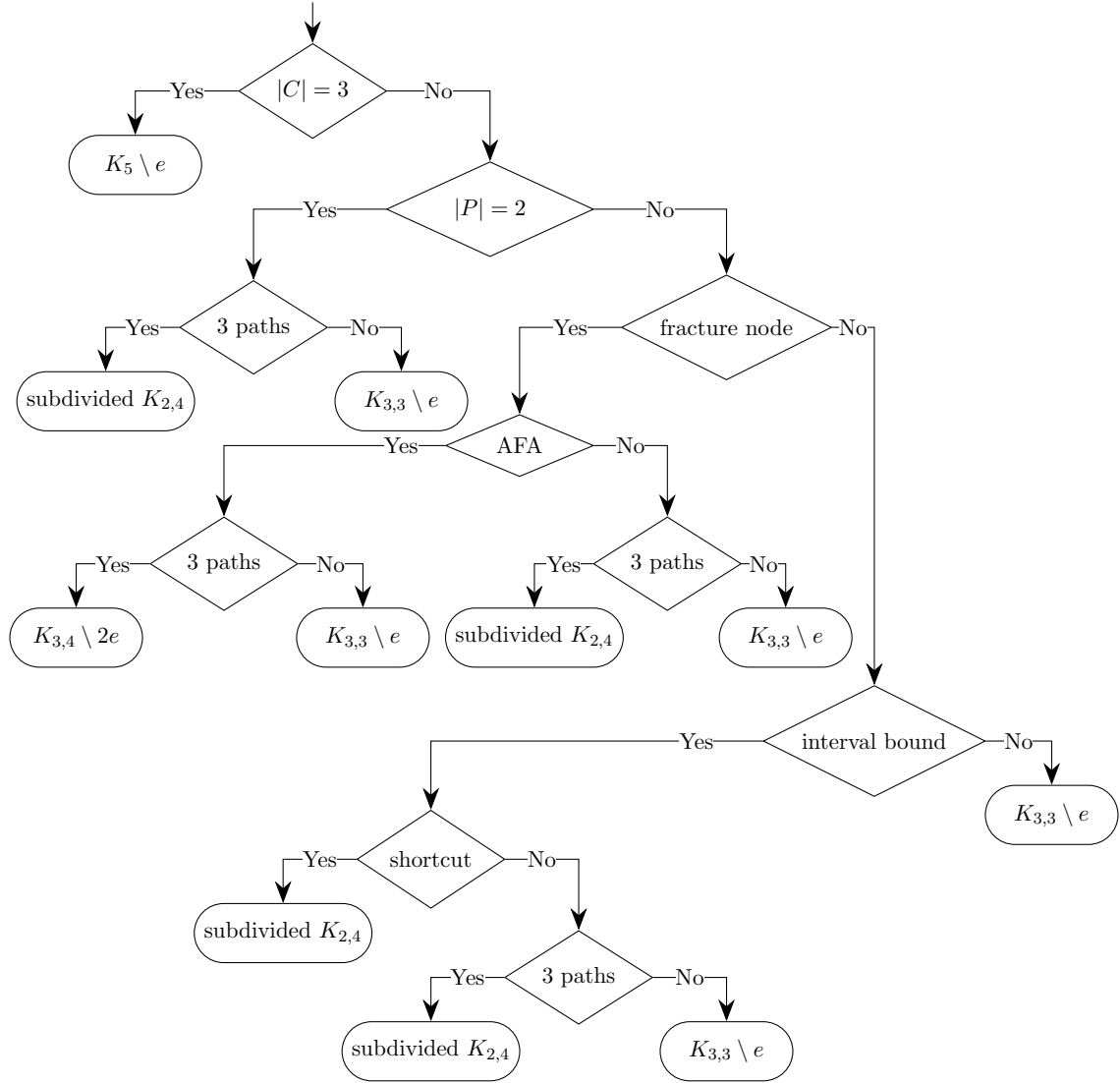
\begin{figure}[t]
        \centering
        \begin{tikzpicture}[scale=.8,every node/.style={scale=0.8}]
            \tikzset{>={Latex[width=3mm,length=3mm]}}
            
            \node[diamond,draw,minimum width=2.5cm,inner sep=0] at (-1,10) (con1) {$|C| = 3$};
            \node[draw,rounded rectangle,minimum width=2.5cm,minimum height=1cm] at (-4,8.75) (res1) {\Kf};
            \node[diamond,draw,minimum width=3.5cm,inner sep=0] at (2,8) (con2) {$|P|=2$};
            \node[diamond,draw,minimum width=2.5cm,inner sep=0] at (-2,6) (con3) {3 paths};
            \node[diamond,draw,minimum width=2.5cm,aspect=2] at (6,6) (con4) {fracture node};
            \node[draw,rounded rectangle,minimum width=2.5cm,minimum height=1cm] at (-4.5,4.75) (res2) {\sK};
            \node[draw,rounded rectangle,minimum width=2.5cm,minimum height=1cm] at (.5,4.75) (res3) {\Kt};
            \node[diamond,draw,minimum width=2.5cm,inner sep=0] at (2.5,4) (con6) {AFA};
            \node[diamond,draw,minimum width=2.5cm,aspect=2] at (8.5,-1) (con8) {interval bound};
            \node[diamond,draw,minimum width=2.5cm,inner sep=0] at (-2.5,2) (con9) {3 paths};
            \node[diamond,draw,minimum width=2.5cm,inner sep=0] at (5,2) (con10) {3 paths};
            \node[draw,rounded rectangle,minimum width=2.5cm,minimum height=1cm] at (-5,.75) (res4) {\Ktf};
            \node[draw,rounded rectangle,minimum width=2.5cm,minimum height=1cm] at (-.15,.75) (res5) {\Kt};
            \node[draw,rounded rectangle,minimum width=2.5cm,minimum height=1cm] at (2.75,.75) (res6) {\sK};
            \node[draw,rounded rectangle,minimum width=2.5cm,minimum height=1cm] at (7,.75) (res7) {\Kt};
            \node[draw,rounded rectangle,minimum width=2.5cm,minimum height=1cm] at (11.5,-2.25) (res8) {\Kt};
            \node[diamond,draw,minimum width=2.5cm,inner sep=0] at (1,-3) (con7) {shortcut};
            \node[draw,rounded rectangle,minimum width=2.5cm,minimum height=1cm] at (-1.5,-4.25) (res9) {\sK};
            \node[diamond,draw,minimum width=2.5cm,inner sep=0] at (3.5,-5) (con11) {3 paths};
            \node[draw,rounded rectangle,minimum width=2.5cm,minimum height=1cm] at (1,-6.25) (res10) {\sK};
            \node[draw,rounded rectangle,minimum width=2.5cm,minimum height=1cm] at (6,-6.25) (res11) {\Kt};
            
            \draw[-{Stealth[length=3mm]}] (-1,11.5) -| (con1);
            \draw[-{Stealth[length=3mm]}] (con1) -| (res1)
    node[pos=0.25,fill=white,inner sep=0]{Yes};
            \draw[-{Stealth[length=3mm]}] (con1) -| (con2)
    node[pos=0.25,fill=white,inner sep=0]{No};
            \draw[-{Stealth[length=3mm]}] (con2) -| (con3)
    node[pos=0.25,fill=white,inner sep=0]{Yes};
            \draw[-{Stealth[length=3mm]}] (con2) -| (con4)
    node[pos=0.25,fill=white,inner sep=0]{No};
            \draw[-{Stealth[length=3mm]}] (con3) -| (res2)
    node[pos=0.25,fill=white,inner sep=0]{Yes};
            \draw[-{Stealth[length=3mm]}] (con3) -| (res3)
    node[pos=0.25,fill=white,inner sep=0]{No};
            \draw[-{Stealth[length=3mm]}] (con4) -| (con6)
    node[pos=0.25,fill=white,inner sep=0]{Yes};
            \draw[-{Stealth[length=3mm]}] (con4) -| (con8)
    node[pos=0.25,fill=white,inner sep=0]{No};
            \draw[-{Stealth[length=3mm]}] (con6) -| (con9)
    node[pos=0.1,fill=white,inner sep=0]{Yes};
            \draw[-{Stealth[length=3mm]}] (con6) -| (con10)
    node[pos=0.25,fill=white,inner sep=0]{No};
            \draw[-{Stealth[length=3mm]}] (con9) -| (res4)
    node[pos=0.25,fill=white,inner sep=0]{Yes};
            \draw[-{Stealth[length=3mm]}] (con9) -| (res5)
    node[pos=0.25,fill=white,inner sep=0]{No};
            \draw[-{Stealth[length=3mm]}] (con10) -| (res6)
    node[pos=0.25,fill=white,inner sep=0]{Yes};
            \draw[-{Stealth[length=3mm]}] (con10) -| (res7)
    node[pos=0.25,fill=white,inner sep=0]{No};
            \draw[-{Stealth[length=3mm]}] (con8) -| (con7)
    node[pos=0.1,fill=white,inner sep=0]{Yes};
            \draw[-{Stealth[length=3mm]}] (con8) -| (res8)
    node[pos=0.25,fill=white,inner sep=0]{No};
            \draw[-{Stealth[length=3mm]}] (con7) -| (res9)
    node[pos=0.25,fill=white,inner sep=0]{Yes};
            \draw[-{Stealth[length=3mm]}] (con7) -| (con11)
    node[pos=0.25,fill=white,inner sep=0]{No};
            \draw[-{Stealth[length=3mm]}] (con11) -| (res10)
    node[pos=0.25,fill=white,inner sep=0]{Yes};
            \draw[-{Stealth[length=3mm]}] (con11) -| (res11)
    node[pos=0.25,fill=white,inner sep=0]{No};
        \end{tikzpicture}
        \caption{Overview of the case distinctions in the proof of \cref{thm:structure} and which rooted minor is found in each case.}
        \label{fig:overview}
    \end{figure}%
    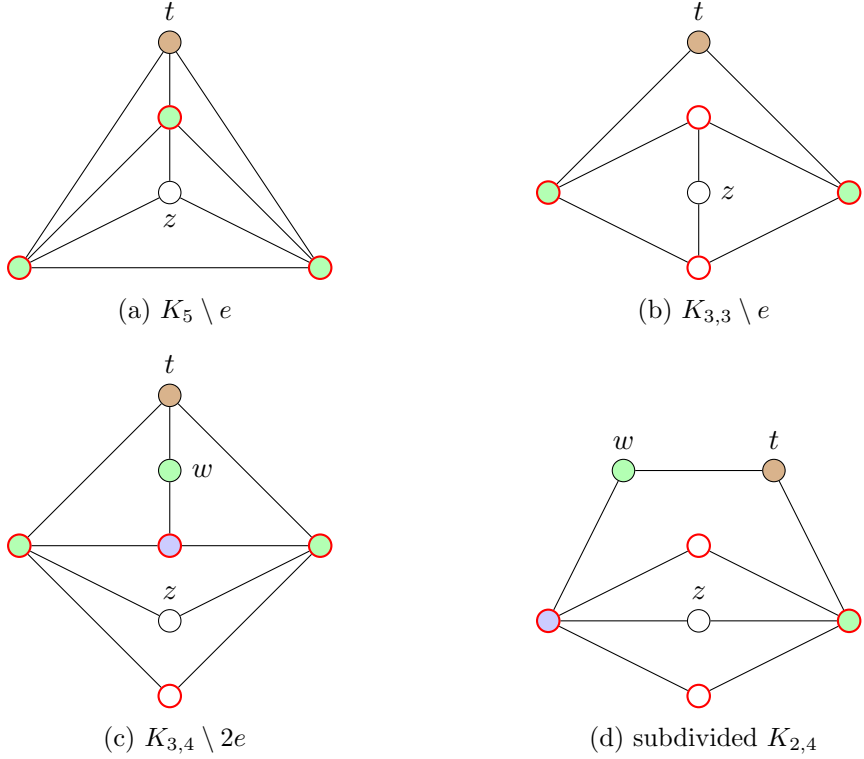
\begin{figure}[t]
        \centering
        \hfill
        \begin{subfigure}{0.27\textwidth}
            \begin{tikzpicture}
                \node[ver,label=below:$z$] at(-3,5) (s) {};
                \node[te] at(-3,7) (t) {};
                \node[por,draw=red,thick] at(-5,4) (u) {} edge(s) edge(t);
                \node[por,draw=red,thick] at(-1,4) (v) {} edge(s) edge(t) edge(u);
                \node[por,draw=red,thick] at(-3,6) (w) {} edge(s) edge(t) edge(u) edge(v);
            \end{tikzpicture}
            \caption{\Kf}
        \end{subfigure}
        \hfill
        \begin{subfigure}{0.27\textwidth}
            \begin{tikzpicture}
                \node[te] at(3,7) (t) {};
                \node[por,draw=red,thick] at(1,5) (u) {} edge(t);
                \node[por,draw=red,thick] at(5,5) (v) {} edge(t);
                \node[ver,draw=red,thick] at(3,6) (p) {} edge(u) edge(v);
                \node[ver,draw=red,thick] at(3,4) (q) {} edge(u) edge(v);
                \node[ver,label=right:$z$] at(3,5) {} edge(p) edge(q);            
            \end{tikzpicture}
            \caption{\Kt}
        \end{subfigure}
        \hfill\hfill
        
        \medskip
    
        \hfill
        \begin{subfigure}{0.27\textwidth}
            \begin{tikzpicture}
                \node[te] at(3,2) (t) {};
                \node[por,label=right:$w$] at(3,1) (w) {} edge(t);
                \node[por,draw=red,thick] at(1,0) (u) {} edge(t);
                \node[por,draw=red,thick] at(5,0) (v) {} edge(t);
                \node[ter,draw=red,thick] at(3,0) {} edge(u) edge(v) edge(w);
                \node[ver,label=$z$] at(3,-1) {} edge(u) edge(v);
                \node[ver,draw=red,thick] at(3,-2) {} edge(u) edge(v);
            \end{tikzpicture}
            \caption{\Ktf}
        \end{subfigure}
        \hfill
        \begin{subfigure}{0.27\textwidth}
            \begin{tikzpicture}
                \node[por,label=$w$] at(-4,2) (w) {};
                \node[te] at(-2,2) (t) {} edge(w);
                \node[ter,draw=red,thick] at(-5,0) (u) {} edge(w);
                \node[por,draw=red,thick] at(-1,0) (v) {} edge(t);
                \node[ver,draw=red,thick] at(-3,1) {} edge(u) edge(v);
                \node[ver,label=$z$] at(-3,0) {} edge(u) edge(v);
                \node[ver,draw=red,thick] at(-3,-1) {} edge(u) edge(v);            
            \end{tikzpicture}
            \caption{\sK}
        \end{subfigure}
        \hfill\hfill
        \caption{The four mentioned rooted graphs. Access nodes are green, fracture nodes are blue, and the red border shows the nodes in~$C$.}
        \label{fig:obstructions2}
    \end{figure}
    
    We start with a case distinction based on the size of~$C$.
    Since there is a simple cycle with node set~$C$, it holds that~${|C|\geq 3}$.
    We first deal with the case where~$|C|=3$.
    In this case, note that all three nodes in~$C$ are pairwise adjacent.
    Since~$G$ is \nice, $P=C$ and~$G_z$ contains all three nodes in~$P$, as if one node in~$C$ is (i) neither a fracture node nor an access node or (ii) not contained in~$G_z$, then the link between the two remaining nodes is a separating link (separating~$t$ from~$z$).
    If we now contract everything outside~$G'$ into~$t$ and all nodes in~$V_z$ into~$z$, then the resulting rooted minor is (potentially a supergraph\footnote{In the following, we will not explicitly mention that we potentially remove certain nodes that are not relevant for forming the rooted minor.} of)~$\Kf$.
    As a simple example, consider the node~$y$ in \cref{fig:triangle}.
    Note that all nodes in~$V_z$ induce a connected graph by definition.
    Moreover, no node outside of~$G'$ can be disconnected from~$t$ by removing~$G'$ as each such node is by definition of~$G'$ only connected via a single node~$x$ in~$G'$ (as otherwise~$x$ would belong to the same biconnected component).
    If there is no other path to~$t$, then~$x$ is a cut node in~$G$, contradicting that~$G$ is \nice.
    Now, the three nodes in~$C$, the result of contracting all nodes in~$V_z$, and the result of contracting the outside into~$t$ forms the~\Kf.
    For the sake of readability, we will here and in the following refer to nodes that are the result of a contraction by any name of a node that was incident to one of the contracted links.
    To see that the five mentioned nodes form the \Kf, note that only the link between~$t$ and~$z$ is missing as the three nodes in~$C$ form a triangle and are adjacent to both~$t$ and~$z$ in the resulting graph.
    This concluding the case where~$|C|=3$.
    
    So we assume in the following that~$|C|>3$.
    Before we continue with the remaining cases, we first introduce one more definition, explain what the condition \emph{3 paths} in \cref{fig:overview} is, and show that if this condition is not met, then we can always find the~\Kt{} as a rooted minor.
    Let~$c_p,c_q \in P$ with~$p < q-1$.
    Let~$U'$ and~$L'$ be the two paths between~$c_p$ and~$c_q$ along the outer face of~$G'$, that is, $U' = (c_{p},c_{p+1},\ldots,c_{q-1},c_{q})$ and~$L'=(c_{p},c_{p-1},\ldots,c_1,c_k,c_{k-1},\ldots,c_{q})$.
    Let~$U$ and~$L$ be the same paths without the endpoints~$c_p$ and~$c_q$.
    Note that~$G_z$ contains at least two nodes~$c_a,c_b$ with~$a < b-1$ in~$C$ as otherwise the link~$\{c_a,c_{a+1}\}$ between the two unique nodes in~$G_z$ in~$C$ would be a separating link.
    Moreover, each node in~$C$ is contained in~$U'$ or~$L'$.
    The condition \emph{3 paths} in \cref{fig:overview} refers to the case that both~$c_a$ and~$c_b$ belong to~$U'$ or both to~$L'$.
    If this is not the case, then note that one of the two belongs to~$L$ and the other to~$U$.
    We will show that in this case the rooted graph~\Kt{} is always a minor of~$(G,t)$.
    Assume without loss of generality that~$c_a$ belongs to~$U$ (the other case is completely symmetric).
    Then, we contract the path~$(c_{a+1},c_{a+2},\ldots,c_{b-1})$ (containing~$c_q$) and the path~$(c_{b+1},c_{b+2}\ldots,c_k,c_1,c_2,\ldots,c_{b-1})$ (containing~$c_p$), and all nodes in~$V_z$ into a single node each.
    We also contract every node outside of~$G'$ into~$t$.
    For the same reason as in the case where~$|C|=3$ and since~$(C,E_C)$ is a simple cycle, all of these contractions are valid.
    The resulting rooted graph is the \Kt{} as the nodes~$c_p$ and~$c_q$ have common neighbors~$c_{a}$, $c_b$, and~$t$.
    Moreover, both~$c_a$ and~$c_b$ are adjacent to~$z$ by definition.

    We now continue with the case distinctions.
    We next consider the two cases~$|P|=2$ and~$|P| \geq 3$.
    Note that~$|P| \geq 2$ as~$G$ is biconnected and hence~$t$ cannot be disconnected from the rest of the graph by at most one node removal.
    We first analyze the case where~$|P|\geq 3$.
    We further distinguish between the cases whether~$P$ contains a fracture node or not.
    We start with the case where a fracture node~$c_i \in P$ exists.
    Here we consider the following two subcases: $c_i$ is the only fracture node,~$c_{i-1}$ and~$c_{i+1}$ are both access nodes, and no further nodes exist in~$P$ (we refer to this case as condition \emph{AFA} in \cref{fig:overview}---it stands for ``access, fracture, access'') or there exists a node~$c_j \in P$ such that~$c_i$ and~$c_j$ are not consecutive on~$C$, that is~$|i-j|>1$ and~$\{i,j\} \neq \{1,k\}$.
    Assume first that there is a node~$c_j \in P$ with~$|i-j|>1$ and~$\{i,j\} \neq \{1,k\}$.
    If~$G_z$ contains a node from~$U$ and a node from~$L$, then we have shown above that~\Kt{} is a rooted minor of~$(G,t)$.
    So assume that~$G_z$ contains two nodes~$c_a$ and~$c_b$ with~$a < b-1$ in either of the two paths (we assume without loss of generality~$U$).
    We also assume that~$i < j$ as the other case is symmetric.
    Then, we contract the path~$L$, the path~$(c_i,c_{i+1},\ldots,c_a)$, the path~$(c_{a+1},\ldots,c_{b-1})$, the path~$(c_{b+1},c_{b+2},\ldots,c_{j})$, and all nodes in~$V_z$ into a single node each.
    We also contract every node outside of~$G'$ except for a single neighbor~$w$ of~$c_i$ other than~$t$ and nodes separated by the link~$\{c_i,w\}$ into~$t$.
    Note that such a node~$w$ exists by definition as~$c_i$ is a fracture node.
    Moreover, $w$ is a neighbor of~$t$ in the resulting graph as otherwise~$c_i$ would be a cut node.
    The resulting rooted graph is the \sK{} as the nodes~$c_i$ and~$c_j$ have common neighbors~$c_{a+1}$, $z$, and~$c_{i-1}$.
    Moreover, $c_i$ is adjacent to~$w$ and both~$w$ and~$c_j$ are adjacent to~$t$.
    
    To conclude the case where~$|P| \geq 3$ and~$P$ contains at least one fracture node, we next analyze the case where~$c_i$ is the only fracture node and~$c_{i-1}$ and~$c_{i+1}$ are both access nodes.
    We consider the paths~$L$ and~$U$ between~$c_{i-1}$ and~$c_{i+1}$.
    Note that~$U'=(c_{i-1},c_i,c_{i+1})$ and~$U$ only contains~$c_i$.
    Again, if~$G_z$ contains a node from~$L$ and a node from~$U$, then the no-3-paths argument shows that the~\Kt{} is a rooted minor of~$(G,t)$.
    So assume that~$G_z$ contains two nodes~$c_a$ and~$c_b$ with~$a < b-1$ in~$L'$.
    Let~$w$ be a neighbor of~$c_i$ outside of~$G'$.
    Let without loss of generality be~$i < a$ (the other case is symmetric).
    We contract everything outside~$G'$ except for~$w$ and nodes separated by the link~$\{c_i,w\}$ into~$t$ and everything in~$V_z$ into~$z$.
    We also contract the paths~$(c_{i+1},c_{i+2},\ldots,c_{a-1},c_a)$, $(c_{a+1},c_{a+2},\ldots,c_{b-1})$, and~$(c_{b},c_{b+1},\ldots,c_k,c_1,c_2,\ldots,c_{i-1})$ into a single node each.
    The result is (a supergraph of) the \Ktf{} as~$c_{i-1}$ and~$c_{i+1}$ have common neighbors~$c_i$, $c_{a+1}$, $z$, and~$t$.
    Moreover,~$w$ is adjacent to~$c_i$ and~$t$.
    This concludes the case where~${|P| \geq 3}$ and~$P$ contains at least one fracture node.

    We stay in the case where~$|P| \geq 3$ and now consider the case where~$P$ does not contain any fracture nodes.
    Note that in this case~$G-t$ is biconnected by definition.
    Hence, $G' = G-t$ and~$C$ contains all nodes on the outer face of~$G-t$.
    Since~$P \subseteq C$ by definition, it holds for each~$p_j \in P$ that~$p_j=c_{a_j}$ for some~$a_j \in [k]$.
    Let~$P=\{p_1,p_2,\ldots,p_{|P|}\}$ such that~$a_i < a_j$ whenever~$i < j$.
    We next define \emph{intervals} of~$C$.
    For any~$j \in [|P|-1]$, the~$j$\textsuperscript{th} interval of~$C$ contains all nodes~$c_\ell$ with~$a_j \leq \ell \leq a_{j+1}$.
    The~$|P|$\textsuperscript{th} interval of~$C$ contains all nodes~$c_\ell$ with~$\ell \leq a_1$ or~$\ell \geq a_{|P|}$.
    Note that each node in~$P$ is contained in exactly two intervals and all other nodes in~$C$ are contained in exactly one interval.
    We say that~$(G,t)$ is \emph{interval bound} if the graph~$G_x$ for each~$x \in V \setminus (C \cup \{t\})$ contains nodes from only one interval of~$C$ (it contains at least two nodes from~$C$ as~$G$ is 2-connected).
    
    If~$(G,t)$ is not interval bound, then there exists some node~$z \notin C$ such that~$G_z$ contains nodes from two intervals of~$C$.
    We distinguish between the following three cases:
    (i) $G_z$ contains a node in~$C \setminus P$, (ii) $G_z$ contains at least three nodes in~$P$, and (iii) $G_z$ contains exactly two nodes in~$C$, they both belong to~$P$ and are not in the same interval (not consecutive).
    We start with the case where~$G_z$ contains at least three nodes in~$P$.
    We find the~\Kf{} similar to the case where~$|C|=3$:
    Contracting~$V_z$ into a single node and~$C$ into three nodes (three nodes in~$P$ in~$G_z$) yields the~\Kf{} since the nodes in~$P$ then form a triangle and both~$z$ and~$t$ are adjacent to all three of these nodes.
    So we may assume that~$G_z$ contains at most two nodes in~$P$.
    
    We next consider the case where~$G_z$ contains a node in~$C \setminus P$.
    In this case, the~\Kt{} is a rooted minor of~$(G,t)$.
    Let~$c_i$ be a node in~$G_z$ in~$C \setminus P$ and let~$p_a,p_{a+1} \in P$ be the two nodes in~$P$ in the interval of~$C$ containing~$c_i$.
    The graph~$G_z$ then also contains a node~$c_j$ that does not belong to the unique interval of~$C$ containing~$c_i$ by the choice of~$z$ (recall that we assume~$(G,t)$ to not be interval bound).
    We contract all nodes in~$V_z$ into a single node and the cycle~$(C,E_C)$ into four nodes~$p_a$, $c_i$, $p_{a+1}$, and~$c_j$.
    Note that~$p_a$ and~$p_{a+1}$ have common neighbors~$c_i$, $c_j$, and~$t$ in the resulting graph and~$z$ has neighbors~$c_i$ and~$c_j$, that is, these six nodes form the \Kt.
    
    The last case to consider is that~$G_z$ contains two nodes~$p_a$ and~$p_b$ that are not contained in the same interval.
    Hence~$p_{a+1} \neq p_b$ and~$p_{b+1} \neq p_a$.
    We then contract all nodes in~$V_z$ into a single node and~$(C,E_C)$ into four nodes~$p_a$, $p_{a+1}$, $p_b$ and~$p_{b+1}$.
    Note that~$p_{a+1}$ and~$p_{b+1}$ have common neighbors~$p_a$, $p_b$, and~$t$ and~$z$ has neighbors~$p_a$ and~$p_b$.
    This forms the \Kt.
    So we assume in the following that~$(G,t)$ is interval bound.

    A \emph{shortcut link} is a link~$\{p_i,p_j\}$ with~$|i-j| > 1$ and~$\{i,j\} \neq \{1,|P|\}$.
    If~$G'=G-t$ contains such a shortcut link, then we show that~$(G,t)$ contains the \sK{} as a rooted minor.
    Consider the paths~$U$ and~$L$ between~$p_i$ and~$p_j$.
    Let~$p_a,p_b \in P \setminus \{p_i,p_j\}$ be a node on~$U$ and~$L$, respectively.
    Note that such nodes exists by definition of~$i$ and~$j$.
    Since~$(G,t)$ is interval bound, $G_z$ contains two nodes~$c_u$ and~$c_v$ such that~$u < v-1$ belonging to the same interval of~$C$.
    We assume without loss of generality this interval is on the region between~$p_i$ and~$p_a$ (the other three cases are symmetric).
    We then contract~$V_z$ into a single node and all nodes in~$(C,E_C)$ into five nodes~$p_i$ (containing~$c_u$), $c_{u+1}$, $p_a$ (containing~$c_v$), $p_j$, and~$p_b$.
    Note that~$p_i$ and~$p_a$ have common neighbors~$p_j$, $c_{u+1}$, and~$z$.
    Moreover, $p_i$ is adjacent to~$p_b$ and~$t$ is adjacent to both~$p_a$ and~$p_b$.
    This forms the claimed~\sK.    
    So we assume in the following that~$G$ does not contain any shortcut links.
    
    For any~$j \in [|P|]$, let~$V_j$ be the set of all nodes in the~$j$\textsuperscript{th} interval of~$C$ and all nodes~$y$ such that~$G_y$ contains nodes from that interval.
    Let~$G_j = G[V_j]$.
    Note that if each~$G_j$ is outerplanar, then~$(G,t)$ is a ring of outerplanar graphs as~$(G,t)$ is interval bound and~$G$ does not contain any shortcut links.
    Since we assume that this is not the case, there exists some~$j^*$ such that~$G_{j^*}$ is not outerplanar.
    Let~$z \in V_{j^*}$ be node such that~$z$ does not belong to the outer face of~$G_{j^*}$.
    Let~$p_a$ and~$p_{a+1}$ be the two nodes in~$P$ in~$G_{j^*}$ and let~$C'$ be all nodes on the outer face of~$G_{j^*}$ (using the same embedding as~$G$).
    Let~$V^*_z$ be the set of nodes in the connected component of~$G_{j^*}-C'$ containing~$z$ and let~$G^*_z$ be the subgraph of~$G$ induced by~$V^*_z$ and all nodes adjacent to nodes in~$V^*_z$.
    Let~$U^*$ and~$L^*$ be the two (not necessarily disjoint) paths between~$p_a$ and~$p_{a+1}$ going along the outer face of~$G^*_z$ (excluding~$p_a$ and~$p_{a+1}$).
    Note that both contain at least one node as the link~$\{p_a,p_{a+1}\}$ would be separating otherwise.
   
    If~$G^*_z$ contains nodes from both~$U^*$ and~$L^*$, then the usual no-3-paths argument shows that the~\Kt{} is a rooted minor of~$(G,t)$.
    Otherwise, all nodes in~$G^*_z$ in~$C'$ belong to one of the two paths (where this time~$p_a$ and~$p_{a+1}$ are allowed neighbors).
    We assume without loss of generality~$U^*$.
    Note that there are also two nodes~$c_u,c_v$ with~$u < v-1$ in~$G^*_z$ as otherwise the link between these two nodes is a separating link.
    We show that in this case the \sK{} is a rooted minor of~$(G,t)$.
    We contract all nodes in~$V^*_z$ into a single node, and~$U^*$ and~$L^*$ into four nodes~$p_{a}$ (containing~$c_u$), $c_{u+1}$, $p_{a+1}$ (containing~$c_{v}$), and one node~$d$ from~$L^*$ (containing all nodes in~$L^*$).
    We also contract all nodes in~$G-t$ except for nodes in~$G^*_j$ into a single node~$w$.
    Note that since~$|P| \geq 3$ and~$(C,E_C)$ is a cycle, it holds that $\{w,t\}$, $\{w,p_a\}$, and~$\{w,p_{a+1}\}$ are all links in the resulting graph.
    Note that~$p_a$ and~$p_{a+1}$ have common neighbors~$d$,~$z$, and~$c_{u+1}$.
    Moreover~$p_a$ and~$w$ are adjacent and~$t$ is adjacent to~$w$ and~$p_{a+1}$.\footnote{The links~$\{t,p_{a+1}\}$ and~$\{p_a,p_b\}$ also exist, but they are not needed to form the \sK.}
    This forms the \sK.
    This concludes the case where~$|P| \geq 3$.
    
    The last remaining case to consider is~$|P|=2$.
    Let~$P=\{c_i,c_j\}$ and~${G''=G-\{c_i,c_j,t\}}$.
    Since~$G$ is \nice, there is no link~$\{c_i,c_j\}$ as this link would be separating.
    Let~$V_1,V_2,\ldots,V_\ell$ be the set of nodes in the connected components of~$G''$.
    Since~$(G,t)$ is not a dipole outerplanar graph, at least one graph~$G_i = G[V_i \cup \{c_i,c_j\}]$ is not outerplanar.
    Let~$G_j$ be such a graph and let~$z$ be a node in it that does not belong to the outer face (using the embedding for~$G$).    
    Let~$G'_j$ be the biconnected component of~$G_j$ containing~$z$, let~$C''$ be the set of nodes on the outer face of~$G'_j$, and let~$c'_i$ and~$c'_j$ be the two nodes in~$C''$ such that all paths from~$z$ to~$c_i$ in~$G_j$ go through~$c'_i$ and all paths from~$z$ to~$c_j$ in~$G_j$ go through~$c'_j$ (potentially~$c'_i = c_i$ and/or $c'_j=c_j$).
    See \cref{fig:biconnected} for an example.
    \begin{figure}[t]
        \centering
        \begin{tikzpicture}
            \node[te] at(0,4) (t) {};
            \node[por,label=left:$c_i$] at(-3,0) (a1) {} edge[bend left=20](t);
            \node[ver] at(-2,1) (a2) {} edge(a1);
            \node[ver] at(-2,-1) (a3) {} edge(a1);
            \node[ver,label=$c'_i$] at(-1,0) (a4) {} edge(a2) edge(a3);
            \node[ver,thick,draw=blue] at(0,1) (a5) {} edge(a4);
            \node[ver,thick,draw=blue] at(1,1) (a6) {} edge(a5);
            \node[ver,thick,draw=blue] at(2,1) (a7) {} edge(a6);
            \node[por,label=right:$c_j$] at(3,0) (a8) {} edge(a7) edge[bend right=20](t);
            \node[ver] at(2,-1) (a9) {} edge(a8);
            \node[ver] at(1,-1) (a10) {} edge(a9);
            \node[ver] at(0,-1) (a11) {} edge(a10) edge(a4);

            \node[ver] at(0,0) (y) {} edge(a9) edge(a11);
            \node[ver,label=left:$z$,fill=gray,thick,draw=blue] at(1,0) (z) {} edge(a5) edge(a6);
            \node[ver,fill=gray,thick,draw=blue] at(2,0) {} edge(a6) edge(a7) edge(z);
            
            \node[ver] at(0,2) (c) {} edge[bend right=30](a1) edge[bend left=30](a8);
            \node[ver] at(0,-2) (b1) {} edge[bend left=30](a1) edge[bend right=30](a8);
            \node[ver] at(0,-3) (b2) {} edge[bend left=20](a1) edge[bend right=20](a8) edge(b1);

            \node[ellipse,draw,minimum width=2em,minimum height=1.8em,inner ysep=0pt, label=$V_1$] at(0,2) {};
            \node[ellipse,draw,minimum width=2em,minimum height=1.6cm,inner ysep=0pt, label=below right:$V_3$] at(0,-2.5) {};
            \node[very thick,dashed,rectangle,draw,minimum width=8cm,minimum height=3cm,inner ysep=0pt,label=left:$G_j$] at(0,0) {};
            \node[very thick,dashed,rectangle,draw,minimum width=5cm,minimum height=2.8cm,inner ysep=0pt] at(1.2,0) {};
            \node at(3.3,-1) {$G'_j$};
        \end{tikzpicture}
        \caption{An example of the definition of~$G'_j$ (the graph within the small dashed rectangle). Here,~$c'_j=c_j$ holds. The nodes in~$V''_z$ are filled gray and nodes in~$G''_z$ are marked with a blue border.}
        \label{fig:biconnected}
    \end{figure}
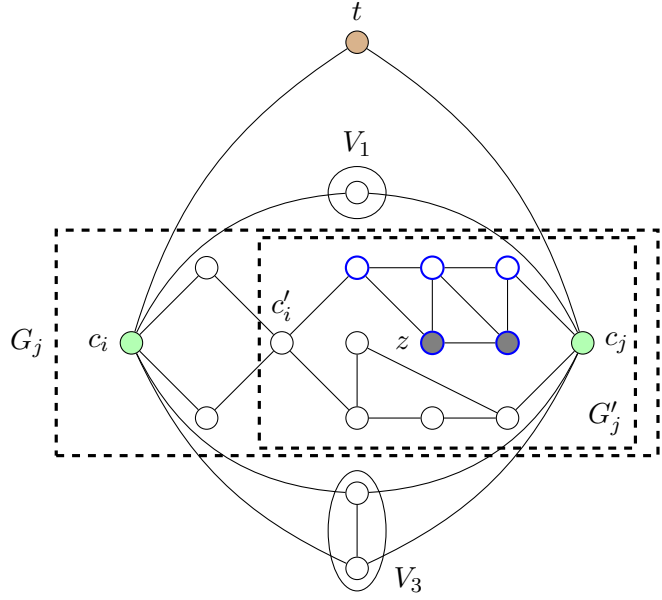%
    Let~$V''_z$ be the set of all nodes in the connected component of~$G'_j - C''$ containing~$z$ and let~$G''_z$ be the graph induces by~$V''_z$ and all nodes that are adjacent to nodes in~$V''_z$ in~$G$.
    Let~$U$ and~$L$ be the two (internally disjoint) paths between~$c'_i$ and~$c'_j$ along the outer face of~$G'_j$.
    Since~$\{c'_i,c'_j\}$ would be separating, each of~$U$ and~$L$ contains at least three nodes.
    If~$G''_z$ contains a node from~$U$ (except for~$c'_i$ or~$c'_j$) and a node from~$L$ (also excluding endpoints), then the usual no-3-paths argument shows that the \Kt{} is a rooted minor of~$(G,t)$.
    
    So assume that~$G''_z$ contains two nodes~$c'_a$ and~$c'_b$ from either of the two paths.
    Note that two such nodes with~$a < b-1$ exist as otherwise the link~$\{c'_a,c'_b\}$ between the two unique nodes in~$G''_z$ in~$C''$ would be separating.
    We assume without loss of generality that~$c'_a$ and~$c'_b$ belong to~$U$ as the other case is analogous.
    Let~$c'_d$ be a node in~$L$ other than~$c'_i$ and~$c'_j$.
    To conclude the entire proof, we show that the \sK{} is a rooted minor of~$(G,t)$ in this last case.
    We ignore any nodes that are not contained in~$G_j$ and that are not~$t$.
    We contract every node in~$G_j$ that is not contained in~$G'_j$ into either~$c_i$ or~$c_j$.
    We also contract~$V''_z$ into a single node.
    If~$c'_i \neq c_i$ or~$c'_j \neq c_j$, then we contract~$C''$ into four nodes~$c'_a$ (containing~$c'_i$), $c'_{a+1}$, $c'_{b}$ (containing~$c'_j$), and~$c'_d$.
    Note that~$c'_a$ and~$c'_b$ have common neighbors~$c'_{a+1}$, $c'_d$, and~$z$.
    Moreover, there is an additional path of length at least three from~$c'_a$ to~$c'_b$ that contains~$c_i$ or~$c_j$ (or both).
    Contracting this path into a path of length exactly three forms the claimed \sK.
    
    If~$c'_i=c_i$ and~$c'_j=c_j$, then note that there needs to be a link~$\{c'_u,c'_\ell\}$ where~$c'_u$ is a node on~$U$, $c'_\ell$ is a node on~$L$, and neither of these two nodes is~$c_i$ or~$c_j$ as otherwise~$U$ and~$L$ would only be connected through~$c_i$ and~$c_j$, contradicting the fact that~$G_j-\{c_i,c_j\}$ is connected.
    Note that~$c'_u$ is part of either the subpath from~$c'_i$ to~$c'_a$ or the subpath from~$c'_b$ to~$c'_j$ of~$U$ as any link from another node of~$U$ to a node on~$L$ would pass through some link of~$G''_z$ or some link incident to~$t$, contradicting the planarity of the chosen embedding.
    We assume without loss of generality that~$c'_u$ is part of the subpath from~$c'_i$ to~$c'_a$.
    Since~$c'_i = c_i$ and~$c'_u \neq c_i$, it holds that~$c'_u \neq c'_i$.
    We again contract all nodes in~$V''_z$ into a single node and all nodes in~$C''$ into five nodes~$c'_i$, $c'_a$ (containing~$c_u$), $c'_{a+1}$, $c'_j$ (containing~$c'_b$), and~$c'_\ell$.
    Note that~$c'_a$ and~$c'_j$ have common neighbors~$c'_{a+1}$, $z$, and~$c'_\ell$ (as~$\{c'_u,c'_\ell\} \in E$).
    Combined with the path~$(c'_a,c'_i,t,c'_j)$, this forms the \sK.
    This concludes the proof.    
\end{proof}

\section{Algorithms}
\label{sec:alg}
In this section, we show how to combine the results from all previous sections to show our second main result.
We restate the result here for convenience.

\algo*

\begin{proof}
    Let~$(G,t)$ be a rooted graph.
    We will first show how to solve \prd{} in~$O(n)$ time.
    We first check whether the connected component of~$G$ containing~$t$ is planar in~$O(n)$ time using \cref{prop:embedding}.
    If the component is not planar, then~$(G,t)$ is not perfectly resilient by \cref{prop:subgraph,prop:planar}.
    Without loss of generality, we can hence assume that~$G$ is connected to avoid introducing a new name for the connected component of~$G$ containing~$t$.
    Next, we check whether the branchwidth of~$G$ is at most~$9$ and if so, then we compute a branch decomposition of width at most~$9$ in~$O(n)$ time~\cite{BT97}.
    If the branchwidth is at most~$9$, then for each graph~$H \in \{\Kf,\Kt,\sK,\Ktf\}$, we check whether~$(H,t)$ is a rooted minor of~$(G,t)$ in~$O(n)$ time using \cref{cor:algforbw}.
    By~\cref{thm:structure}, $(G,t)$ is perfectly resilient if and only if none of the four is a rooted minor of~$(G,t)$.
    If the branchwidth is at least~$10$, then we show that~$(G,t)$ is not perfectly resilient.
    By \cref{cor:branchwidth}, $G$ contains the $4 \times 4$ grid as a minor.
    Since~$G$ is connected, if~$t$ is not already contained in the $4 \times 4$ grid after contractions, then it can be contracted into one of the 16 nodes of the~$4 \times 4$ grid, so~$(G,t)$ contains as a rooted minor the~$4 \times 4$ grid where~$t$ is any of the nodes.
    A simple case distinction on where the node~$t$ is reveals that the~\Kf{} is a rooted minor of~$(G,t)$ in each case.
    See \cref{fig:grid} for an illustration of this fact.
    \begin{figure}[t]
        \centering
        \begin{tikzpicture}
            \foreach \i in {0,1,2,3}{
                \foreach \j in {0,1,2,3}{
                    \node[ver] at(\j,\i) (v\i\j) {};
                }
            }
            \foreach \i in {0,1,2,3}{
                \foreach \j in {0,1,2}{
                    \pgfmathtruncatemacro\k{\j+1};
                    \draw (v\i\j) to (v\i\k);
                    \draw (v\j\i) to (v\k\i);
                }
            }
            
            \node[ver,label=below:$s$] at(8,1) (s) {};
            \node[te] at(8,3) (t) {};
            \node[por,label=left:$u$] at(6,0) (u) {} edge(s) edge(t);
            \node[por,label=right:$v$] at(10,0) (v) {} edge(s) edge(t) edge(u);
            \node[por,label=right:$w$] at(8,2) (w) {} edge(s) edge(t) edge(u) edge(v);
    
        \end{tikzpicture}%
        \vspace*{1cm}
        \begin{tikzpicture}
            % t in center
            \foreach \i in {0,1,2,3}{
                \foreach \j in {0,1,2,3}{
                    \node[ver] at(\j,\i) (v\i\j) {};
                }
            }
            \node[te] at(1,2) {};
            \node[label=below:$s$] at(1,0) {};
            \node[label=above right:$w$] at(1,1) {};
            \node at(0,3.5) {$u$};
            \node at(2,3.5) {$v$};
    
            \draw (-.3,-.3) -- (-.3,3.3) -- (1.3,3.3) -- (1.3,2.7) -- (0.3,2.7) -- (.3,-.3) -- cycle;
            \draw (1.7,-.3) -- (1.7,3.3) -- (2.3,3.3) -- (2.3,-.3) -- cycle;
    
            \draw[dashed] (v00) to (v10);
            \draw[dashed] (v10) to (v20);
            \draw[dashed] (v20) to (v30);
            \draw[dashed] (v30) to (v31);
            \draw[dashed] (v02) to (v12);
            \draw[dashed] (v12) to (v22);
            \draw[dashed] (v22) to (v32);
            
            \draw (v00) to (v01);
            \draw (v10) to (v11);
            \draw (v20) to (v21);
            \draw (v01) to (v02);
            \draw (v11) to (v12);
            \draw (v21) to (v22);
            \draw (v31) to (v32);
            \draw (v01) to (v11);
            \draw (v11) to (v21);
            
            %t on edge
            \foreach \i in {0,1,2,3}{
                \foreach \j in {5,6,7,8}{
                    \node[ver] at(\j,\i) (v\i\j) {};
                }
            }
            \node[te] at(6,3) {};
            \node[label=below:$s$] at(6,1) {};
            \node[label=above right:$w$] at(6,2) {};
            \node at(5,3.5) {$u$};
            \node at(7,3.5) {$v$};
    
            \draw (4.7,-.3) -- (4.7,3.3) -- (5.3,3.3) -- (5.3,.3) -- (6.3,.3) -- (6.3,-.3) -- cycle;
            \draw (6.7,-.3) -- (6.7,3.3) -- (7.3,3.3) -- (7.3,-.3) -- cycle;
    
            \draw[dashed] (v05) to (v15);
            \draw[dashed] (v15) to (v25);
            \draw[dashed] (v25) to (v35);
            \draw[dashed] (v05) to (v06);
            \draw[dashed] (v07) to (v17);
            \draw[dashed] (v17) to (v27);
            \draw[dashed] (v27) to (v37);
            
            \draw (v15) to (v16);
            \draw (v25) to (v26);
            \draw (v35) to (v36);
            \draw (v06) to (v07);
            \draw (v16) to (v17);
            \draw (v26) to (v27);
            \draw (v36) to (v37);
            \draw (v16) to (v26);
            \draw (v26) to (v36);
    
            %t on corner
            \foreach \i in {0,1,2,3}{
                \foreach \j in {10,11,12,13}{
                    \node[ver] at(\j,\i) (v\i\j) {};
                }
            }
            \node[circle, inner sep=3pt, draw,fill=brown!60!white] at(10,3) {};
            \node at(10,3.5) {$t$};
            \node[label=below:$s$] at(11,1) {};
            \node[label=above right:$w$] at(11,2) {};
            \node at(9.5,2) {$u$};
            \node at(12,3.5) {$v$};
    
            \draw (9.7,-.3) -- (9.7,2.3) -- (10.3,2.3) -- (10.3,.3) -- (11.3,.3) -- (11.3,-.3) -- cycle;
            \draw (11.7,-.3) -- (11.7,3.3) -- (12.3,3.3) -- (12.3,-.3) -- cycle;
            \draw (9.7,2.7) -- (9.7,3.3) -- (11.3,3.3) -- (11.3,2.7) -- cycle;
    
            \draw[dashed] (v010) to (v011);
            \draw[dashed] (v010) to (v110);
            \draw[dashed] (v110) to (v210);
            \draw[dashed] (v310) to (v311);
            \draw[dashed] (v012) to (v112);
            \draw[dashed] (v112) to (v212);
            \draw[dashed] (v212) to (v312);
            
            \draw (v110) to (v111);
            \draw (v210) to (v211);
            \draw (v210) to (v310);
            \draw (v011) to (v012);
            \draw (v111) to (v112);
            \draw (v211) to (v212);
            \draw (v311) to (v312);
            \draw (v111) to (v211);
            \draw (v211) to (v311);
        \end{tikzpicture}
        \caption{The $4 \times 4$ grid (top left) and the~\Kf{} (top right). The bottom row shows how the~\Kf{} is a rooted minor of the $4 \times 4$ grid for all possible placements of~$t$: In the center (left), on an edge (middle), or in the corner (right). Note that due to the symmetry of the~$4 \times 4$ grid, this covers all possible cases. The boxes indicate which nodes are contracted into a single node (along dashed links in the grid) and the regular links (which are also part of the grid) depict the links in the~\Kf.}
        \label{fig:grid}
    \end{figure}
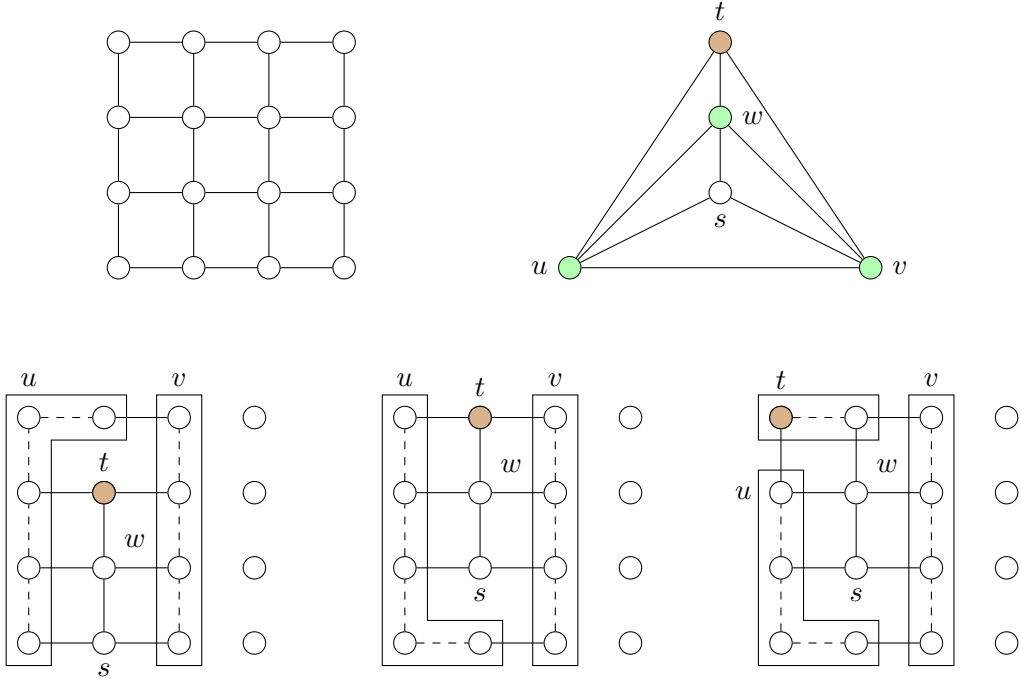%
    This concludes the proof for \prd.

    We next show how to solve \prs{} in~$O(nm)$ time using skipping forwarding patterns.
    We first use \cref{cor:pre1syn} to compute the connected component of~$G$ containing~$t$ in~$O(n)$ time.
    We will show in the following how to construct a perfectly resilient skipping forwarding pattern for it in~$O(n^2)$ time given that it is perfectly resilient.
    The corollary then gives us a perfectly resilient skipping forwarding pattern for~$(G,t)$ in~$O(nm)$ time.
    If it is not perfectly resilient, then we will detect this later.

    For the sake of readability, we will call the output of \cref{cor:pre1syn} still~$(G,t)$.
    We can now assume that~$(G,t)$ is connected and planar.
    We then use \cref{prop:pre2syn} to construct a set of some number~$k$ of rooted minors~$(G_i,t_i)$ of~$(G,t)$ whose total size is in~$O(n)$ and where each graph is biconnected and planar.
    We will show in the following how to construct a perfectly resilient skipping forwarding pattern for each rooted minor~$(G_i,t_i)$ in~$O(|G_i|^2)$ time (if they are all perfectly resilient).
    \cref{prop:pre2syn} then constructs a perfectly resilient skipping forwarding pattern for~$(G,t)$ in~$O(n^2)$ time.
    If at least one of the rooted graphs is not perfectly resilient, then we will detect this and output that the entire rooted graph is not perfectly resilient.
    This is correct by \cref{prop:pre2dec}.
    Since~$O(\sum_{i=1}^k |G_i|^2) \subseteq O((\sum_{i=1}^k |G_i|)^2) = O(n^2)$, this will conclude the proof.
    
    So it remains to show that if~$(G,t)$ is a perfectly resilient and~$G$ is biconnected and planar, then we can construct a perfectly resilient skipping forwarding pattern for it in~$O(n^2)$ time and if~$(G,t)$ is not perfectly resilient, then we need to detect this.
    We now compute the set~$S$ of separating links in~$G$ in~$O(n^2)$ time using \cref{obs:computeS}.
    If we can construct a perfectly resilient right-hand forwarding pattern and a corresponding planar embedding for~$(G \setminus S,t)$, then \cref{prop:pre3syn} constructs a perfectly resilient skipping forwarding pattern for~$(G,t)$ in~$O(n^2)$ time.
    To construct such a pattern for~$(G \setminus F,t)$, we observe that~$(G \setminus S,t)$ does not contain any separating links by \cref{lem:onego} and is therefore \nice.
    Hence, \cref{thm:structure} states that~$(G,t)$ is not perfectly resilient, ${G-t}$ is outerplanar, $(G,t)$ is a dipole outerplanar graph, or~$(G,t)$ is a ring of outerplanar graphs.

    Testing whether~$G-t$ is outerplanar and computing an outerplanar embedding if this is the case takes~$O(n)$ time by \cref{cor:outerplanar}.
    By \cref{thm:gt}, we can compute a perfectly resilient right-hand forwarding pattern for this embedding in~$O(n^2)$ time.
    If~$G-t$ is not outerplanar, then we can check whether~$(G,t)$ is a dipole planar graph in~$O(n)$ time by checking whether~$|N(t)|=2$.
    If so, then let~$u$ and~$v$ be the neighbors of~$t$.
    We compute all connected components of~$G-\{u,v,t\}$ in~$O(n)$ time and then check whether each is outerplanar in~$O(n)$ time again using \cref{cor:outerplanar}.
    By \cref{prop:dipole}, if~$(G,t)$ is a dipole outerplanar graph, then we can compute a planar embedding and a corresponding perfectly resilient right-hand forwarding pattern for it in~$O(n^2)$ time.
    Finally, we test whether~$(G,t)$ is a ring of outerplanar graphs.
    We first check that~$|N(t)| \geq 3$.
    If this is the case, then we compute~$G - (N(t) \cup \{t\})$ and all connected components in it in~$O(n)$ time.
    We can also check for each node~$v \in N(v)$, which connected components they connect to.
    Starting from an arbitrary pair that are connected to a common connected component, we can build the ring of outerplanar graphs,a planar embedding for it, and the corresponding perfectly resilient right-hand forwarding pattern in~$O(n^2)$ time using \cref{prop:ring}.
    If none of the above cases apply, then~$(G,t)$ is not perfectly resilient by \cref{thm:structure}.
    This concludes the proof.
\end{proof}

\section{Conclusion}

We charted a complete landscape of when perfect resilience can be achieved.
We also showed that
both the decision problem of whether a given instance is perfectly resilient as well as the
synthesis problem of constructing perfectly resilient rerouting tables in case such tables exist, can be solved efficiently.
Specifically, we showed that \prd{} and \prs{} can be solved in~$O(n)$ time and~$O(nm)$ time, respectively.
These are both optimal as long as skipping forwarding patterns are considered for~\prs{} (since the output for such patterns can have size~$\Theta(nm)$).
While the analysis is quite intricate, the actual algorithms are surprisingly simple.
In the case of \prs, the running time does not hide any huge constants and the algorithm does not rely on any complicated black-box results.
This makes it particularly promising for implementation and a practical evaluation on real-world networks.
In the case of \prd, we rely on existing algorithms for determining whether a graph of bounded branchwidth (at most 9) contains a rooted minor~$(H,t)$ with~$H \in \{\Kf,\,\Kt,\,\Ktf,\,\sK\}$.
It would be interesting to investigate whether the relatively simple structure of these four rooted graphs allows for simpler and faster algorithms than the general case.
This would allow us to eliminate the hidden constants in our~$O(n)$-time algorithm caused by calling the above algorithm as a subroutine.

Note that only a subset of (rooted) planar graphs are perfectly resilient. While some communication network types are indeed sparse~\cite{pignolet2017tomographic}, others are not. A very interesting avenue for future research is to study more powerful rerouting models like \textsc{Perfect Resilience Synthesis with Source Matching}, where rerouting rules can additionally depend on the source in addition to the in-port and the target, or \textsc{Perfect Resilience Synthesis with Header Rewriting}, where rerouting rules can additionally rewrite a constant number of packet header bits.
The hope is that these more powerful models allow larger classes of networks to provide perfect resilience.
We believe that our approach is versatile enough to also be used for these related problems.
In particular, all of our preprocessing steps generalize directly to \textsc{Perfect Resilience Synthesis with Source Matching}.
Additionally, it is already known that non-planar connected instances can be perfectly resilient if source matching is permitted.
The forbidden rooted minors and the classes of rooted graphs that allow for perfectly resilient forwarding patterns thus change significantly in this setting and their characterization remains a practically and theoretically well motivated major open problem.

Our work also shows that skipping priority lists are as powerful as arbitrary forwarding functions in the context of perfect resilience.
Whether this is also true in other contexts like ideal resilience remains an intriguing open problem.
Also the question whether ideal resilience can always be achieved, that is, a $k$-link-connected graph is always resilient against~$k-1$ link failures, remains a major open problem.
We note that our preprocessing step ensuring biconnectivity generalizes to ideal resilience and it can therefore be assumed that the input graph is~$2$-node-connected as well as~$k$-link-connected.
It would also be interesting to settle the complexity questions regarding ideal resilience.
Even if ideal resilience can always be achieved, a proof of this fact might be non-constructive or only imply an exponential-time algorithm to compute an ideally resilient forwarding pattern.

\section*{Acknowledgments}
The authors would like to thank Dario Giuliano Cavallaro and Jiří Srba for fruitful discussions. 
This research is supported by the German Research Foundation (DFG), SPP 2378 (Resilience in Connected Worlds: Mastering Failures, Overload, Attacks, and the Unexpected), ReNO-2: Unterstützung von Netzwerkbetreibern durch ML -- Ein kognitiver Ansatz (project number 511099228).

\bibliographystyle{alpha}
\bibliography{ref.bib}

\end{document}